\documentclass[12pt]{article}
\usepackage[english]{babel}
\usepackage{amsmath}
\usepackage{amssymb}
\usepackage{amsthm}
\usepackage{graphicx}
\usepackage{natbib}
\usepackage{mathtools}
\usepackage{accents}
\usepackage{booktabs}
\usepackage{xr}
\usepackage[ruled, vlined, nofillcomment]{algorithm2e}
\usepackage{setspace}
\usepackage{adjustbox}

\usepackage{tikz}
\usetikzlibrary{positioning, calc, shapes.geometric, shapes.multipart,
  shapes, arrows.meta, arrows,
  decorations.markings, external, trees}
\tikzstyle{Arrow} = [
thick,
decoration={
  markings,
  mark=at position 1 with {
    \arrow[thick]{latex}
  }
},
shorten >= 3pt, preaction = {decorate}
]
\newcommand*\circled[1]{\tikz[baseline=(char.base)]{
    \node[shape=circle,draw,solid,inner sep=2pt] (char) {#1};}}

\newcommand{\blind}{0}

\theoremstyle{plain}
\newtheorem{theorem}{Theorem}[section]
\newtheorem{proposition}[theorem]{Proposition}
\newtheorem{lemma}[theorem]{Lemma}

\theoremstyle{definition}
\newtheorem{definition}{Definition}[section]

\theoremstyle{remark}
\newtheorem{remark}{Remark}[section]

\newcommand{\eps}{\varepsilon}
\newcommand{\R}{\mathbb{R}}
\newcommand{\PR}{{\mathbb{P}}}
\newcommand{\E}{\mathbb{E}}
\newcommand{\indep}{\perp\kern-1.25ex\perp}
\newcommand{\Id}{\mathrm{Id}}

\DeclareMathOperator*{\minmax}{min\,/\, max}

\DeclareMathOperator*{\var}{var}
\DeclareMathOperator*{\cov}{cov}
\DeclareMathOperator*{\corr}{corr}
\DeclareMathOperator*{\sd}{sd}

\DeclareMathOperator*{\covh}{\widehat{cov}}

\DeclareMathOperator*{\Span}{span}
\DeclareMathOperator*{\pir}{PIR}

\newcommand{\rh}{\hat{R}}
\newcommand{\fh}{\hat{f}}
\newcommand{\hth}{\hat{\theta}}
\newcommand{\bh}{\hat{\beta}}

\newcommand{\xd}{\dot{X}}
\newcommand{\xt}{\tilde{X}}

\newcommand{\yc}{\mathcal{Y}}
\newcommand{\xc}{\mathcal{X}}
\newcommand{\zc}{\mathcal{Z}}
\newcommand{\wc}{\mathcal{W}}
\newcommand{\hc}{\mathcal{H}}
\newcommand{\pd}{\dot{p}}
\newcommand{\varpm}{\mathbin{\vcenter{\hbox{%
        \oalign{\hfil$\displaystyle+$\hfil\cr
          \noalign{\kern-.3ex}
          $\displaystyle[{-}]$\cr}%
      }}}}

\newlength{\dhatheight}
\newcommand{\hathat}[1]{%
  \settoheight{\dhatheight}{\ensuremath{\hat{#1}}}%
  \addtolength{\dhatheight}{-0.3ex}%
  \hat{\vphantom{\rule{1pt}{\dhatheight}}%
    \smash{\hat{#1}}}}

\newcommand\addtag[1][]{%
  \refstepcounter{equation}\hfill(\theequation)%
  \notblank{#1}{\label{#1}}{}}

\usepackage{soul}

\newcommand{\bt}[1]{{\color{black} #1}}

\makeatletter
\newcommand*{\addFileDependency}[1]{
\typeout{(#1)}
\@addtofilelist{#1}
\IfFileExists{#1}{}{\typeout{No file #1.}}
}\makeatother

\newcommand*{\myexternaldocument}[1]{%
\externaldocument{#1}%
\addFileDependency{#1.tex}%
\addFileDependency{#1.aux}%
}

\myexternaldocument{paper-supplement}

\usepackage{hyperref}%

\begin{document}

\def\spacingset#1{\renewcommand{\baselinestretch}%
{#1}\small\normalsize} \spacingset{1}


\if0\blind
{
  \title{\bf Optimization-based Sensitivity Analysis for Unmeasured Confounding using Partial Correlations}
  \author{Tobias Freidling\thanks{
    TF is supported by a PhD studentship from GlaxoSmithKline Research \& Development.}\hspace{.2cm}\thanks{New affiliation: Department of Mathematics, \'Ecole polytechnique f\'ed\'erale de Lausanne, Switzerland}\\
    {\small Statistical Laboratory, DPMMS, University of Cambridge, United Kingdom}\\
    and \\
    Qingyuan Zhao\thanks{QZ is partly supported by the EPSRC grant EP/V049968/1.} \\
    {\small Statistical Laboratory, DPMMS, University of Cambridge, United Kingdom}}
    \date{}
  \maketitle
} \fi

\if1\blind
{
  \bigskip
  \bigskip
  \bigskip
  \begin{center}
    {\Large\bf Optimization-based Sensitivity Analysis for Unmeasured Confounding using Partial Correlations}
\end{center}
  \medskip
} \fi


\bigskip
\begin{abstract}
Causal inference necessarily relies upon untestable assumptions; hence, it is crucial to assess the robustness of obtained results to violations of identification assumptions. However, such sensitivity analysis is only occasionally undertaken in practice, as many existing methods \bt{require analytically tractable solutions and their results are often difficult to interpret.}
We take a more flexible approach to sensitivity analysis and view it as a constrained stochastic optimization problem. This work focuses on sensitivity analysis for a linear causal effect when an unmeasured confounder and a potential instrument are present. We show how the bias of the OLS and TSLS estimands can be expressed in terms of partial correlations. Leveraging the algebraic rules that relate different partial correlations, practitioners can specify intuitive sensitivity models which bound the bias. We further show that the heuristic ``plug-in'' sensitivity interval may not have any confidence guarantees; instead, we propose a bootstrap approach to construct sensitivity intervals which performs well in numerical simulations. We illustrate the proposed methods with a real study on the causal effect of education on earnings and provide user-friendly visualization tools.
\end{abstract}

\noindent%
\vfill

\newpage
\spacingset{1.75} 

\section{Introduction}
In many scientific disciplines, provisional causal knowledge is
predominantly generated from observational data as randomized
controlled experiments are often infeasible or too costly. Because the
treatment is not randomly assigned in an observational study, any
causal conclusions must rely on untestable assumptions, such as
absence of unmeasured confounders or validity of instrumental
variables. Hence, the causal inference is inherently sensitive to
violations of any identification and modelling assumptions, so
reseachers are advised to investigate the robustness of their
results.

The importance of sensitivity analysis has been emphasized in guidelines for designing and reporting observational studies \citep{strobe,pcori}. Yet, it is still rarely conducted in actual studies \citep{thabane2013, de_souza2016}, despite a decade-long line of research. Examples include the seminal work of \citet{cornfield} on the effect of smoking on lung cancer, the Rosenbaum sensitivity model \citep{rosenbaum1983,rosenbaum87}, the E-value proposed by \citet{ding2016}, or the marginal sensitivity model \citep{tan_distributional_2006,zhao,dorn}. Since specifying a meaningful and coherent sensitivity model is oftentimes an intricate task, many previous works concentrate on simple scenarios where closed-form solutions are available. This, however, renders many methods inapplicable or hard to interpret in practical data analysis.

Compared with previous work, a crucial distinction of this article is that we do not require closed-form solutions \bt{for the upper or lower bound on an estimand of interest}. Instead, we cast sensitivity analysis as a constrained stochastic optimization problem. \bt{While the drivers of potential bias are not as obvious as for closed-form solutions, optimization-based sensitivity analysis allows practitioners to define more and more \emph{interpretable} constraints. Moreover, in complex models, it may not be possible to derive closed-form solutions at all.}

We conceptualize sensitivity analysis for the effect of unobserved variables in the following general framework: Consider an i.i.d.\ sample $(V_i,U_i)_{i=1}^n$ from some population, but only the variables $(V_i)_{i=1}^n$ are observed. Denote the joint probability distribution of $(V_i, U_i)$ as $\PR = \PR_{V,U}$ and the marginal distribution of $V_i$ as $\PR_V$. We are interested in estimating and conducting inference for some functional $\beta~{\!=\!}~\beta(\PR_{V,U})$, e.g.\ the causal effect of one observed variable on another. Since this functional generally depends on the joint distribution of $V_i$ and $U_i$, we require additional, \emph{untestable} assumptions on the unmeasured variables to consistently estimate $\beta$. We can assess the robustness of our findings to deviations from such assumptions in these four steps:
\begin{enumerate}
    \item \textbf{Interpretable Sensitivity Parameters} In many cases, $\beta$ can be expressed as a function of two types of parameters, ${\theta=\theta(\PR_V)}$ and $\psi = \psi(\PR_{V,U})$: $\beta(\PR_{V,U}) = \beta(\,\theta(\PR_V),\, \psi(\PR_{V,U})\,)$. The former only depends on $\PR_V$ and can therefore be estimated; the latter additionally depends on the distribution of $U$ and thus cannot be directly estimated. Finding an interpretable parameterization $\psi$ is crucial as practitioners need to express their domain knowledge in terms of plausible values of $\psi$.
    \item \textbf{Specifying Constraints} The assumptions in the primary analysis typically correspond to $\psi$ taking the value (conventionally 0 or 1) that corresponds to the unobserved variables $U$ being ``ignorable''. To investigate the sensitivity to this assumption, partially identified sensitivity analysis asks domain experts to specify a set of plausible $\psi$-values denoted by $\Psi(\theta)$. 
    This set is usually defined in terms of equality and inequality constraints, e.g.\ $\lvert \psi_1 \rvert \leq 0.5$, and may depend on the estimable parameters~$\theta$ as well. Under the condition $\psi \in \Psi(\theta)$, the functional $\beta$ is only partially identified. We call the corresponding set of $\beta$-values the partially identified region~(PIR):
    \begin{equation}\label{eq:def-pir}
      \pir(\PR_V) := \big\{\beta(\theta(\PR_V), \psi)\colon \psi \in \Psi(\theta(\PR_V)) \big\}.
    \end{equation}
    In general, the PIR may be a quite complex object; however, if $\beta$ is real-valued and one-dimensional, we can describe it by solving the following optimization problems:
    \begin{equation} \label{eq:pir-minmax}
        \minmax \beta(\theta(\PR_V), \psi),\qquad \text{subject to } \psi \in \Psi(\theta(\PR_V)).
    \end{equation}
    If the PIR is indeed an interval, the interval from the minimum to the maximum $\beta$-value equals it; otherwise, it is conservative. 
    

    \item \textbf{Solving the Optimization Problem} As both the objective and the constraints in~\eqref{eq:pir-minmax} depend on $\PR_V$, from which we have $n$ samples, this is an instance of stochastic optimization or stochastic programming \citep{shapiro2009}. A natural, plug-in estimator of the optimal values of this problem can be obtained by solving \begin{equation} \label{eq:pir-minmax-plugin}
        \minmax \beta(\hat{\theta}, \psi),\qquad \text{subject to } \psi \in \Psi(\hat{\theta}),
    \end{equation}
    where $\hat{\theta}$ is an estimator of $\theta$ based on the observed
    data.
    \item \textbf{Uncertainty Quantification}
    Constructing confidence intervals in partially identified sensitivity analysis is subtle as there exist two notions of coverage: (1) the true value of $\beta$ or (2) the PIR is covered with high probability. Under either paradigm, uncertainty quantification is a challenging task and few articles address it. We use bootstrapping to construct confidence intervals, where we replace $\hat{\theta}$ with the bootstrap distribution $\hathat{\theta}$ in \eqref{eq:pir-minmax-plugin}.
    \end{enumerate}

In this work, we operationalize the framework above to conduct sensitivity analysis for linear regression and instrumental variables models in the presence of an unmeasured confounder. We generalize the analysis of \citet{ch}, propose a tailored grid search algorithm to solve the resulting optimization problem and provide uncertainty quantification via the bootstrap. In particular, we contribute a flexible and interpretable method to the sparse literature on sensitivity analysis for instrumental variables, e.g.\ \citep{altonji2005,small,wang}.

This article is organized as follows. In Section~\ref{sec:objective}, we use partial correlations as parameters to identify the linear causal effect of interest in terms of $\theta$ and $\psi$. Section~\ref{sec:sensitivity-model} shows how the sensitivity parameters relate to common identification assumptions and provides several ways for practitioners to specify constraints. In Section~\ref{sec:algorithm}, we describe the grid search algorithm. Section~\ref{sec:inference} addresses the construction of sensitivity intervals via the bootstrap. In Section~\ref{sec:example}, we apply our proposed method to a famous study in labour economics by \citet{card} and visualize the results using $b$- and $R$-contour plots. Finally, Section \ref{sec:discussion} concludes this article with a discussion of the developed methodology. Derivations, proofs, and additional results can be found in the supplementary materials. An implementation of our proposed methods is available in the R-package \texttt{optsens} which can be downloaded from \url{https://github.com/tobias-freidling/optsens}. The code to replicate the simulation studies and data analysis in this article can be found at \url{https://github.com/tobias-freidling/optsens-replication}.

\section{Parameterization in terms of Partial Correlations}\label{sec:objective}
Our main goal in this article is to describe a unified approach to sensitivity analysis in the following common situation: We are interested in the causal effect of a one-dimensional treatment $D$ on a continuous outcome~$Y$; moreover, we may also observe some covariates~$X$ and a potential instrument $Z$. Hence, the observed variables are given by $V = (Y,D,X,Z)$. 

In a sensitivity analysis, we are worried about some unmeasured variable $U$ that confounds the causal effect of $D$ on $Y$. This can potentially be addressed by finding an instrumental variable $Z$ for the treatment~$D$, but this instrument may itself be invalid. Readers who are unfamiliar with instrumental variables are referred to Section~\ref{sec:sensitivity-parameters}; one possible graphical representation of the relations of the variables is depicted in Figure~\ref{fig:causal-graph}.

In this paper, the direct causal effect of $D$ on $Y$ -- and thus the objective function -- is defined based on the covariance matrix of $(V,U)$ which is assumed to be positive definite. We define the residual of $Y$ after partialing/regressing out $X$ by $Y^{\perp X} := Y - X^T\var(X)^{-1}\cov(X,Y)$ and use the analogous expressions for other combinations of variables. Then, the direct linear causal effect is given by
\begin{equation*}
    \beta = \beta_{Y\sim D\vert X,Z,U} := \frac{\cov(Y^{\perp X,Z,U}, D^{\perp X,Z,U})}{\var(D^{\perp X,Z,U})},
\end{equation*}
\bt{assuming that $(X,Z,U)$ does not contain descendants of $D$ and $Y$.}
If the involved variables follow a linear structural equation model (LSEM), see for example \citet{spirtes_causation_2000}, $\beta$ equals the coefficient of $D$ in the equation determining $Y$. Even if the true model involves categorical variables or is non-linear, $\beta$ can still have a causal interpretation. For instance, if $D$ is binary, it can be shown that $\beta$ equals a weighted treatment effect contrasting the groups $D=0$ and $D=1$ \citep{angrist2009mostly}.

Since $U$ is unobserved, $\beta$ cannot be estimated. Therefore, researchers usually make additional identification assumptions such as ignorability of $U$ or $Z$ being a valid instrument. If these hold, $\beta$ equals the ordinary least squares (OLS) estimand $\beta_{Y\sim D\vert X,Z}$ or the two-stage least squares (TSLS) estimand $\beta_{D\sim Z\vert X,\, Y\sim Z \vert X}$, respectively. These are defined as
\begin{equation*}
  \beta_{Y\sim D\vert X,Z} := \frac{\cov(Y^{\perp X,Z}, D^{\perp
      X,Z})}{\var(D^{\perp X,Z})} \qquad\text{and}\qquad \beta_{D\sim Z\vert X,\, Y\sim Z \vert X}
  := \frac{\cov(Y^{\perp X}, Z^{\perp
      X})}{\cov(D^{\perp X}, Z^{\perp X})}.
\end{equation*}
In order to assess the sensitivity of these identification approaches, we express the objective $\beta$ in terms of estimable parameters~$\theta$ and sensitivity parameters~$\psi$. To this end, we first summarize (partial) correlations/$R$-values and $R^2$-values.

\subsection{$R$- and $R^2$-values}\label{sec:r2-value}
Correlations/$R$-values and $R^2$-values are often introduced together with the multivariate normal distribution, see e.g.\ \citet[sec.\ 2.5]{anderson}; yet, they rely on no distributional assumptions other than positive definiteness of the covariance matrix of the involved random variables. $R$- and $R^2$-values can be defined for both random variables and samples from them. For brevity, we only state the results for the population versions\bt{, i.e.\ the estimands and not their estimators,} and assume without loss of generality that random variables are centred. More details can be found in Appendix~\ref{app:r2-linear-models}.

Let $Y$ be a random variable, let $X$ and $Z$ be two random
vectors and suppose that they all have finite variances and the covariance matrix of $(Y,X,Z)$ is positive definite. Using the notation of residuals from above, we define $\sigma^2_{Y\sim X} := \var(Y^{\perp X})$; let $\sigma_{Y\sim X}$ denote its square root.

\begin{definition}
\label{def:r2-value}
  
  The (marginal) $R^2$-value of $Y$ on $X$, the partial $R^{2}$-value $Y$ on $X$ given $Z$ and the $f^2$-value of $Y$ on $X$ given $Z$ are
  defined as
  \begin{equation*}\label{eq:r2}
    R^{2}_{Y \sim X}:= 1 - \frac{\sigma^2_{Y\sim X}}{\sigma^2_Y},\qquad
    R^2_{Y\sim X\vert Z} := \frac{R^2_{Y\sim X + Z} - R^2_{Y\sim
        Z}}{1-R^2_{Y \sim Z}},
        \qquad
        f^{2}_{Y \sim X\vert Z} := \frac{R^2_{Y\sim
        X\vert Z}}{1-R^2_{Y\sim X\vert Z}},
  \end{equation*}
  respectively.  
  If $X$ is one-dimensional, 
  the partial $R$- and $f$-value \citep{cohen} are given by 
  \begin{equation*}
  R_{Y\sim X \vert Z} :=\corr(Y^{\perp Z}, X^{\perp Z}),\qquad
  f_{Y\sim X \vert Z} := \frac{R_{Y\sim X\vert
    Z}}{\sqrt{1-R^{2}_{Y\sim X \vert Z}}}.
\end{equation*}
The marginal $f^2$-, $R$- and $f$-values can be further defined by
using an ``empty'' $Z$ in the definitions above; details are omitted.
\end{definition}


The partial $R^2$ takes values in $[0,1]$ and is a measure of how well the variables in $X$ can be linearly combined to explain the variation in $Y$ after already using linear combinations of~$Z$. Values close to 1 indicate high explanatory capability. This simple interpretation makes the $R^{2}$-value a popular tool to assess the goodness of fit of a linear model. The partial $f^{2}$ is a monotone transformation of the partial $R^{2}$ and takes values in $[0,\infty]$. As the square of the partial correlation of $Y$ and $X$ given $Z$ indeed equals the corresponding partial $R^2$-value, we refer to correlations also as $R$-values in the following. They not only capture the strength of association but also the direction of dependence.
Different (partial) $R$- and $R^2$-values are related by a set of algebraic rules which we named the $R^2$-calculus. The full set of relationships along with related results can be found in Appendix~\ref{app:r2-calculus}.

Due to their straightforward interpretation, $R^2$- and partial $R^2$-values are widely used to help practitioners understand the results of sensitivity analyses, e.g.\ \citet{imbens_2003, veitch2020}. The use of partial $R$-values proves particularly useful in parameterizing the bias of ordinary least squares (OLS) \bt{estimands}. This is elaborated next.

\subsection{Parameterization of the Objective}\label{sec:parameterization}
In previous works, \citet{frank}, \citet{hosman} and \citet{oster} conduct sensitivity analysis with parameterizations of the objective that are partly based on $R$- or $R^2$-values. \citet{ch} take this idea further and show that the absolute value of the bias of the regression \bt{estimand} can be expressed solely using partial $R^2$-values and standard deviations. This parameterization can be refined by using partial $R$-values which naturally capture the direction of confounding and expanded to the TSLS estimand.

\begin{proposition}\label{lem:objective}
    \begin{align}
        \beta = \beta_{Y\sim D \vert X,Z,U} &= \beta_{Y\sim D\vert X,Z} - R_{Y\sim U\vert X,Z,D}\,f_{D\sim U\vert X,Z} \frac{\sigma_{Y\sim X+Z+D}}{\sigma_{D\sim X+Z}}, \label{eq:objective-identification}\\[1ex]
        &= \beta_{D\sim Z\vert X, Y\sim Z \vert X} - \left[\frac{f_{Y\sim Z\vert X,D}}{R_{D\sim Z\vert X}} + R_{Y\sim U \vert X, Z, D}\,f_{D\sim U \vert X,Z} \right] \frac{\sigma_{Y\sim X+Z+D}}{\sigma_{D\sim X+Z}}.\label{eq:objective-iv}
    \end{align}
\end{proposition}

\begin{remark}
    Similar formulas can be obtained for the family of $k$-class estimands \citep{theil,nagar} and anchor regression estimands \citep{anchor} as they interpolate between the OLS and TSLS estimands.
\end{remark}
In equation \eqref{eq:objective-identification}, all quantities but $R_{Y\sim U\vert X,Z,D}$ and $R_{D\sim U\vert X,Z}$ can be estimated from the observed data. Hence, we take these partial correlations as sensitivity parameters $\psi$. How we may use them to assess deviations from the identification assumptions of ignorability and/or $Z$ being a valid instrument, is elaborated in Section~\ref{sec:sensitivity-parameters}. First, we show that Proposition~\ref{lem:objective} can also be used to conduct sensitivity analysis for multiple unmeasured confounders by reinterpreting $U$.

\subsection{Multiple Unmeasured Confounders}\label{sec:multiple-confounders}
The assumption that the unmeasured confounder $U$ is one-dimensional has kept the algebra tractable thus far. In order to obtain a bias formula for multiple confounders, a generalization of Proposition~\ref{lem:objective} and more sensitivity parameters are required. Such extensions are explored in the Appendix~\ref{app:multiple-confounders}.


Alternatively, we may view Proposition~\ref{lem:objective} as providing an upper bound on $|\beta_{Y\sim D\vert X,Z} - \beta_{Y\sim D\vert X,Z,U}|$ that can be immediately generalized to multi-dimensional $U$ as stated in the next result. Heuristically, this is because the confounding effects of several unmeasured variables can partly cancel each other; see \citet[sec.\ 4.5]{ch}. To our knowledge, this result is first obtained by \citet{hosman}.

\begin{proposition}\label{lem:objective-multi}
Let $U$ be a random vector. Then,
  \begin{equation*}\label{eq:bias-bound}
    \left|\beta_{Y\sim D\vert X,Z} - \beta \right|
    \leq \sqrt{R^2_{Y\sim U \vert X,Z,D}\,f^2_{D \sim U \vert X,Z}
      \frac{\sigma^2_{Y\sim X+Z+D}}{\sigma^2_{D\sim X+Z}}}.
  \end{equation*}
\end{proposition}
\bt{
If the true unmeasured confounder $U$ is multi-dimensional, there is a one-dimensional linear combination of its components $U^* = \lambda^T U$ such that $\beta_{Y\sim D\vert X,Z,U} = \beta_{Y\sim D\vert X,Z,U^*}$, cf.\ proof of Proposition~\ref{lem:objective-multi} in Appendix~\ref{app:multiple-confounders}. Hence, only sensitivity analysis for $U^*$ would be required. However, since we do not know $\lambda$, we recommend to use the partial $R^2$-values involving~$U$ instead of $U^*$ to define the sensitivity model. This generally yields conservative results as $U$ has more explanatory capability than $U^*$.
}


\section{Sensitivity Model}\label{sec:sensitivity-model}
Proposition \ref{lem:objective} in the previous section has established the dependence of the objective~$\beta$ on two sensitivity parameters: $R_{D\sim U\vert X,Z}$ and $R_{Y\sim U\vert X,Z,D}$. In this section, we use this parameterization to specify sensitivity models that can assess deviations from the typical identification assumptions of ignorability and existence of a valid instrument. \bt{In Figure~\ref{fig:causal-graph}, we give \emph{one} causal graphical model that illustrates these assumptions; other possible graphs that accord to the required conditional independences may be verified by the familiar d-separation criterion \citep{pearl_2009}.

Further, we denote $[p]:=\{1,\ldots,p\}$. For $I \subseteq [p]$, we define $X_I := (X_i)_{i\in I}$ and $I^c := [p]\setminus I$. Finally, let $X_{-j} := X_{\{j\}^c}$ for any $j \in [p]$. 
}




\bt{
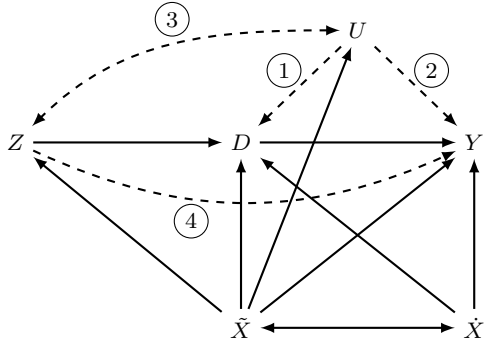
\begin{figure}[t] \centering
  \begin{tikzpicture}
  \node (z) {$Z$};
  \node (d) [right=2.5cm of z] {$D$};
  \node (u) [above right=1.5cm of d] {$U$};
  \node (y) [below right=1.5cm of u] {$Y$};
  \node (x) [below = 1.5cm of d] {$X$};

  \draw[Arrow, thick] (z) -- (d);
  \draw[Arrow, thick] (d) -- (y);
  \draw[Arrow, thick, dashed] (u) -- node[pos=0.3,left=0.08cm] {\circled{1}} (d);
  \draw[Arrow, thick, dashed] (u) -- node[pos=0.3,right=0.08cm] {\circled{2}} (y);
  \draw[Arrow, thick] (x) -- (d);
  \draw[Arrow, thick] (x) -- (y);
  \draw[Arrow, thick] (x) -- (z);

  \draw[latex-latex, thick, dashed, out=180, in=45] (u) to node[above]
  {\circled{3}} (z);
  \draw[Arrow, thick, dashed, out=-20, in = -160] (z) to
  node[pos=0.37, below=-0.05cm] {\circled{4}} (y);

  \draw[Arrow, thick] (x) -- (u);
\end{tikzpicture}
\caption{Causal diagram for regression and instrumental
  variables. Directed edges represent causal effects and bidirected
  edges represent dependence due to unmeasured common causes.}
\label{fig:causal-graph}
\end{figure}
}


\subsection{Sensitivity Parameters}\label{sec:sensitivity-parameters}
The most common identification assumption that allows us to consistently estimate the causal effect of $D$ on $Y$ from observed data, is ``ignorability'' of the unmeasured confounder. Since $\beta$ is defined in terms of the covariance matrix of $(V,U)$, ignorability of $U$ is satisfied if either
\begin{equation}\label{eq:ignorability}
    R_{D\sim U\vert X,Z}=0\qquad \text{or}\qquad R_{Y\sim U\vert X,Z,D} = 0
\end{equation}
holds true. From equation \eqref{eq:objective-identification}, we can directly read off that \eqref{eq:ignorability} suffices to identify $\beta$ via linear regression: $\beta = \beta_{Y\sim D\vert X,Z}$. This makes $R_{D\sim U\vert X,Z}$ and $R_{D\sim U\vert X,Z}$ intuitive sensitivity parameters. More commonly, ignorability is understood as $U$ not being a causal parent of~$D$ and $Y$; in Figure \ref{fig:causal-graph}, this corresponds to the absence of the edges 1 and 2, respectively. In fact, this definition of ignorability implies \eqref{eq:ignorability} if $(V,U)$ adheres to a LSEM. 



Since the violation of ignorability is a common problem -- especially in observational studies --, the method of instrumental variables (IV) is often used to overcome unmeasured confounding. Here, we only provide a very brief introduction to it; the reader is referred to \citet{wooldridge} for a more comprehensive discussion. In our setting, a variable $Z$ is called an instrument for $D$ if (i) it is an independent predictor of $D$, (ii) it is exogenous in the sense that $Z$ is partially uncorrelated with $U$ and (iii) it has no effect on the outcome $Y$ that is not mediated by $D$. These conditions can be expressed as
\begin{equation*}
    \text{(i) } R_{Z\sim D\vert X} \neq 0,\qquad
    \text{(ii) } R_{Z\sim U\vert X} = 0,\qquad
    \text{(iii) } R_{Y\sim Z\vert X,U,D} = 0.
\end{equation*}
While the first condition can be verified with observable data, the second and third cannot and therefore require sensitivity analysis. To this end, we use $R_{Z\sim U\vert X}$ and $R_{Y\sim Z\vert X,U,D}$ as sensitivity parameters for the IV assumptions. If $(V,U)$ follows an LSEM, Figure \ref{fig:causal-graph} graphically depicts the IV assumptions: (i) corresponds to the presence of the directed arrow from $Z$ to $D$; (ii) and (iii) hold if the arrows 3 and 4, respectively, are absent.

\bt{Our approach to sensitivity analysis for instrumental variables differs from \citet{cinelli_omitted_2025} in two notable ways. First, they assume that the effect of the instrument~$Z$ on the outcome~$Y$ is mediated by~$U$. Hence, $R_{Y\sim Z\vert X,U,D}$ is equal to $0$ in their set-up. Second, they use the inversion of Anderson-Rubin tests to define an estimator for the effect of $D$ on $Y$ and formulate their sensitivity model for the regression coefficient of~$Z$ in the reduced form regression $Y = \lambda Z + X\beta + U \gamma + \varepsilon_Y$. In analogy with Proposition~\ref{lem:objective}, the sensitivity parameters in this regression problem are $R_{Z\sim U\vert X}$ and $R_{Y\sim U\vert X,Z}$.}


As the parameterization \eqref{eq:objective-identification} does not involve the IV sensitivity parameters, it is not immediately clear how constraints on $R_{Z\sim U\vert X}$ and $R_{Y\sim Z\vert X,U,D}$ affect $\beta$. We can leverage the $R^2$-calculus to elucidate the relationship between the regression and IV sensitivity parameters. Deferring the derivation to Appendix~\ref{sec:reg-iv-relationship}, we obtain the equations
\begin{equation}\label{eq:reg-iv-connection}
    \begin{aligned}
        f_{Y\sim Z \vert X, U, D}\sqrt{1-R^{2}_{Y\sim U \vert X,Z,D}}
  &= f_{Y\sim Z\vert X,D}\sqrt{1-R^{2}_{Z\sim U \vert X,D}} - R_{Y\sim
    U\vert X, Z, D} R_{Z\sim U\vert X,D},\\
    f_{Z\sim U \vert X,D} \sqrt{1-R^{2}_{D\sim U\vert X,Z}}
  &= f_{Z \sim U\vert X} \sqrt{1-R^{2}_{D\sim Z\vert X}} - R_{D\sim Z\vert X} R_{D\sim U\vert X, Z}.
    \end{aligned}
\end{equation}
\begin{remark}
    Note that we can solve the second equation for $R_{Z\sim U\vert X,D}$ and plug the solution into the first. Hence, \eqref{eq:reg-iv-connection} can be understood as one equation.
\end{remark}
The equations above show that bounding the IV sensitivity parameters may constrain the values of the regression sensitivity parameters as well. Typically, an additional constraint on either $R_{Y\sim U\vert X,Z,D}$ or $R_{D\sim U\vert X,Z}$ is required to ensure that the PIR of the objective has finite length. However, if $Z$ is a valid instrument, \eqref{eq:reg-iv-connection} yields point identification: $\beta = \beta_{Y\sim Z\vert X, D\sim Z\vert X}$. This is elaborated in Appendix~\ref{sec:reg-iv-relationship}.


\subsection{Specifying Constraints}
\begin{table}
  \caption{\label{tab:constraints} Specification of the sensitivity model. Practitioners can specify any combination of senstivity bounds in the second column; the third column states the ensuing optimization constraints. (SB stands for the sensitivity bound in the respective row.)
  }\centering
  \vspace{0.3cm}
  \begin{tabular}{ccll}
    \toprule
    Edge & & Sensitivity bound & Optimization constraints \\
    \midrule
    \circled{1}  $U \rightarrow D$ & (1) & $R_{D\sim U\vert X,Z} \in [B^l_{U\!D}, B^u_{U\!D}]$ & SB\\
         & (2) & $R^2_{D \sim U \vert \xt, \xd_I,Z} \leq b_{U\! D} R^2_{D \sim \xd_J\vert \xt, \xd_I,Z}$ & \eqref{eq:reg-cond-d2}\\
    \circled{2} $U \rightarrow Y$ & (1) & $R_{Y \sim U\vert X,Z,D} \in [B^l_{UY}, B^u_{UY}]$ & SB\\
        & (2) & $R^2_{Y \sim U\vert \xt, \xd_I,Z} \leq b_{UY} R^2_{Y \sim \xd_J\vert \xt, \xd_I,Z}$ & \eqref{eq:reg-con-y22}, \eqref{eq:reg-con-y21} \\
        & (3) & $R^2_{Y \sim U\vert \xt, \xd_I,Z,D} \leq b_{UY} R^2_{Y
               \sim \xd_J\vert \xt, \xd_I,Z,D}$ & \eqref{eq:reg-con-y22},
                                                  \eqref{eq:reg-con-y32},\,\,SB \\
    \midrule
     \circled{3} $U \leftrightarrow Z$ & (1) & $R_{Z \sim U\vert X} \in [B^l_{U\!Z},B^u_{U\!Z}]$ & \eqref{eq:reg-iv-connection},\,\,SB\\
         & (2) & $R^2_{Z \sim U\vert \xt,\xd_{-j}}
         \leq b_{U\! Z} R^2_{Z \sim \xd_j\vert \xt, \xd_{-j}}$ & \eqref{eq:reg-iv-connection}, (\ref{eq:iv-con-z2})\\
     \circled{4} $Z \rightarrow Y$ & (1) & $R_{Y\sim Z\vert X,U,D} \in [B^l_{ZY}, B^u_{ZY}]$ & \eqref{eq:reg-iv-connection},\,\,SB\\
         & (2) & $R^2_{Y\sim Z\vert X,U,D} \leq b_{ZY} R^2_{Y\sim
                \xd_j\vert \xt, \xd_{-j},Z,U,D}$ & \eqref{eq:reg-iv-connection},
                                                   \eqref{eq:iv-con-ex22},\,\,SB\\
    \bottomrule
    \end{tabular}
\end{table}

Having determined the sensitivity parameters, we now focus on putting interpretable constraints on them. Table \ref{tab:constraints} summarizes all bounds that are used in this paper: The constraints on $U\rightarrow D$ and $U\rightarrow Y$ were proposed by \citet{ch}; the IV constraints $U\leftrightarrow Z$ and $Z \rightarrow Y$ are novel. 
Unlike previous works, the optimization approach to sensitivity analysis allows practitioners to use any combination of constraints as we do not need to derive an analytical expression of the optimal values for every single combination. In particular, it is possible to simultaneously conduct sensitivity analysis for the identification assumption of regression and IV.

The provided bounds can be categorized in two groups: direct and comparative constraints. In the first case, we directly bound the sensitivity parameter in question. For instance, the constraint $R_{D\sim U\vert X,Z} \in [B^{l}_{U\!D}, B^u_{U\!D}]$,
where $-1 < B^{l}_{U\!D} < B^u_{U\!D} < 1$, means that the correlation between $D$ and $U$, after accounting for linear effects of $X$ and $Z$, lies within the interval $[B^{l}_{U\!D}, B^u_{U\!D}]$.

\bt{In the second case, we compare the confounding influence of $U$ to that of an observed covariate (or group of covariates). To this end, we assume that the random vector $X \in \R^{p}$ can be partitioned into $\xd\in
\R^{\dot{p}}$ and $\xt \in \R^{\tilde{p}}$ such that $X = (\xd,\xt)$ and $R^2_{U \sim \xd \vert \xt, Z} =0$. Furthermore, we restrict ourselves to comparisons with $\xd$-covariates. This facilitates both the derivation of the sensitivity model and the interpretation of constraints. For example, the bound
\begin{equation}\label{eq:bound-ex1}
    R^2_{Y\sim U\vert \xt, \xd_{-j},Z} \leq b_{UY} R^2_{Y\sim \xd_j\vert \xt, \xd_{-j},Z},
\end{equation}
where $j \in [\pd], b_{U\! Y} \geq 0$, expresses the belief that $U$ can explain at most $b_{U\! Y}$  as much variance in $Y$ as $\xd_j$ can, after accounting for linear effects of $(\xt,\xd_{-j},Z)$. Since $U$ and $\xd_j$ are partially uncorrelated, the explanatory capability of $U$ is not increased or decreased by a ``spurious correlation'' with $\xd_j$ and vice versa. Hence, the predictive strengths of~$U$ and~$\xd_j$ are compared on a level playing field.
} Since the inequality~\eqref{eq:bound-ex1} does not directly constrain the sensitivity parameters of interest, we need to use the algebraic rules of the $R^2$-calculus to connect $R_{Y\sim U\vert X,Z,D}$ and $R_{D\sim U\vert X,Z}$ to the user-specified constraint. Indeed, under assumption $R^2_{U \sim \xd \vert \xt, Z} =0$, \eqref{eq:bound-ex1} is equivalent to a constraint on $R_{Y\sim U\vert X,Z}$ which in turn is related to $R_{Y\sim U\vert X,Z,D}$ and $R_{D\sim U\vert X,Z}$:
\begin{equation}\label{eq:bound-ex2}
  R^2_{Y\sim U \vert X,Z} \leq b_{UY}\frac{R^2_{Y\sim \xd_j\vert \xt, \xd_{-j},Z}}{1-R^2_{Y\sim \xd_j\vert \xt,\xd_{-j},Z}}, \qquad
  R_{Y\sim U \vert X,Z,D} = \frac{R_{Y\sim U \vert X,Z} - R_{Y\sim D\vert X,Z}\, R_{D\sim U \vert X,Z}}{\sqrt{1-R^{2}_{Y\sim D\vert X,Z}}\sqrt{1-R^{2}_{D \sim U \vert X,Z}}}.
\end{equation}
Hence, the user-specified constraint acts on both $R_{Y\sim U\vert X,Z,D}$ and $R_{D\sim U\vert X,Z}$.

To formulate the sensitivity model or rather the constraints of the optimization problem, we translate the interpretable, user-specified bounds into constraints on the sensitivity parameters as demonstrated in the example above. The last column of Table~\ref{tab:constraints} lists the associated optimization constraints for every user-specified bound. These inequalities and equalities and their derivations are deffered to Appendix~\ref{app:sens-model}.

\section{Adapted Grid-search Algorithm}\label{sec:algorithm}
Since users can specify any number and kind of bounds on the
sensitivity parameters, the resulting constraint set
$\Psi(\hat{\theta})$ is potentially very complex. It may be non-convex 
and can contain multiple non-linear equality- and inequality
constraints. For instance, when IV-related bounds are specified, the non-linear equation~\eqref{eq:reg-iv-connection} is always part of the optimization constraints. This only leaves few standard optimization algorithms to
compute a global solution of \eqref{eq:pir-minmax-plugin}. Since they often require careful choice of hyper-parameters and sometimes fail to solve the problem, we propose an adapted grid-search
algorithm that is more robust and tailored to our specific
optimization problem. First, we
characterize the set of potential minimizers and maximizers; then, we
explain how we can use monotonicity of equality constraints to reduce
the number of dimensions of the grid search algorithm.

\subsection{Characterization of the Solution}
According to Proposition~\ref{lem:objective}, the objective $\beta$ is
identified in terms of the sensitivity parameters $(\psi_1, \psi_2) =
(R_{D\sim U\vert X,Z}, R_{Y\sim U\vert X,Z,D})$. Due to its
monotonicity in $\psi_2$, the objective~$\beta$ attains its optimal
values on a subset of the boundary of $\Psi(\hat{\theta})$. In order
to show this, we may express the feasible set as
$\Psi(\hat{\theta}) = \bigcup_{\psi_1 \colon P_{\psi_1}\neq \emptyset} \{\psi_1\} \times P_{\psi_1}$, where $P_{\psi_1} := \{\psi_2 \colon (\psi_1, \psi_2) \in \Psi(\hat{\theta})\}$.
For every fixed $\psi_1$ such that $P_{\psi_1} \neq \emptyset$, the objective $\beta$ is a linear function in $\psi_2$. Hence, for any $\psi_2 \in P_{\psi_1}$, we obtain
\begin{equation*}
    \beta(\hat{\theta}, \psi_1, \min P_{\psi_1}) \,\lesseqgtr\, \beta(\hat{\theta}, \psi_1, \psi_2) \,\lesseqgtr\, \beta(\hat{\theta}, \psi_1, \max P_{\psi_1}),
\end{equation*}
where the direction of the inequalities depends on the sign of $\psi_1$. Therefore, $\psi$-values that minimize/maximize $\beta$ are contained in
\begin{equation*}
    \Psi^*(\hat{\theta}) := \bigcup_{\psi_1 \colon P_{\psi_1}\neq \emptyset} \{\psi_1\} \times  \{\min P_{\psi_1}, \max P_{\psi_1}\},
\end{equation*}
which is a subset of the boundary of $\Psi(\hat{\theta})$. Therefore, it suffices to discretize the \bt{one-dimensional} set $\Psi^*(\hat{\theta})$
instead of \bt{the two-dimensional} $\Psi(\hat{\theta})$ to find an approximate solution to the optimization problem.

\subsection{Transfering Bounds via Monotonicity}
Regular grid search algorithms are highly computationally expensive as their complexity grows exponentially in the number of unknown parameters. However, for our specific optimization problem, the computational costs can be significantly reduced by leveraging the monotonicity of many equality constraints. 

For instance, practitioners may specify the sensitivity bounds $R_{D\sim U\vert X,Z} \in [-0.5, 0.5]$ and $R^2_{Y\sim U\vert \xt, \xd_{-j},Z} \leq 2 R^2_{Y\sim \xd_j\vert \xt, \xd_{-j},Z}$,
which yield the optimization constraints, cf. Table~\ref{tab:constraints}:
\begin{gather*}
    R_{D\sim U\vert X,Z} \in [-0.5, 0.5],\qquad
    R^2_{Y\sim U\vert X,Z} \leq 2\, f^2_{Y\sim \xd_j\vert \xt, \xd_{-j},Z},\\[1ex]
    R_{Y\sim U\vert X,Z,D} = \frac{R_{Y\sim U\vert X,Z} - R_{Y\sim D \vert X,Z}\, R_{D\sim U\vert X,Z}}{\sqrt{1-R^2_{Y\sim D \vert X,Z}} \sqrt{1-R^2_{D\sim U\vert X,Z}}}.
\end{gather*}
A brute-force implementation of grid-search would create a three-dimensional grid of points -- one dimension per $R_{D\sim U\vert X,Z}, R_{Y\sim U\vert X,Z,D}$ and $R_{Y\sim U\vert X,Z}$ -- and only keep points that (approximately) fulfill the optimization constraints. (Partial $R$- and $f$-values that only depend on observed variables are estimated.) 
Instead, we can leverage the fact that, for any fixed $R_{D\sim U\vert X,Z}$-value, $R_{Y\sim U\vert X,Z,D}$ is a monotone function of $R_{Y\sim U\vert X,Z}$. Therefore, we only need to create a one-dimensional grid of $R_{D\sim U\vert X,Z}$ values, i.e.\ discretize $[-0.5, 0.5]$: For every such value, we can directly compute the smallest/largest value of $R_{Y\sim U\vert X,Z,D}$ within $\Psi^*(\hat{\theta})$ by plugging $R_{Y\sim U\vert X,Z} = -\sqrt{2}\,\lvert \fh_{Y\sim \xd_j\vert \xt, \xd_{-j},Z} \rvert$ into the equality constraint. Thus, we can reduce the computational complexity from cubic in the number of points per grid dimension to linear. The lower and upper end of the PIR are then simply estimated by evaluating $\beta$ over the discretization of $\Psi^*(\hat{\theta})$ and taking the smallest and largest value, respectively.

\begin{algorithm}
\setstretch{1.05}
    \SetKwComment{Comment}{/* }{ */}
    \SetSideCommentLeft
    \KwIn{
    bounds on $U \rightarrow D$: $B_{1,j}^l, B_{1,j}^u, B_{1,j}^l(\hth),B_{1,j}^u(\hth)$,\\ \phantom{Inputttt} bounds on $U\rightarrow Y$: $B_{2,j}^l, B_{2,j}^u, B_{3,j}^l(\hth),B_{3,j}^u(\hth), B_{4,j}^l(\hth),B_{4,j}^u(\hth)$,\\
    \phantom{Inputttt} number of grid points $N$}
        \caption{Grid approximation of $\Psi^*(\hth)$ for OLS bounds}
        \label{alg:grid-simple}
        \DontPrintSemicolon
        $\psi_{1}^l \gets \max_j\{B_{1,j}^l, B_{1,j}^l(\hth)\}$ \Comment*[r]{U -> D (1),(2)}
        $\psi_{1}^u \gets \min_j\{B_{1,j}^u, B_{1,j}^u(\hth)\}$\; 
        \For{$i \in \{1,\ldots,N\}$}{

            $\psi_{1,i} \gets \psi_{1}^l + (i-1)\, (\psi_{1}^u-\psi_{1}^l)/(N-1)$\;

            $\psi_{3,i}^l \gets \max_j\{B_{3,j}^l(\hth)\} \vee \max_{j}\{g_3(B_{4,j}^l(\hth), \psi_{1,i}, \hth)\}$ \Comment*[r]{U -> Y (2),(3)}
            $\psi_{3,i}^u \gets \min_j\{B_{3,j}^u(\hth)\} \wedge \min_{j}\{g_3(B_{4,j}^u(\hth), \psi_{1,i}, \hth)\}$\;

            $\psi_{2,i}^l \gets \max_j\{B_{2,j}^l\} \vee g_2(\psi_{3,i}^l, \psi_{1,i}, \hth)$ \Comment*[r]{U -> Y (1)}
            $\psi_{2,i}^u \gets \min_j\{B_{2,j}^u\} \wedge g_2(\psi_{3,i}^u, \psi_{1,i}, \hth)$\;
        }
        \KwRet{$\cup_{i=1}^N \{\psi_{1,i}\} \times \{\psi_{2,i}^l, \psi_{2,i}^u\}$}
\end{algorithm}
The principle of using monotonicity of the equality constraints to reduce the dimension of the grid applies beyond the above example. In Algorithm~\ref{alg:grid-simple}, we sketch the pseudocode for the grid approximation when only bounds on $U\rightarrow D$ and $U\rightarrow Y$ are specified. We use the notation $\psi_1 = R_{D\sim U\vert X,Z}, \psi_2 = R_{Y\sim U\vert X,Z,D}, \psi_3 = R_{Y\sim U\vert X,Z}, \psi_4 = R_{Y\sim U\vert \xt \xd_I, Z, D},$
where $\psi_4$ may be a vector with entries corresponding to different choices of $I$. The $j$-th direct lower and upper bounds on $\psi_k$ are denoted by $B_{k,j}^l$ and $B_{k,j}^u$, respectively. If the bounds arise from comparative user-specified constraints, we add the argument $\hth$. The functions $g_2$ and $g_3$ denote the equality constraints~\eqref{eq:reg-con-y22} and~\eqref{eq:reg-con-y32}, 
and, importantly, are monotonically increasing in $\psi_3$ and $\psi_4$, respectively. This algorithm has linear complexity. More details and the pseudocode for the general case with OLS- and IV-related bounds can be found in Appendix~\ref{app:algorithm}. This more general algorithm has cubic worst-case complexity; evaluating the objective over the discretized set is always linear.

\section{Sensitivity Intervals}\label{sec:inference}

By replacing the estimable (partial) $R$- and $R^2$-values and standard deviations with empirical estimators in the objective and constraints, we obtain an estimator of the PIR by solving the optimization problem \eqref{eq:pir-minmax-plugin}. Assessing the uncertainty of this estimator is a subtle task as even the notion of confidence is not immediately clear in partially identified problems. For a given $0 < \alpha < 1$, we call $S$ a $(1-\alpha)$-sensitivity interval of $\beta$ if
\[
  \PR_V\big(\beta(\theta(\PR_V), \psi) \in S\big) \geq
  1-\alpha \quad \text{for all} \quad \PR_V~\text{and}~\psi \in
  \Psi(\theta(\PR_V)),
\]
and $S$ a $(1-\alpha)$-sensitivity interval of the partially identified region if
\[
  \PR_V\big(\pir(\PR_V) \subseteq S\big) \geq 1-\alpha  \quad
  \text{for all} \quad \PR_V.
\]
Obviously, the second notion of confidence is stronger. For a more detailed discussion on confidence statements in partially identified problems including issues with asymptotic sensitivity intervals, the reader is referred to \citet{imbens}, \citet{stoye} and \citet{molinari2020}.

The construction of sensitivity intervals is an intricate problem. For instance, \citet{ch} derive a formula for the change in variance of the OLS estimator when $U$ could be included in the regression. Combining this result with the bias formula, the confidence interval can be expressed as
\begin{equation*}\label{eq:ch-confint}
  \left[\bh_{Y\sim D\vert X,Z} +\Bigg(\! - \rh_{Y\sim U\vert X,Z,D}\, \fh_{D\sim U\vert X,Z}  \pm \frac{q_{1-\alpha/2}}{\sqrt{n}}\sqrt{\frac{1-\rh^2_{Y\sim U\vert X,Z,D}}{1-\rh^2_{D\sim U\vert X,Z}}}\,
  \Bigg)
  \frac{\hat{\sigma}_{Y\sim X+Z+D}}{\hat{\sigma}_{D\sim X+Z}}\right],
\end{equation*}
where $q_{1-\alpha/2}$ is the $(1-\alpha/2)$-quantile of the standard normal distribution. Since $\hat{\psi}=(\rh_{Y\sim U\vert X,Z,D}, \rh_{D\sim U\vert X,Z})$ cannot be estimated from the observed data, it is a seemingly reasonable idea to minimize/maximize the confidence bounds over $\hat{\psi} \in \Psi(\hat{\theta})$.

Yet, a closer look at this heuristic approach shows that it achieves no obvious confidence guarantees in the frequentist sense. This is because the sensitivity parameter~$\hat{\psi}$ depends on the data and thus its value changes when another sample is drawn. Hence, correct coverage can only be guaranteed if $\hat{\psi}$ was almost certainly contained in~$\Psi(\hat{\theta})$, i.e.\ $\PR(\hat{\psi} \in \Psi(\hat{\theta})) = 1$.
Numerical simulations -- see Appendix~\ref{sec:simulation-study} -- confirm this intuitive argument; in particular, the heuristic interval has coverage 50\% in one setting and above $99\%$ in another, where the nominal coverage is $1 - \alpha = 90\%$ in both.

In this article, we use the bootstrap \citep{efron1994introduction} to construct sensitivity intervals: We compute a collection of estimators $\hathat{\theta}$ using resamples of the observable data, solve the plug-in optimization problem \eqref{eq:pir-minmax-plugin} with $\hat{\theta} = \hathat{\theta}$, and then use the bootstrap distribution to construct one-sided confidence intervals $[\beta_\text{min}^l, \infty)$ and $(-\infty, \beta_\text{max}^u]$ with level $(1 -\alpha/2)$ for the minimal and maximal values, respectively. Finally, we construct the $(1-\alpha)$ sensitivity interval as $[\beta_\text{min}^l, \beta_\text{max}^u]$. In our simulation studies, we compare three different bootstrap procedures: percentile, bias-corrected and accelerated (BCa), and basic (reverse percentile). We found that the former two exhibit coverage close to the nominal level -- even for small sample sizes. By contrast, the empirical coverage of basic bootstrap intervals is 5 to 10 percentage points too low.

\bt{Hence, the bootstrap works well in our simulation study, but the theory requires further investigation.} Previous works by \citet{shapiro91} and \citet{zhao} have studied the asymptotic behaviour of stochastic optimization problems and the use of bootstrap methods to construct confidence/sensitivity intervals. However, due to the complex (non-linear and non-convex) nature of the stochastic constraints, these theoretical guarantees do not cover the optimization problems in this article.

\section{Data Example}\label{sec:example}
We demonstrate the practicality of the proposed method using a
prominent study of the economic return of
schooling. The dataset was compiled by \citet{card} from the National
Longitudinal Survey of Young Men (NLSYM) and contains a sample
of 3010 young men at the age of 14 to 24 in 1966 who were followed up
until 1981. The data is available in the R-package \texttt{ivmodel} \citep{kang_ivmodel_2021}. Card uses several linear models to estimate the causal
effect of education, measured by years of schooling and denoted as
$D$, on the logarithm of earnings, denoted as $Y$. For
brevity, we only consider the most parsimonious model used by
Card which includes, as covariates for adjustment and denoted as $X$,
years of
labour force experience and its square, and indicators for living in the
southern USA, being black and living in a metropolitan area.

\citet{card} recognizes that many researchers are reluctant to
interpret the established positive correlation between education and
earnings as a positive causal effect due to the large number
of potential unmeasured confounders. In our analysis, we
will consider the possibility that an
unmeasured variable $U$, which represents the motivation of the young men, may influence both schooling and salary. To address this issue, Card suggests to
use an instrumental variable, namely the indicator for growing up in
proximity to a 4-year college; this is denoted as $Z$
below. Nonetheless, proximity to college may not be a valid instrumental
variable. For example, growing up near a college may be correlated with
a higher socioeconomic status, more career opportunities, or stronger
motivation. 

\subsection{Sensitivity Analysis with Different Sets of Constraints}\label{sec:sensana-example}
For the purpose of sensitivity analysis, we assume that being black
and living in the southern USA are not directly related with
motivation and treat them as $\xd$; the remaining covariates are
regarded as $\xt$ in the sensitivity analysis. We assume that this
partition satisfies the partial uncorrelatedness condition $R^2_{U \sim \xd \vert \xt, Z} =0$. In this example, we use
comparative bounds to express our beliefs about the effects of the
unmeasured confounder $U$ on $Y$ and $D$. 
We assume that motivation can explain at most 4 times as
much variation in the level of education as being black (denoted as
$\xd_j$) does
after accounting for all other observed covariates, and that
motivation can explain at most 5 times as much variation in
log-earnings as being black does after accounting for the other
covariates and education:
\begin{equation*}
    \text{(B1)}\,\,R^2_{D\sim U\vert \xt, \xd_{-j}, Z} \leq 4\, R^2_{D\sim \xd_j\vert \xt, \xd_{-j},Z},\quad
    \text{(B2)}\,\, R^2_{Y\sim U\vert \xt, \xd_{-j}, Z, D} \leq 5\, R^2_{Y\sim \xd_j\vert \xt, \xd_{-j}, Z, D}.
\end{equation*}
The bounds (B1) and (B2) address deviations from the identification
assumptions of a linear regression. Likewise, we can also specify
deviations from the instrumental variable assumptions. We suppose that
motivation $U$ can explain at most half as much variation in $Z$ (college
proximity) as $\xd_j$ (black) can after accounting for the
effects of $(\xt, \xd_{-j})$. Furthermore, we assume that college
proximity $Z$ can explain at most 10\% as much variance in
log-earnings after
excluding effects of $(X, U, D)$ as being black can explain
log-earnings after accounting for $(\xt, \xd_{-j}, Z, U,
D)$. These assumptions translate to
\begin{equation*}
    \text{(B3)}\,\, R^2_{Z\sim U\vert \xt, \xd_{-j}} \leq 0.5\, R^2_{Z\sim \xd_j\vert \xt, \xd_{-j}},\quad
    \text{(B4)}\,\, R^2_{Y\sim Z\vert X, U, D} \leq 0.1\, R^2_{Y\sim \xd_j\vert \xt, \xd_{-j}, Z, U,D}.
\end{equation*}
When the bound (B1) is not part of the constraints, we additionally require $R_{D\sim U\vert X,Z} \in [-0.98, 0.98]$. 
This ensures that $\lvert R_{D\sim U\vert X,Z}\rvert$ is bounded away from $1$ and that the PIR has finite length.

\begin{figure}[t!]
    \centering
    \includegraphics[scale=0.45]{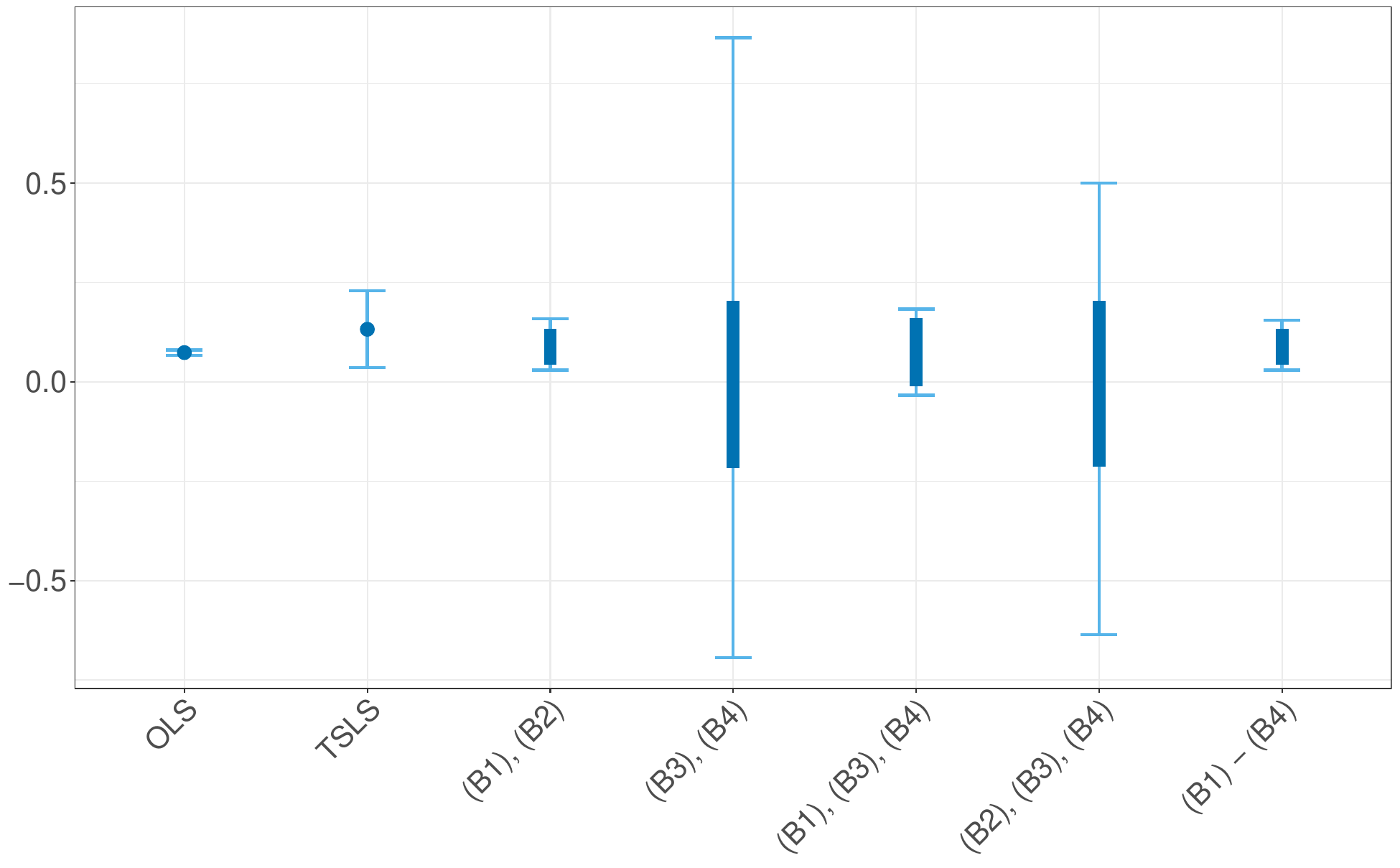}
    \caption{\label{fig:iv-sensint} Two estimation strategies and five sensitivity models for the causal effect $\beta$: Point estimates/estimates of the PIR (thick lines/points); 95\% confidence/sensitivity intervals (thin whiskers).}
  \end{figure}

Figure \ref{fig:iv-sensint} shows the OLS and TSLS estimates as well as their corresponding
95\% confidence intervals. The same plot shows the estimated partially identified regions and 95\% sensitivity intervals (obtained by BCa bootstrap) for five different sensitivity models that involve different combinations of the bounds (B1) to (B4). Both the OLS and the TSLS estimates suggest a statistically significant positive effect of education on earnings. 
 In the first sensitivity model in Figure \ref{fig:iv-sensint}, we
 relax the assumption of no unmeasured confounders, which would be
 required to interpret the OLS estimate causally, and
 assume that the effects of $U$ on $D$ and $Y$ are bounded by (B1)
 and (B2), respectively. In this case, the sensitivity interval remains positive. In other cases, the estimated partially identified regions
 and the sensitivity intervals become very wide whenever
 (B1) is not part of the constraints. Only specifying constraints on the IV sensitivity parameters, such as (B3) and (B4), is generally not sufficient to bound $\lvert R_{D\sim U\vert X,Z}\rvert$ away from 1. This is because the two regression and two IV sensitivity parameters only need to fulfill the equality constraint \eqref{eq:reg-iv-connection} leaving three degrees of freedom. In fact, if the loose bound $R_{D\sim U\vert X,Z} \in [-0.98, 0.98]$ was not imposed, the PIR would have an infinite length and the association between $D$ and $Y$ may be entirely driven by the unmeasured confounder. Furthermore, we notice that adding the bound (B2) to (B3) and (B4) reduces the length of the sensitivity interval but yields the same estimated PIR. Since (B2) is a comparative bound and acts on both regression sensitivity parameters, similarly to \eqref{eq:bound-ex2}, it depends on the specific dataset if (B2) actually bounds $\lvert R_{D\sim U\vert X,Z}\rvert$ away from 1. In the NLSYM data, (B2) does not accomplish this and we need to impose $R_{D\sim U\vert X,Z} \in [-0.98, 0.98]$ to keep the length of the PIR finite.
  
 Moreover, comparing
 the first and last sensitivity model in Figure~\ref{fig:iv-sensint},
 we notice that imposing the IV-related bounds (B3) and (B4) on top of
 (B1) and (B2) does not shorten the estimated PIR and sensitivity
 intervals. These findings suggest that the results of \citet{card}
 are more robust towards deviations from the OLS than from the IV
 assumptions.

\subsection{Contour Plots}
This subsection presents two graphical tools to further aid the interpretation of sensitivity analysis. The main idea is to plot the estimated lower or upper end of the PIR against the sensitivity parameters or the parameters in the comparative bounds. Contour lines in this plot allow practitioners to identify the magnitude of deviation from the identification assumptions that is needed to alter the conclusion of the study qualitatively. This idea dates back at least to \citet{imbens_2003}; our method below extends the proposal in \citet{ch}. 

\subsubsection{$b$-contour Plot}

For comparative bounds, the $b$-factor, such as $b_{U\! D} = 4$ in (B1), controls how tightly the corresponding
sensitivity parameter is constrained. Hence, it is important to gain a
practical understanding of $b$. The $b$-sensitivity contour plot
shows the estimated lower/upper end of the PIR on a grid of
$b$-factors.

\begin{figure}
    \centering
    \begin{minipage}[t]{0.45\textwidth}
        \centering
        \includegraphics[width=\textwidth]{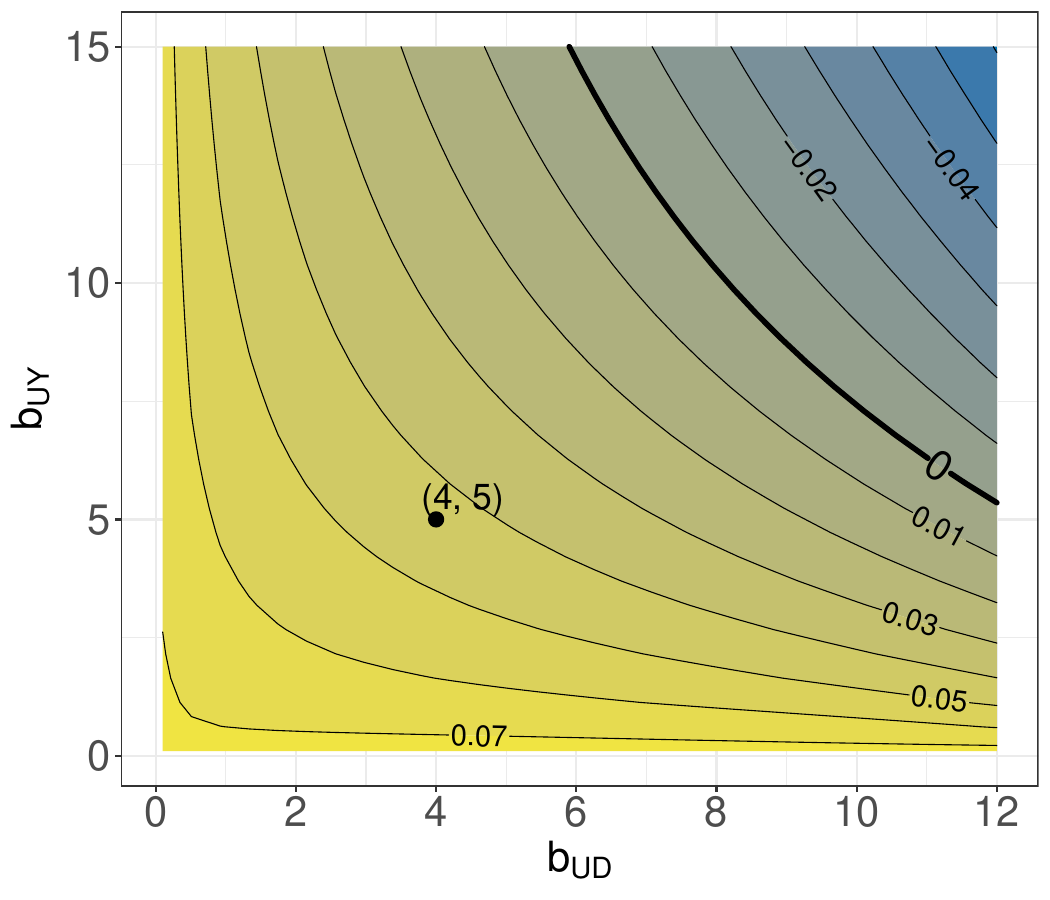}
        \caption{\label{fig:b-contour-reg}$b$-sensitivity contours for (B1), (B2).}
      \end{minipage}
    \hfill
    \begin{minipage}[t]{0.45\textwidth}
        \centering
        \includegraphics[width=\textwidth]{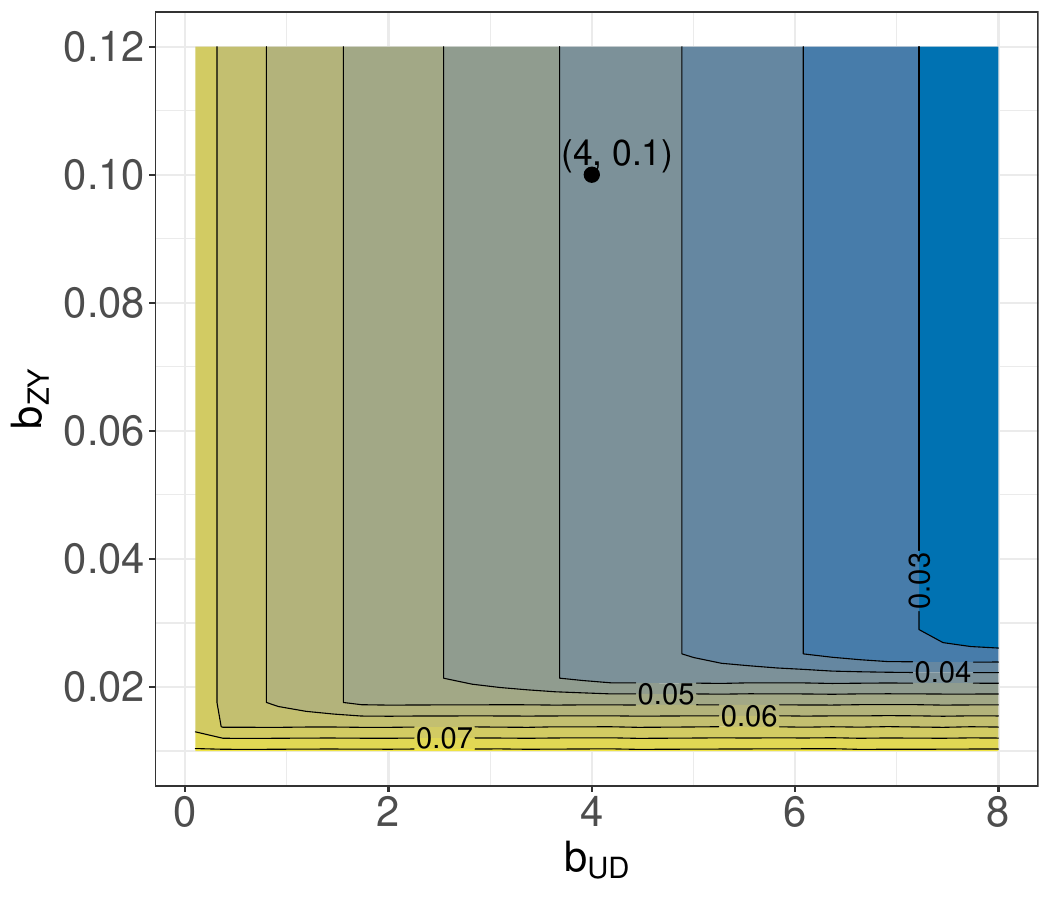}
        \caption{\label{fig:b-contour-iv}$b$-sensitivity contours for (B1)-(B4).}
      \end{minipage}
\end{figure}

In Figure \ref{fig:b-contour-reg}, we consider the sensitivity model
with the bounds (B1) and (B2) and investigate our choice $(b_{U\! D},
b_{UY}) = (4,5)$. The plot shows that the estimated lower end of
the PIR is still positive even for more conservative values such as
$(b_{U\! D}, b_{UY}) = (6, 10)$ or $(b_{U\! D}, b_{UY}) = (10,
5)$. Thus, a substantial deviation from the OLS-related assumptions is
needed to alter the sign of the estimate.

Figure \ref{fig:b-contour-iv} considers the sensitivity model using
the constraints (B1) to (B4) with changing $(b_{U\! D}, b_{ZY})$. This
plot confirms our observation in Section \ref{sec:sensana-example} that
imposing the IV-related
bounds (B3) and (B4) does not change the estimated lower bound when (B1) and (B2) are already present. In the terminology of constrained optimization, this means that the ``shadow
prices'' for (B3) and (B4) are small.

\subsubsection{$R$-contour Plot}\label{sec:r-contours}
\begin{figure}[b!]
    \centering
    \includegraphics[scale=0.55]{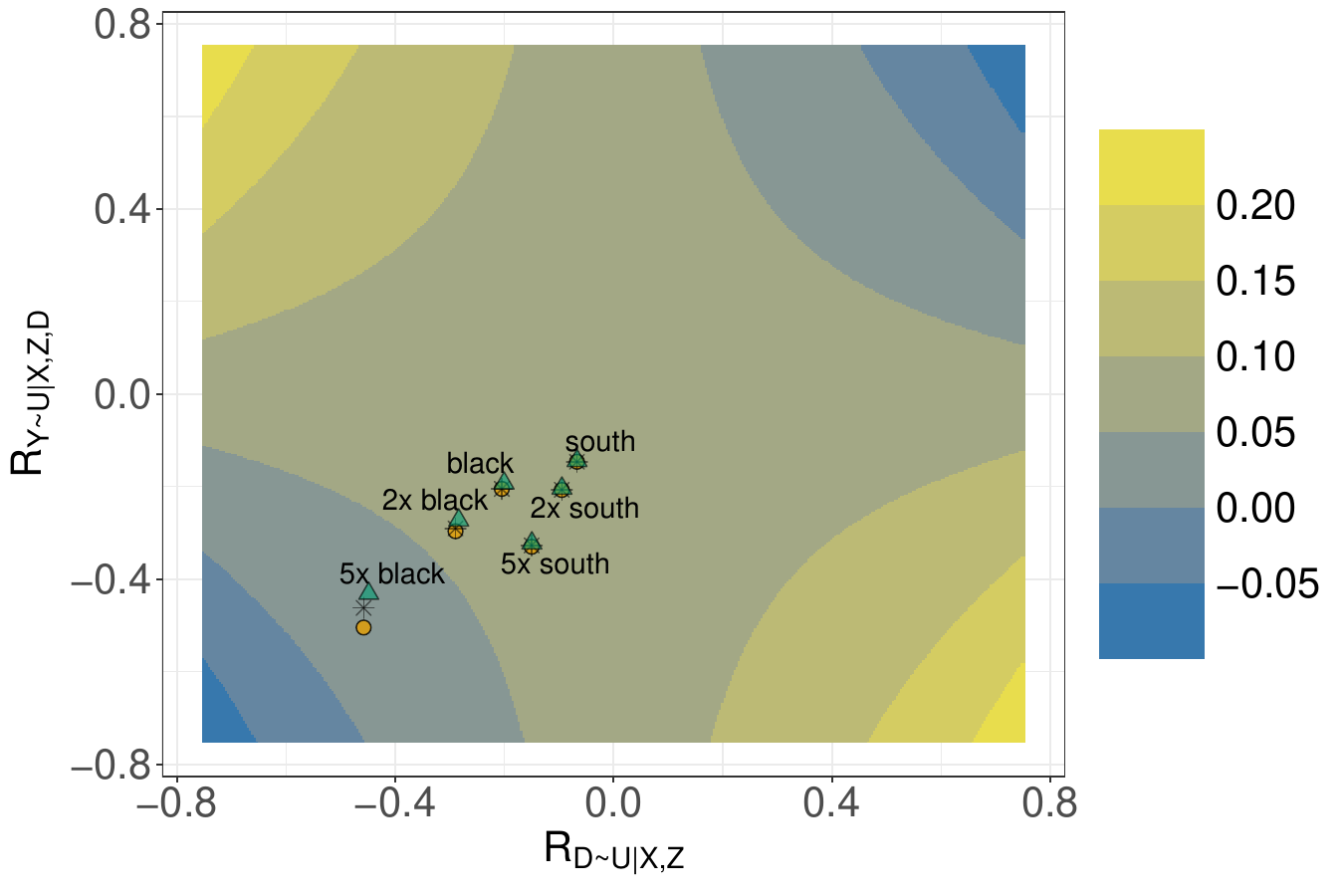}
    \caption{\label{fig:r-contour-reg}$R$-sensitivity contours for the lower end of the estimated PIR: Comparison points unconditional on $D$ (dots), comparison points conditional on $D$ (stars) and non-rigorous comparison points (triangles).}
\end{figure}

We may also directly plot the estimated lower/upper end of the PIR
against the sensitivity parameters $R_{D\sim U\vert X,Z}$ and
$R_{Y\sim U\vert X,Z,D}$. 
This idea has been adopted in several previous articles already
\citep{imbens_2003, veitch2020,ch}.

For such $R$-contour plots, the key challenge is to benchmark or calibrate the partial $R$-values. This was often done informally by adding (multiples of) observed partial $R$-values to provide context for plausible values of the sensitivity parameters. Here, we build upon \citet{ch}'s contour plots and introduce rigorous benchmarking points using the comparative bounds on $U\rightarrow D$ and $U\rightarrow Y$. Since there are two comparative bounds on $U\rightarrow Y$ \bt{(one that partials out $D$ and one that does not)}, we also obtain two comparative values for $R_{Y\sim U\vert X,Z,D}$: one that does not condition on $D$ and one that does. The formulas for the comparison points and their derivations can be found in Appendix~\ref{app:comparison-points}.


In Figure \ref{fig:r-contour-reg}, we depict the $R$-contour plot for the estimated lower end of the PIR. To provide context for the values of $R_{D\sim U\vert X,Z}$ and $R_{Y\sim U\vert X,Z,D}$, we use the observed indicators for being black and living in the southern USA and add three different kinds of benchmarks each: our two proposed rigorous comparison points (unconditional and conditional on $D$) and an informal comparison point following \citeauthor{imbens_2003}' proposal. We notice that, even if the unmeasured confounder was five times as strong as being black in terms of its capability of explaining the variation of $D$ and $Y$, the estimate would still be positive. This holds true regardless of whether the comparison for the effect of $U$ on $Y$ was made conditional or unconditional on $D$. Moreover, Figure \ref{fig:r-contour-reg} shows that the informal comparison point yields a similar result. In our experience, the difference between the methods usually does not change the qualitative conclusion of the sensitivity analysis.



\section{Discussion}\label{sec:discussion}

The methodological developments in this article are underpinned by two insights. First, sensitivity analysis (or more generally, any one-dimensional partially identified problem) may be viewed as a constrained stochastic program and we can leverage methods developed in stochastic optimization. Second, the partial correlations and the $R^2$-calculus provide a parameterization of the linear causal effect and a systematic approach to specify interpretable sensitivity models.

Partial identification has attracted considerable attention in econometrics and causal inference since the work of \citet{manski1990} and \citet{balke}; for instance \citet{manski2003,imbens,chernozhukov2007,aronow2013,miratrix2018,molinari2020}. Existing methods typically assume a closed-form solution to the stochastic program \eqref{eq:pir-minmax} (the lower/upper end of the PIR) and that the plug-in estimator is asymptotically normal. As such results are only known for relatively simple models, these methods only have limited utility in practice. The constrained optimization perspective of partial identification is only beginning to get embraced in the literature \citep{kaido19_confid_inter_projec_partial_ident_param,hu,guo_partial_2022,pmlr-v213-padh23a}.

Our article further shows the need for a more complete, asymptotic
theory of the optimal value of a general stochastic program. This may
allow us to extend the methodology~developed here to sensitivity
models with high- or infinite-dimensional parameters. In
particular, \bt{theoretical guarantees for the convergence of the plug-in estimator as well as results for its bootstrap distribution are required.}
\bt{Moreover, extensions of our proposed methods to multiple (potentially invalid) instruments as well as a deeper investigation of settings with multiple unmeasured confounders are interesting directions for future research.}

The $R^2$-values, $R$-values and generalizations thereof are popular for the calibration of sensitivity analysis. Recently, they have been used for linear models with multiple treatments \citep{zheng2021}, mediation analysis
\citep{zhang}, missing values \citep{colnet2022published} and models with
factor-structured outcomes \citep{zheng_sensitivity_2023}. In these works, certain
algebraic relationships about $R^2$-values and benchmarking techniques
such as contour plots and robustness values are frequently used. Thus, the methodology developed in this article may also benefit the calibration of other sensitivity models. Moreover, our proof of the $R^2$-calculus in general Hilbert spaces -- see Appendix~\ref{sec:r2-calculus-hilbert} -- suggests that it may be useful in nonlinear models, too. See \citet{chernozhukov_long_2022} for
related work in partially linear and semiparametric models using the
Riesz-Frechet representation of certain causal parameters.



\newpage

\appendix

\section{R\textsuperscript{2}-Calculus}\label{app:r2-calculus}

This section states the rules of the $R^2$-calculus and related results as used in the main text. We employ the notation of Section \ref{sec:r2-value}, concentrate on results for the population versions of the (partial) $R$- and $R^2$-values and assume without loss of generality that random variables are centred. The proofs along a more general, Hilbert space version of the $R^2$-calculus can be found in the next section.

The following result justifies calling $R^2_{Y \sim X \mid Z}$ a \emph{partial} $R^2$-value and shows that the definitions of $R^2$- and $R$-value are consistent. It follows from the Gram-Schmidt orthogonalization.

\begin{lemma}\label{lem:r2-corr}
  In the setting of Definition \ref{def:r2-value}, $R^{2}_{Y\sim X\vert Z} = R^{2}_{Y^{\perp Z} \sim X^{\perp Z}}$. Moreover, if $X$ is one-dimensional, then $R^{2}_{Y \sim
    X \vert Z} = (R_{Y \sim
    X \vert Z})^2
  $.
\end{lemma}

Partial $R$- and $R^2$-values that involve overlapping sets of random variables need to fulfill certain algebraic relationships. These can be used to identify the causal effect $\beta$ in terms of sensitivity parameters and derive interpretable constraints. The next Proposition collects several useful algebraic rules we harness in this work.

\begin{proposition}[$R^2$-calculus]\label{prop:r2calc}
    In the setting of Definition \ref{def:r2-value}, assume that $W$ is another random vector such that the covariance matrix of $(Y,X,Z,W)$ is positive definite. Then, the following algebraic rules hold:
  \begin{itemize}
  \addtolength{\leftskip}{3ex}
  \item[{\normalfont[i]}] \emph{Orthogonality}: if $\corr(Y,X_i) = 0$ for all components $i$ of $X$, then $R^2_{Y\sim X} = 0$;
  \item[{\normalfont[ii]}]\emph{Orthogonal additivity}: if $R^2_{X_i \sim W} = 0$ for all components $i$ of $X$, then
  \begin{equation*}
      R^{2}_{Y\sim X+W} = R^{2}_{Y\sim X} + R^{2}_{Y\sim
      W};
  \end{equation*}
  \item[{\normalfont[iii]}] \emph{Decomposition of unexplained variance}:
    \begin{equation*}
      1-R^{2}_{Y\sim X + W} = (1-R^{2}_{Y\sim X})(1-R^{2}_{Y\sim W \vert X});
    \end{equation*}
  \item[{\normalfont[iv]}] \emph{Recursion of partial correlation}: if $X$ and $W$ are
    one-dimensional, then
    \begin{equation*}
      R_{Y\sim X\vert W} = \frac{R_{Y\sim X} - R_{Y\sim W}R_{X\sim
          W}}{\sqrt{1-R_{Y\sim W}^{2}}\sqrt{1-R^{2}_{X \sim W }}};
    \end{equation*}
  \item[{\normalfont[v]}] \emph{Reduction of partial correlation}: if $X$ is
    one-dimensional and $R^2_{Y\sim W}=0$, then
    \begin{equation*}
      R_{Y\sim X\vert W} = \frac{R_{Y\sim X}}{\sqrt{1-R^{2}_{X\sim W}}};
    \end{equation*}
  \item[{\normalfont[vi]}] \emph{Three-variable identity}: if both $X$ and $W$ are
    one-dimensional, then
    \begin{equation*}
      f_{Y\sim X \vert W}\sqrt{1-R^{2}_{Y\sim W\vert X}} = f_{Y\sim
        X}\sqrt{1-R^{2}_{X\sim W}}-R_{Y\sim W\vert
        X}R_{X\sim W}.
    \end{equation*}
  \end{itemize}
  All of the relationships above also hold when $Z$ is partialed out, i.e.\ if ``\,$|Z$'' is appended to the subscripts of all $R$-, $R^2$-, and $f$-values and $Z$ is partialed out in the correlations.
\end{proposition}

\begin{remark}
  Rule [vi] may appear unintuitive at first. To see how this identity may come up, consider three random variables $Y$, $X$ and $W$. There are in total three marginal $R$-values, $R_{Y\sim X}$, $R_{Y\sim W}$ and $R_{X\sim W}$, and three partial $R$-values, $R_{Y\sim X\vert W}$, $R_{Y\sim W\vert X}$ and $R_{X\sim W\vert Y}$. Rule [iv] shows that the partial $R$-values can be determined by all the   marginal values. In other words, there are only three degrees of freedom among the six $R$-values. This implies that there must be an equality constraint relating $R_{Y\sim X}$, $R_{X \sim W}$, $R_{Y\sim X\vert W}$, and $R_{Y\sim W\vert X}$.
\end{remark}

In general, a conditional independence statement such as $Y \indep W \,\vert\, X$ does not imply that $Y$ and $W$ are \emph{partially} uncorrelated given $X$. However, if the conditional expectation of $Y$ and $W$ is affine linear in $X$, $\corr(Y^{\perp X}, W^{\perp X})=0$ indeed holds true \citep[Thm.\ 1]{baba_partial_2004}. Hence, if we assume that $(Y,X,Z,D,U)$ follow a linear structural equation model, we could replace the assumptions in rules [i], [ii] and [v] with corresponding (conditional) independence statements. For more details on the relationship of partial correlation, conditional correlation and conditional independence, we refer to \citet{baba_partial_2004}.



\section{Hilbert Space R\textsuperscript{2}-Calculus and Proofs}
\label{sec:r2-calculus-hilbert}

The algebraic rules of the $R^2$-calculus -- both the population
version in Proposition \ref{prop:r2calc} and its sample counterpart -- are not specific to linear models. In fact, all
relationships fundamentally stem from the geometry of projections in
Hilbert spaces. For this reason, the definitions of $R^2$-
and $R$-values can be generalized and the corresponding algebraic rules be proven in more generality. 

This section is organized as follows. First, we recall some results on
Hilbert space theory
\citep[sec. 26-29]{halmos2000} and define generalized (partial) $R^2$-
and $R$-values. Then, we prove Hilbert space generalizations to Lemma~\ref{lem:r2-corr} and Proposition~\ref{prop:r2calc}. Finally, in Section~\ref{app:r2-linear-models}, we explain how the $R^2$-calculus for linear
models directly follows from these results and provide more
details on the assumptions and notation involved.

\subsection{Hilbert Space R\textsuperscript{2}-value}
Let $(\hc,\langle\cdot,\cdot\rangle)$ be a Hilbert space over the
field $\mathbb{K}$ of real or complex numbers; denote its associated
norm as $\lVert \cdot \rVert$ and let $\xc, \yc, \zc\subseteq \hc$ be
closed linear subspaces. The Minkowski sum of $\yc$ and $\xc$ is given by
$\xc+\yc :=\{x+y\colon x\in \xc, y\in \yc\}$. For $x \in \xc$ and
$y\in \yc$, we write $x \perp y$, if $\langle x, y \rangle =0$, $x
\perp \yc$, if $x\perp y$ for all $y \in \yc$, and $\xc\perp \yc$, if
$x \perp \yc$ for all $x \in \xc$. For every element $h \in \hc$,
there are unique $x \in \xc$ and $x^\perp \in \hc$ such that $x \perp
x^\perp$ and $h = x + x^\perp$.

\begin{definition}\label{def:projection}
The projection on $\xc$ is the operator $P_\xc \colon \hc \to \xc$ defined by the assignment ${h = x + x^\perp \mapsto x}$. The projection off $\xc$ is
the operator $Q_\xc \colon \hc \to \hc$ defined by ${h = x + x^\perp \mapsto x^\perp}$.
\end{definition}

Clearly, the projection on and off $\xc$ add up to the identity
operator, i.e.\ $P_\xc + Q_\xc = \Id$. Furthermore, we introduce the
notations $y^{\perp \xc} := Q_\xc\, y$ and $\yc^{\perp \xc} :=
\{y^{\perp \xc} \colon {y\in\yc}\}$. The space $\yc^{\perp \xc}$ is a
closed linear subspace of $\hc$; thus, the projections $P_{\yc^{\perp
    \xc}}$ and $Q_{\yc^{\perp \xc}}$ are well-defined. They can be
used to define conditional orthogonality: $\yc \perp \xc \,\vert\, \zc
\Leftrightarrow \yc^{\perp \zc} \perp \xc^{\perp \zc}$.

\begin{lemma}\label{lem:proj-HS}
\phantom{hello}\\[-5ex]
\begin{itemize}
\addtolength{\leftskip}{3ex}
    \item[{\normalfont(i)}] $P_\xc$ and $Q_\xc$ are linear, self-adjoint, and idempotent operators.
    \item[{\normalfont(ii)}] If $\xc\perp \yc$, $P_{\xc+\yc} = P_\xc + P_\yc$ and $Q_{\xc+\yc} = Q_\xc\,Q_\yc$.
    \item[{\normalfont(iii)}] $P_{\xc+\yc} = P_{\xc} + P_{\yc^{\perp \xc}}$ and $Q_{\xc+\yc} = Q_\xc\,Q_{\yc^{\perp \xc}} = Q_{\yc^{\perp \xc}}\, Q_\xc$.
    \item[{\normalfont(iv)}] If $h_1, h_2 \in \hc$ and $h_1 \perp h_2$, $\lVert h_1 + h_2 \rVert^2 = \lVert h_1\rVert^2 + \lVert h_2\rVert^2$.
\end{itemize}
\end{lemma}

\begin{proof}\phantom{hello}\\[-5ex]
\begin{itemize}
\addtolength{\leftskip}{3ex}
\item[{\normalfont(i)}] See \citet[sec. 26, Thm. 1]{halmos2000}.
\item[{\normalfont(ii)}] See \citet[sec. 28, Thm. 2]{halmos2000} for the proof of $P_{\xc+\yc} = P_\xc + P_\yc$. According to \citet[sec. 29, Thm. 1]{halmos2000}, $P_\xc P_\yc = 0$ holds due to $\xc \perp \yc$. Hence, the second statement directly follows:
$Q_{\xc+\yc} = \Id - P_\xc - P_\yc = (\Id - P_\xc) (\Id - P_\yc) = Q_\xc \, Q_\yc$.
\item[{\normalfont(iii)}] We rewrite the direct sum $\xc + \yc$ as follows
\begin{align*}
    \xc + \yc &= \{x+y \colon x\in \xc, y \in \yc\} = \{x + P_\xc\, y + Q_\xc\, y \colon x\in \xc, y \in \yc\}\\
    &= \{x + Q_\xc\, y \colon x\in \xc, y \in \yc\} = \xc + \yc^{\perp \xc}.
\end{align*}
Since $\xc$ and $\yc^{\perp \xc}$ are orthogonal by definition, the result directly follows from (ii).
\item[{\normalfont(iv)}] See \citet[sec. 4, Thm. 3]{halmos2000}. \hfill\qedhere
\end{itemize}
\end{proof}

Any one-dimensional linear subspace $\xc$ can be expressed as $\xc = \{\lambda\, x\colon \lambda \in \mathbb{K}\}$, where $x$ is an arbitrary element in $\xc \setminus \{0\}$. Hence, we can identify a one-dimensional subspace with any non-zero element contained in it.

\begin{definition}[Hilbert space $R^2$- and $R$-value]\label{def:HS-r2}
Let $\xc, \yc, \zc \subseteq \hc$ be closed linear subspaces. Assume $\yc$ is one-dimensional, let $y \in \yc\setminus \{0\}$ and suppose $\lVert y^{\perp \zc}\rVert^2 > 0$. The (marginal) $R^2$-value of $\yc$ on $\xc$, the partial $R^2$-value of $\yc$ on $\xc$ given $\zc$ and the partial $f^2$-value of $\yc$ on $\xc$ given $\zc$ are defined as
\begin{equation*}
    R^2_{\yc \sim \xc} := 1 - \frac{\lVert y^{\perp \xc}\rVert^2}{\lVert y\rVert^2},\qquad
    R^2_{\yc \sim \xc\vert \zc} := \frac{R^2_{\yc\sim \xc + \zc} - R^2_{\yc\sim \zc}}{1- R^2_{\yc\sim \zc}},\qquad
    f^2_{\yc \sim \xc\vert \zc} := \frac{R^2_{\yc \sim \xc\vert \zc}}{1-R^2_{\yc \sim \xc\vert \zc}},
\end{equation*}
respectively.
If $\xc$ is one-dimensional, $x\in \xc\setminus\{0\}$ and $\lVert x^{\perp \zc}\rVert^2 > 0$, the partial $R$- and $f$-value are given by
\begin{equation*}
    R_{\yc \sim \xc \vert \zc} := \frac{\langle y^{\perp \zc}, x^{\perp \zc}\rangle}{\lVert y^{\perp \zc} \rVert \, \lVert x^{\perp \zc} \rVert}, \qquad
    f_{\yc \sim \xc \vert \zc} := \frac{R_{\yc \sim \xc \vert \zc}}{\sqrt{1-R_{\yc \sim \xc \vert \zc}^2}}.
\end{equation*}
\end{definition}
The choice of the non-zero elements $y$ and $x$ does not change the (partial) $R^2$- and $R$-values due to the normalization. Therefore, all quantities above are well-defined.

\subsection{Proofs of Results in Appendix~\ref{app:r2-calculus}}
In this subsection, we state and prove the generalized versions of Lemma \ref{lem:r2-corr} and Proposition~\ref{prop:r2calc}.

\begin{lemma}\label{lem:HS-r2-corr}
  In the setting of Definition \ref{def:HS-r2}, $R^{2}_{\yc\sim \xc\vert \zc} = R^{2}_{\yc^{\perp \zc} \sim \xc^{\perp \zc}}$ holds true. Moreover, if $\xc$ is a one-dimensional subspace, then $R^{2}_{\yc \sim
    \xc \vert \zc} = (R_{\yc \sim \xc \vert \zc})^2$.
\end{lemma}

\begin{proof}
The first statement of the lemma follows from some elementary algebraic manipulations and applying Lemma \ref{lem:proj-HS} (iii):
\begin{align*}
    R^2_{\yc\sim \xc\vert \zc} &= \frac{R^2_{\yc\sim \xc + \zc} - R^2_{\yc\sim \zc}}{1- R^2_{\yc\sim \zc}}
    \left[1 - \frac{\lVert y^{\perp \xc + \zc} \rVert^2}{\lVert y \rVert^2} - 1 + \frac{\lVert y^{\perp \zc} \rVert^2}{\lVert y \rVert^2} \right] \bigg/ \frac{\lVert y^{\perp \zc} \rVert^2}{\lVert y \rVert^2}\\
    &= 1- \frac{\lVert y^{\perp \xc + \zc} \rVert^2}{\lVert y^{\perp \zc} \rVert^2}
    \stackrel{\text{(iii)}}{=} 1- \frac{\lVert Q_{\xc^{\perp \zc}}\, y^{\perp \zc} \rVert^2}{\lVert y^{\perp \zc} \rVert^2}
    = R^2_{\yc^{\perp \zc} \sim \xc^{\perp \zc}}.
\end{align*}
To prove the second part of the lemma, we assume that $\xc$ is one-dimensional and choose $x \in \xc\setminus \{0\}$. If $\xc^{\perp \zc} = 0$, the projection on $\xc^{\perp \zc}$ is 0; otherwise, it is given by
\begin{equation}\label{eq:1d-proj}
    P_{\xc^{\perp \zc}} h = \frac{\langle h, x^{\perp \zc}\rangle x^{\perp \zc}}{\lVert x^{\perp \zc} \rVert^2},\quad \text{for } h \in \hc.
\end{equation}
This can be easily checked: $P_{\xc^{\perp \zc}}$ is linear and its image is contained in $\xc^{\perp \zc}$. Moreover,
\begin{equation*}
    \bigg\langle\frac{\langle h, x^{\perp\zc}\rangle\, x^{\perp\zc}}{\lVert x^{\perp\zc} \rVert^2},\, h - \frac{\langle h, x^{\perp\zc}\rangle\, x^{\perp\zc}}{\lVert x^{\perp\zc} \rVert^2} \bigg\rangle
    = \frac{\langle h,x^{\perp\zc} \rangle^2}{\lVert x^{\perp\zc} \rVert^2} - \frac{\langle h,x^{\perp\zc} \rangle^2 \lVert x^{\perp\zc}\rVert^2}{\lVert x^{\perp\zc} \rVert^4} = 0.
\end{equation*}
Following from the first part of the proof and Lemma \ref{lem:proj-HS} (iv), we infer
\begin{equation*}
    R^2_{\yc\sim \xc\vert \zc} = 1- \frac{\lVert Q_{\xc^{\perp \zc}}\, y^{\perp \zc} \rVert^2}{\lVert y^{\perp \zc} \rVert^2}
    \stackrel{\text{(iv)}}{=} \frac{\lVert P_{\xc^{\perp \zc}}\, y^{\perp \zc} \rVert^2}{\lVert y^{\perp \zc} \rVert^2}.
\end{equation*}
Inserting the formula for the projection on $
\xc^{\perp \zc}$ yields the second statement of the lemma
\begin{equation*}
    R^2_{\yc\sim \xc \vert \zc} = \frac{\lVert \langle y^{\perp \zc}, x^{\perp \zc}\rangle\, x^{\perp \zc}
    \rVert^2}{\lVert y^{\perp \zc} \rVert^2\, \lVert x^{\perp \zc} \rVert^4}
    = \frac{\langle y^{\perp \zc} , x^{\perp \zc} \rangle^2 \rVert x^{\perp \zc}
    \rVert^2}{\lVert y^{\perp \zc} \rVert^2\, \lVert x^{\perp \zc} \rVert^4}
    = \big(R_{\yc\sim \xc \vert \zc} \big)^2.\hfill\qedhere
\end{equation*}
\end{proof}

\begin{proposition}[Hilbert space $R^2$-calculus]\label{prop:HS-r2calc}
  In the setting of Definition \ref{def:HS-r2}, let $\wc$ be another closed linear subspace. Assume
  $\lVert \yc^{\perp \xc + \wc + \zc} \rVert^2> 0$. Further suppose ${\lVert \xc^{\perp \wc + \zc}\rVert^2\! >\! 0}$ and $\lVert \wc^{\perp \xc + \zc}\rVert^2 > 0$ when $\xc$ and/or $\wc$ are one-dimensional subspaces. Then, the following rules hold
  \begin{itemize}
  \addtolength{\leftskip}{3ex}
  \item[{\normalfont[i]}] Orthogonality: if $\yc\perp \xc$, $R^2_{\yc\sim \xc} = 0$;
  \item[{\normalfont[ii]}] Orthogonal additivity: if $\xc \perp \wc$, then
  \begin{equation*}
      R^{2}_{\yc\sim \xc + \wc} = R^{2}_{\yc\sim \xc} + R^{2}_{\yc\sim \wc};
  \end{equation*} 
  \item[{\normalfont[iii]}] Decomposition of unexplained variation: 
    \begin{equation*}
      1-R^{2}_{\yc\sim \xc + \wc} = (1-R^{2}_{\yc\sim \xc})(1-R^{2}_{\yc\sim \wc \vert \xc});
    \end{equation*}
  \item[{\normalfont[iv]}] Recursion of partial $R$-value: if $\xc$ and $\wc$ are one-dimensional,
    \begin{equation*}
      R_{\yc\sim \xc\vert \wc} = \frac{R_{\yc\sim \xc} - R_{\yc\sim \wc}R_{\xc\sim \wc}}{\sqrt{1-R_{\yc\sim \wc}^{2}}\sqrt{1-R^{2}_{\xc \sim \wc }}};
    \end{equation*}
  \item[{\normalfont[v]}] Reduction of partial $R$-value: if $\xc$ is one-dimensional and ${\yc\perp \wc}$, then 
    \begin{equation*}
      R_{\yc\sim \xc\vert \wc} = \frac{R_{\yc\sim \xc}}{\sqrt{1-R^{2}_{\xc\sim \wc}}};
    \end{equation*}
  \item[{\normalfont[vi]}] Three-dimensional restriction: if $\xc$ and $\wc$ are one-dimensional, then
    \begin{equation*}
      f_{\yc\sim \xc \vert \wc}\sqrt{1-R^{2}_{\yc\sim \wc\vert \xc}} = f_{\yc\sim
        \xc}\sqrt{1-R^{2}_{\xc\sim \wc}}-R_{\yc\sim \wc\vert \xc}R_{\xc\sim \wc}.
    \end{equation*}
  \end{itemize}
  All of the relationships above also hold when $\zc$ is partialed out (i.e.\ if ``\,$|\zc$'' is appended to the subscripts of all $R$-, $R^2$-, and $f$-values) and the orthogonality assumptions are conditional on $\zc$.
\end{proposition}

\begin{proof}
\phantom{hello}\\[-4ex]
\begin{itemize}
    \addtolength{\leftskip}{3ex}
    \item[{\normalfont[i]}] Since $\yc^{\perp \zc}$ and $\xc^{\perp \zc}$ are orthogonal, $Q_{\xc^{\perp \zc}} y^{\perp \zc} = y^{\perp \zc}$ holds. Hence, we obtain
    \begin{equation*}
        R^2_{\yc \sim \xc \vert \zc} = 1 - \frac{\lVert y^{\perp \zc}\rVert^2}{\lVert y^{\perp \zc}\rVert^2}=0.
    \end{equation*}
    \item[{\normalfont[ii]}] According to Lemma \ref{lem:HS-r2-corr} and its proof, we obtain
    \begin{equation*}
        R^2_{\yc \sim \xc + \wc \vert \zc}
        = R^2_{\yc^{\perp \zc} \sim \xc^{\perp \zc} + \wc^{\perp \zc}}
        = \frac{\lVert P_{\xc^{\perp \zc} + \wc^{\perp \zc}}\, y^{\perp \zc}\rVert^2}{\lVert y^{\perp \zc}\rVert^2}.
    \end{equation*}
    Following from Lemma \ref{lem:proj-HS} (ii) and (iv), we get
    \begin{align*}
        R^2_{\yc \sim \xc + \wc \vert \zc}
        &\stackrel{\text{(ii)}}{=} \frac{\lVert P_{\xc^{\perp \zc}} y^{\perp \zc} + P_{\wc^{\perp \zc}}\, y^{\perp \zc}\rVert^2}{\lVert y^{\perp \zc}\rVert^2}
        \stackrel{\text{(iv)}}{=} \frac{\lVert P_{\xc^{\perp \zc}}\, y^{\perp \zc} \rVert^2}{\lVert y^{\perp \zc}\rVert^2}
        + \frac{\lVert P_{\wc^{\perp \zc}}\, y^{\perp \zc}\rVert^2}{\lVert y^{\perp \zc}\rVert^2}\\
        &= R^2_{\yc \sim \xc \vert \zc} + R^2_{\yc \sim \wc \vert \zc}.
    \end{align*}
    \item[{\normalfont[iii]}] The statement directly follows from the definition of the partial $R^2$-value
    \begin{align*}
        \left(1-R^2_{\yc \sim \xc \vert \zc}\right)\left(1-R^2_{\yc\sim \wc\vert \xc + \zc}\right)
        &= \frac{\lVert y^{\perp \xc + \zc}\rVert^2}{\lVert y^{\perp \zc}\rVert^2} \frac{\lVert y^{\perp \wc+\xc+\zc}\rVert^2}{\lVert y^{\perp \xc+\zc}\rVert^2}\\
        &= \frac{\lVert y^{\perp \wc+\xc+\zc}\rVert^2}{\lVert y^{\perp \zc}\rVert^2} = 1- R^2_{\yc\sim \wc + \xc\vert \zc}.
    \end{align*}
    \item[{\normalfont[iv]}] Plugging in the definition of the partial $R$-value into the right-hand side, we get
    \begin{align*}
        \text{RHS} &= \frac{R_{\yc \sim \xc\vert\zc} - R_{\yc \sim \wc\vert\zc}\,R_{\xc \sim \wc\vert\zc}}{\sqrt{1-R_{\yc \sim \wc\vert\zc}^2}\sqrt{1-R_{\xc \sim \wc\vert\zc}^2}}\\[1ex]
        &=\left[\frac{\langle y^{\perp \zc}, x^{\perp \zc}\rangle}{\lVert y^{\perp \zc}\rVert \lVert x^{\perp \zc}\rVert} -\frac{\langle y^{\perp \zc}, w^{\perp \zc}\rangle}{\lVert y^{\perp \zc}\rVert \lVert w^{\perp \zc}\rVert} \frac{\langle x^{\perp \zc}, w^{\perp \zc}\rangle}{\lVert x^{\perp \zc}\rVert \lVert w^{\perp \zc}\rVert} \right] \bigg/\left[ \frac{\lVert y^{\perp \wc + \zc} \rVert}{\lVert y^{\perp \zc} \rVert} \frac{\lVert x^{\perp \wc + \zc} \rVert}{\lVert x^{\perp \zc} \rVert}\right]\\[1ex]
        &= \frac{\langle y^{\perp \zc}, x^{\perp \zc}\rangle}{\lVert y^{\perp \wc+\zc}\rVert \lVert x^{\perp \wc+\zc}\rVert}
        - \frac{\langle y^{\perp \zc}, w^{\perp \zc}\rangle\,\langle x^{\perp \zc}, w^{\perp \zc}\rangle}{\lVert w^{\perp \zc}\rVert^2 \lVert y^{\perp \wc+\zc}\rVert \lVert x^{\perp \wc+\zc}\rVert}.
    \end{align*}
    Recalling the formula \eqref{eq:1d-proj} for the projection operator on a one-dimensional subspace, we can reformulate the upper equation further
    \begin{align*}
        \text{RHS} &= \frac{\Big\langle y^{\perp \zc}, x^{\perp \zc} - \frac{\langle x^{\perp \zc}, w^{\perp \zc}\rangle\,w^{\perp \zc}}{\lVert w^{\perp \zc}\rVert^2} \Big\rangle}{\lVert y^{\perp \wc+\zc}\rVert \lVert x^{\perp \wc+\zc}\rVert}
        = \frac{\langle y^{\perp \zc}, Q_{\wc^{\perp \zc}}\, x^{\perp \zc} \rangle}{\lVert y^{\perp \wc+\zc}\rVert \lVert x^{\perp \wc+\zc}\rVert}\\
        &\stackrel{\text{(iii)}}{=} \frac{\langle y^{\perp \wc+\zc}, x^{\perp \wc+\zc} \rangle}{\lVert y^{\perp \wc+\zc}\rVert \lVert x^{\perp \wc+\zc}\rVert} = R_{\yc \sim \xc\vert \wc+\zc} = \text{LHS},
    \end{align*}
    where the third equality follows from Lemma \ref{lem:proj-HS} (iii).
    \item[{\normalfont[v]}] Let $(w_j^{\perp \zc})_{j\in \{1,\ldots,J\}}$, be an orthonormal basis of $\wc^{\perp \zc}$. The subspace spanned by the first $j$ vectors is denoted by $\wc_j^{\perp \zc} := \Span\{w_1^{\perp \zc},\ldots,w_j^{\perp \zc}\}$. Due to rule [i] and $\yc \perp \wc \,\vert\, \zc$, $R^2_{\yc \sim \wc_j\vert \zc} = 0$ and $R^2_{\yc \sim \wc_{j+1}\vert \wc_j + \zc} = 0$ hold for all $j\in\{1,\ldots, J-1\}$. By induction, we prove the statement $R_{\yc \sim \xc\vert \zc + \wc_j} = R_{\yc\sim \xc\vert \zc}/\sqrt{1-R^2_{\xc \sim \wc_j\vert \zc}}$ for all $j \in \{1,\ldots, J\}$.
    For the base case, we apply rule [iv] and $R_{\yc \sim \wc_1\vert \zc} = 0$ as follows
    \begin{equation*}
        R_{\yc \sim \xc\vert \wc_1 + \zc} \stackrel{\text{[iv]}}{=}
        \frac{R_{\yc\sim \xc\vert \zc} - R_{\yc\sim \wc_1\vert \zc}\,R_{\xc\sim \wc_1\vert \zc}}{\sqrt{1-R_{\yc\sim \wc_1\vert \zc}^2}\sqrt{1-R_{\xc\sim \wc_1\vert \zc}^2}}
        = \frac{R_{\yc\sim \xc\vert \zc} }{\sqrt{1-R_{\xc\sim \wc_1\vert \zc}^2}}.
    \end{equation*}
    The induction step uses rule [iv] and simplifies the resulting expression via\\ ${R_{\yc \sim \wc_{j+1}\vert \wc_j + \zc}=0}$, the induction hypothesis and rule [iii]:
    \begin{align*}
        R_{\yc \sim \xc \vert \wc_{j+1}+\zc}
        &\stackrel{\text{[iv]}}{=} \frac{R_{\yc\sim \xc\vert \wc_j + \zc} - R_{\yc\sim \wc_{j+1}\vert \wc_j+\zc}\,R_{\xc\sim \wc_{j+1}\vert \wc_j + \zc}}{\sqrt{1-R_{\yc\sim \wc_{j+1}\vert \wc_j+\zc}^2}\sqrt{1-R_{\xc\sim \wc_{j+1}\vert \wc_j + \zc}^2}}
        = \frac{R_{\yc\sim \xc\vert \wc_j + \zc}}{\sqrt{1-R_{\xc\sim \wc_{j+1}\vert \wc_j + \zc}^2}}\\
        &= \frac{R_{\yc\sim \xc\vert \zc}}{\sqrt{1-R^2_{\xc \sim \wc_j\vert \zc}}\sqrt{1-R_{\xc\sim \wc_{j+1}\vert \wc_j + \zc}^2}}
        \stackrel{\text{[iii]}}{=} \frac{R_{\yc\sim \xc\vert \zc}}{\sqrt{1-R^2_{\xc \sim \wc_{j+1}\vert \zc}}}.
    \end{align*}

    \item[{\normalfont[vi]}] First, we apply rule [iv] to $R_{\yc\sim \xc \vert \wc + \zc}$ and
    $R_{\yc \sim \wc \vert \xc+\zc}$
    \begin{align*}
      R_{\yc\sim \xc\vert \wc +\zc} &= \frac{R_{\yc\sim \xc\vert \zc} - R_{\yc\sim \wc
          \vert \zc}R_{\xc\sim \wc\vert \zc}}{\sqrt{1-R_{\yc\sim \wc
            \vert \zc}^{2}}\sqrt{1-R_{\xc\sim \wc\vert \zc}^{2}}},\\
      R_{\yc\sim \wc \vert \xc + \zc} &= \frac{R_{\yc\sim \wc \vert \zc} - R_{\yc\sim \xc
          \vert \zc}R_{\xc\sim \wc\vert \zc}}{\sqrt{1-R^{2}_{\yc\sim \xc \vert
            \zc}}\sqrt{1-R^{2}_{\xc \sim \wc\vert \zc}}},
    \end{align*}
    and compute
    \begin{multline*}
      R_{\yc\sim \xc \vert \wc + \zc}\sqrt{1-R^{2}_{\yc\sim \wc \vert \zc}} + R_{\yc\sim
        \wc\vert \xc + \zc} R_{\xc\sim \wc \vert \zc}\sqrt{1-R^{2}_{\yc\sim \xc\vert
          \zc}}\\
      \begin{aligned}
        &=
        \frac{R_{\yc\sim \xc
          \vert \zc} - R_{\yc\sim \wc \vert \zc}R_{\xc\sim \wc\vert \zc}
          + R_{\yc\sim
          \wc \vert \zc}R_{\xc\sim \wc\vert \zc} - R_{\yc\sim \xc \vert
          \zc}R^{2}_{\wc\sim \xc \vert
          \zc}}{\sqrt{1-R^{2}_{\xc\sim \wc\vert \zc}}}\\
        &= R_{\yc\sim \xc \vert \zc} \sqrt{1-R^{2}_{\xc\sim \wc \vert \zc}}.
      \end{aligned}
    \end{multline*}
    Next, we divide both sides of the equation by
    $\sqrt{1-R^{2}_{\yc\sim \xc \vert \zc}}$ and rearrange it which results in
    \begin{equation*}
      R_{\yc\sim \xc \vert \wc + \zc}\frac{\sqrt{1-R^{2}_{\yc\sim \wc \vert
            \zc}}}{\sqrt{1-R^{2}_{\yc\sim \xc \vert \zc}}} = f_{\yc\sim
        \xc\vert \zc}\sqrt{1-R^{2}_{\xc\sim \wc \vert \zc}}-R_{\yc\sim \wc\vert
        \xc+\zc}R_{\xc\sim \wc \vert \zc}.
    \end{equation*}
    According to rule [iii], we obtain
    \begin{equation*}
     (1-R^{2}_{\yc\sim \xc \vert \zc})(1-R^{2}_{\yc\sim \wc \vert \xc + \zc})
     =1-R^{2}_{\yc\sim \xc + \wc \vert \zc} = (1-R^{2}_{\yc\sim \wc \vert
       \zc})(1-R^{2}_{\yc\sim \xc\vert \wc + \zc})
   \end{equation*}
   and, thus, $(1-R^{2}_{\yc\sim \wc \vert \zc})/(1-R^{2}_{\yc\sim \xc\vert \zc}) = (1-R^{2}_{\yc\sim \wc\vert \xc + \zc})/(1-R^{2}_{\yc\sim \xc\vert \wc + \zc})$.
   Plugging this relationship into the left-hand side of the upper
   equation, we arrive at
   \begin{equation*}
     f_{\yc\sim \xc \vert \wc+\zc}\sqrt{1-R^{2}_{\yc\sim \wc\vert \xc + \zc}} = f_{\yc\sim
      \xc\vert \zc}\sqrt{1-R^{2}_{\xc\sim \wc \vert \zc}}-R_{\yc\sim \wc\vert
      \xc + \zc}\,R_{\xc\sim \wc \vert \zc}. \hfill\qedhere
   \end{equation*}
\end{itemize}
\end{proof}

\subsection{R\textsuperscript{2}-calculus for Covariance Matrices}\label{app:r2-linear-models}

The $R^2$-calculus as presented in Section~\ref{app:r2-calculus} is
a special case of the $R^2$-calculus for Hilbert spaces.
To be consistent with the standard notation for $R^2$-values in linear
models, we make two slight changes to the Hilbert space
notation. First, a random vector denotes the linear space that is
spanned by its components. Analogously, for an i.i.d. sample of size
$n$ for a $p$-dimensional random vector $X$, we use the matrix $X \in
\R^{n\times p}$ to denote the row-space. Second, we replace the
plus-sign with a comma for partialed out variables. For instance, we
write $R^2_{Y\sim X\vert W,Z}$ instead of $R^2_{Y\sim X\vert W+Z}$ in
the main text.

Denote the space of square-integrable random variables $L^2 := \{X \colon \E[X^2] < \infty\}$.
We define the following four Hilbert spaces with associated inner
products
\begin{alignat*}{8}
    &\hc & &:= L^2,\quad\!& \langle X &, Y \rangle_{\hc} & &:= \E[X Y],\quad&
    &\quad\hc_0 & &:= \big\{X \in L^2 \colon \E[X] = 0\big\}, \quad\!& \langle X &, Y \rangle_{\hc_0} & &:= \cov(X,Y),\\
    &\hat{\hc} & &:= \R^n,\quad\!& \langle x &, y \rangle_{\hat{\hc}} & &:= n^{-1}x^T y,\quad&
    &\quad\hat{\hc}_0 & &:= \big\{x \in \R^n \colon \bar{x} = 0\big\},\quad\!& \langle x &, y \rangle_{\hat{\hc}_0} & &:= \covh(x,y),
\end{alignat*}
where $\bar{x}$ denotes the empirical mean of $x$. The population
$R^2$-calculus for linear models as stated in the main text follows
from choosing the Hilbert space $(\hc_0,
\langle\cdot,\cdot\rangle_{\hc_0})$ in Lemma \ref{lem:HS-r2-corr} and
Proposition \ref{prop:HS-r2calc}. Likewise, we use $(\hat{\hc}_0,
\langle\cdot,\cdot\rangle_{\hat{\hc}_0})$ for the empirical
$R^2$-calculus. Since we choose the scaling $n^{-1}$ in the empirical
covariance, the estimators of covariance, variance and standard
deviation are not unbiased. To account for the loss of degrees of
freedom through estimation of the mean and potentially partialing out
a $p$-dimensional subspace, the factor $(n-p-1)^{-1}$ must be used. We
choose the scaling $n^{-1}$ instead to accord with the textbook
definition of the empirical $R^2$-value \citep[chap.~1]{mackinnon}. Besides, for a sufficiently large sample size $n$ the
difference will be negligible.

In the main text, we made the assumption that the random variables and
the observations are centred and thus are elements of $\hc_0$.
If this does not hold, we can
redefine the population $R^2$-value via the inner product
$\langle\cdot,\cdot\rangle_\hc$ 
as $R^2_{Y\sim X} := 1 - \E[(Y^{\perp X})^2]/\E[Y^2]$.
Similarly, we replace the inner product in the definition of partial
$R^2$-, $R$-, $f^2$- and $f$-values. This formulation contains the
definition of $R^2$-value in the main text as a special case because,
for centred random variables, $\langle \cdot, \cdot \rangle_\hc$ and the covariance are equal.

Furthermore, we can treat the constant $1$ as an additional covariate; it is easy to check that the relationship
\begin{equation*}
    R^2_{Y- \E[Y] \sim X - \E[X]} = R^2_{Y\sim X \vert 1}
\end{equation*}
holds in this case.
Hence, centering random variables is equivalent to partialing out the effect of the constant, and thus always observed, covariate. As our focus lies on the explanatory capability of the non-constant covariates, we always partial out 1 or equivalently center the observed variables. The same arguments also apply to the empirical $R^2$-value and centering the samples.

\section{Bias under Unmeasured Confounding}\label{sec:proofs-sec2}
Without loss of generality, we assume that all random variables/vectors
are centered; moreover, we only state and prove the population version
of the results. As explained in Section~\ref{app:r2-linear-models},
the sample and non-centered counterparts of the results and proofs
follow by the same arguments but choosing a different Hilbert space
and inner product.

\subsection{Single Unmeasured Confounder}\label{app:proofs-single}
\begin{proof}[\textbf{Proof of Proposition \ref{lem:objective}}]
Throughout this proof, all quantities partial out $(X,Z)$ which is indicated by either the subscript ``$\vert (X,Z)$'' or the superscript ``$\perp\!\! (X,Z)$''. In order to shorten the notation, we only indicate partialing out $(X,Z)$ in the estimands and drop the $(X,Z)$-dependence in the other quantities. First, we express the estimands $\beta_{Y\sim D\vert X,Z}$ and $\beta_{Y\sim D\vert X,Z,U}$ in terms of standard deviations and correlations and use the $R$- and $\sigma$-notation
\begin{align*}
    \beta_{Y\sim D\vert X,Z} - \beta_{Y\sim D\vert X,Z,U}
    &= \frac{ \corr(Y, D)\sd(Y)\sd(D)}{\sd(D)^2} - \frac{\corr(Y^{\perp U}, D^{\perp U}) \sd(Y^{\perp U})\sd(D^{\perp U})}{\sd(D^{\perp U})^2}\\
    &= R_{Y \sim D}\frac{\sigma_{Y}}{\sigma_{D}} - R_{Y \sim D\vert U}\frac{\sigma_{Y\sim U}}{\sigma_{D\sim U}}.
\end{align*}
Next, we extract the common factor $\sigma_{Y\sim D}/\sigma_{D}$ by applying rule [iii] of the $R^2$-calculus four times. We then rewrite the difference so that it
is expressed in terms of $R_{Y\sim U\vert D}$ instead of $R_{Y\sim D\vert U}$. To this end, we subsequently replace $R_{Y\sim D\vert U}$ and $R_{Y\sim U}$ via the recursion of partial
correlation formula [iv]. In summary, we get
\begin{align*}
    \beta_{Y\sim D\vert X,Z} &- \beta_{Y\sim D\vert X,Z,U}\\
        &\stackrel{\text{[iii]}}{=} \left[
    \frac{R_{Y \sim D}}{\sqrt{1-R_{Y \sim D}^2}} - R_{Y \sim D\vert U} \frac{\sqrt{1-R^2_{Y\sim U}}}{\sqrt{1-R^2_{Y\sim D}}\sqrt{1-R^2_{D\sim U}}}
    \right] \frac{\sigma_{Y\sim D}}{\sigma_{D}}\\
    &\stackrel{\text{[iv]}}{=} \left[
    f_{Y\sim D}-
    \frac{R_{Y\sim D}-R_{Y\sim U}\,R_{D\sim U}}{\sqrt{1-R^2_{Y\sim D}}\big(1-R^2_{D\sim U}\big)}
    \right] \frac{\sigma_{Y\sim D}}{\sigma_{D}}\\
    &\stackrel{\text{[iv]}}{=}\begin{aligned}[t]
        &\Bigg[ f_{Y\sim D} - \frac{1}{\sqrt{1-R^2_{Y\sim D}}\big(1-R^2_{D\sim U}\big)}\bigg(
        R_{Y\sim D}\\
        &-R_{D\sim U}
        \Big(
    \sqrt{1-R^2_{Y\sim D}}\sqrt{1-R^2_{D\sim U}}R_{Y\sim U\vert D}+R_{Y\sim D}\,R_{D\sim U}
    \Big)\bigg)
        \Bigg] \frac{\sigma_{Y\sim D}}{\sigma_{D}}
    \end{aligned}\\
    &= \left[f_{Y\sim D}\bigg(1-\frac{1}{1-R^2_{D\sim U}}\!+\!\frac{R^2_{D\sim U}}{1-R^2_{D\sim U}} \bigg)\! +\! f_{D\sim U}\,R_{Y\sim U\vert D} \right]\! \frac{\sigma_{Y\sim D}}{\sigma_{D}}\\
    &= f_{D\sim U}\,R_{Y\sim U\vert D} \frac{\sigma_{Y\sim D}}{\sigma_{D}} = R_{Y\sim U\vert X,Z,D}\,f_{D\sim U\vert X,Z} \,\frac{\sigma_{Y\sim D + X+ Z}}{\sigma_{D\sim X+Z}}.
\end{align*}
This establishes~\eqref{eq:objective-identification}. In order to prove the second equation~\eqref{eq:objective-iv}, we first relate the TSLS estimand $\beta_{Y\sim Z\vert X,D\sim Z\vert X}$ to the OLS estimand $\beta_{Y\sim D\vert X,Z}$. To this end, we express the former in terms of partial $R$-values, expand the expression and apply the three-variable identity with $Y\equiv Y, X \equiv Z, Z \equiv X$ and $W\equiv D$:
\begin{align*}
    \beta_{Y\sim Z\vert X,D\sim Z\vert X} &= \frac{R_{Y\sim Z\vert X}\,\sigma_{Y\sim X}}{R_{D\sim Z\vert X}\,\sigma_{D\sim X}} = \frac{1}{R_{D\sim Z\vert X}} \frac{\sigma_{Y\sim X}}{\sigma_{D\sim X}} \frac{\sqrt{1-R^2_{Y\sim Z\vert X}}}{\sqrt{1-R^2_{D\sim Z \vert X}}} f_{Y\sim Z\vert X} \sqrt{1-R^2_{D\sim Z\vert X}}\\
    &\stackrel{\text{(iii)}}{=}\! \frac{1}{R_{D\sim Z\vert X}} \frac{\sigma_{Y\sim X}}{\sigma_{D\sim X}} \frac{\sqrt{1-R^2_{Y\sim Z\vert X}}}{\sqrt{1-R^2_{D\sim Z \vert X}}} \bigg[\! R_{Y\sim D\vert X,Z}\,R_{D\sim Z\vert X} \!+\! f_{Y\sim Z\vert X,D} \sqrt{1-R^2_{Y\sim D\vert X,Z}}\bigg] \\
    &= \frac{R_{Y\sim D\vert X,Z}\,\sigma_{Y\sim X+Z}}{\sigma_{D\sim X+Z}} + \frac{f_{Y\sim Z \vert X, D}}{R_{D\sim Z \vert X}}\frac{\sigma_{Y\sim D+X+Z}}{\sigma_{D\sim X+Z}} \\
    &= \beta_{Y\sim D\vert X,Z} + \frac{f_{Y\sim Z \vert X, D}}{R_{D\sim Z \vert X}}\frac{\sigma_{Y\sim D+X+Z}}{\sigma_{D\sim X+Z}}.
\end{align*}
Expressing the difference of the TSLS estimand and the causal effect as telescoping sum and applying Proposition~\ref{lem:objective}, we find the desired result
\begin{align*}
    \beta_{Y\sim Z\vert X,D\sim Z\vert X} - \beta &=
    (\beta_{Y\sim Z\vert X,D\sim Z\vert X} - \beta_{Y\sim D\vert X,Z}) + (\beta_{Y\sim D\vert X,Z} - \beta)\\[1ex]
    &=\left[\frac{f_{Y\sim Z \vert X, D}}{R_{D\sim Z \vert X}} + R_{Y\sim U\vert X,Z,D}\,f_{D\sim U\vert X,Z}\right]\frac{\sigma_{Y\sim D+X+Z}}{\sigma_{D\sim X+Z}}.\hfill\qedhere
\end{align*}
\end{proof}

\subsection{Multiple Unmeasured Confounders}
\label{app:multiple-confounders}

\begin{proof}[\textbf{Proof of Proposition \ref{lem:objective-multi}}]
Analogously to the proof of Proposition~\ref{lem:objective}, we only indicate partialing out $(X,Z)$ in the estimands and drop the $(X,Z)$-dependence in the other quantities for ease of notation. We define the vector $\lambda$ as follows
\begin{equation*}
    \lambda = \var(U^{\perp D})^{-1}\cov(U^{\perp D},Y^{\perp D}).
\end{equation*}
It equals the regression coefficients of $U$ in the linear model $Y\sim D + X + Z + U$.
In order to reduce the number of dimensions of $U$, we introduce a new
random variable $U^* := \lambda^T U$. Since it captures all linear
influence of $U$ on $Y$, the estimands $\beta_{Y\sim D\vert X,U}$  and
$\beta_{Y\sim D\vert X,U^*}$ are equal. To formally prove this result, we let $A$ denote either $Y$ or $D$ and show that $A^{\perp U^*} = A^{\perp U}$. By definition of
$\lambda$ and some algebraic manipulations we derive
\begin{align*}
    A^{\perp U^*} &= A - (U^*)^T \var(U^*)^{-1} \cov(U^*, A)\\
    &=\begin{multlined}[t]
    A - U^T \var(U^{\perp D})^{-1} \cov(U^{\perp D},Y^{\perp D}) \Big[\var(U^{\perp D})^{-1} \cov(U^{\perp D},Y^{\perp D}) \Big]^{-1}
    \var(U)^{-1} 
    \\
    \times \Big[\var(U^{\perp D})^{-1}
    \cov(U^{\perp D},Y^{\perp D}) \Big]^{-T}
    \cov(U^{\perp D}, Y^{\perp D})^T \var(U^{\perp D})^{-T} \cov(U,A)
    \end{multlined}\\
    &= A - U^T \var(U)^{-1} \cov(U,A) = A^{\perp U}.
\end{align*}
Choosing $Y$ and $D$ for $A$, we get
\begin{equation*}
    \beta_{Y\sim D\vert X,U^*} = \frac{\cov(Y^{\perp U^*}, D^{\perp U^*})}{\var(D^{\perp U^*})}=
    \frac{\cov(Y^{\perp U}, D^{\perp U})}{\var(D^{\perp U})} = \beta_{Y\sim D\vert X,U}.
\end{equation*}
Since $U^*$ is one-dimensional, we can use Proposition~\ref{lem:objective} to find a precise characterization for the difference between the OLS estimand that does not and does adjust for $U$:
\begin{equation}\label{eq:app-proof-multi}
    \beta_{Y\sim D\vert X,Z} - \beta_{Y\sim D\vert X,Z,U} = \beta_{Y\sim D\vert X,Z} - \beta_{Y\sim D\vert X,Z,U^*} = R_{Y\sim U^* \vert D}\, f_{D\sim U^*} \frac{\sigma_{Y\sim D}}{\sigma_D}.
\end{equation}
Moreover, the explanatory capabilities of $U$ and $U^*$ for $Y$ are
identical. According to Lemma \ref{lem:proj-HS} (iii), we infer
\begin{equation*}
    Y^{\perp D,U} = Q_{(D,U)} Y = Q_{D^{\perp U}} Q_U Y = Q_{D^{\perp U^*}} Y^{\perp U^*} = Y^{\perp D,U^*},
\end{equation*}
which yields
\begin{equation*}
    R^2_{Y\sim U\vert D} = 1- \frac{\var(Y^{\perp D,U})}{\var(Y^{\perp D})} = 1- \frac{\var(Y^{\perp D,U^*})}{\var(Y^{\perp D})} = R^2_{Y\sim U^*\vert D}.
\end{equation*}
The new random variable $U^*$ fully captures the effect of $U$ on $Y$
but does not capture the entire effect of $U$ on $D$ due to the
reduced dimension, i.e.\ $R^2_{D\sim U} \geq
R^2_{D\sim U^*}$. To prove this result, we rewrite $D^{\perp U}$ using
Lemma \ref{lem:proj-HS} (iii) as follows
\begin{equation*}
    D^{\perp U} = Q_{ P_{U^*} W + Q_{U^*} U } D = Q_{ (U^* ,Q_{U^*} U) } D = Q_{U^{\perp U^*}} D^{\perp U^*}.
\end{equation*}
Based on this equation, Lemma \ref{lem:proj-HS} (iv) yields the inequality
\begin{equation*}
    \var(D^{\perp U}) \leq \var(Q_{U^{\perp U^*}} D^{\perp U^*}) + \var(P_{U^{\perp U^*}} D^{\perp U^*}) = \var(D^{\perp U^*}),
\end{equation*}
which implies
\begin{equation*}
    R^2_{D\sim U} = 1 - \frac{\var(D^{\perp U})}{ \var(D)} \geq 1 - \frac{\var(D^{\perp U^*})}{ \var(D)} = R^2_{D\sim U^*}.
\end{equation*}
Returning to (\ref{eq:app-proof-multi}), we use the equality and
inequality derived for the $R^2$-values concerning $U^*\rightarrow Y$
and $U^* \rightarrow D$, respectively. Since $f^2$ is a monotone
transformation of $R^2$, we have
\begin{equation*}
    \lvert \beta_{Y\sim D\vert X,Z} - \beta_{Y\sim D\vert X,Z,U}\rvert^2 \leq R^2_{Y\sim U \vert D,X,Z}\,f^2_{D \sim U \vert X,Z}
    \frac{\sigma^2_{Y\sim D+X+Z}}{\sigma^2_{D\sim X+Z}}.\hfill \qedhere
\end{equation*}
\end{proof}

In presence of multiple unmeasured confounders, finding a precise and
interpretable charaterization of the difference $\beta_{Y\sim D\vert X, Z} - \beta_{Y\sim D\vert X,Z,U}$ is a challenging task. For instance, when $U$ is $l$-dimensional, one could repeatedly apply Proposition~\ref{lem:objective} to derive the following telescoping series:
\begin{multline*}
\beta_{Y\sim D\vert X,Z} - \beta_{Y\sim D\vert X,Z,U}
= \sum_{j=1}^l \beta_{Y\sim D\vert X,Z,U_{[j-1]}} - \beta_{Y\sim D\vert X,Z,U_{[j]}}\\
  = \sum_{j=1}^{l} R_{Y\sim U_{j}\vert X,Z,D,U_{[j-1]}}\,f_{D\sim U_{j}\vert X,Z,U_{[j-1]}}
  \sqrt{\frac{1-R^2_{Y\sim U_{[j-1]}\vert X,Z,D}}{1-R^2_{D \sim U_{[j-1]}\vert X,Z}}}
  \,\frac{\sigma_{Y\sim X+Z+D}}{\sigma_{D\sim X+Z}},
\end{multline*}
where $[j] := \{1,\ldots,j\}$ and $[0] := \emptyset$. However, this characterization requires $2l$ sensitivity parameters that need to be bounded; moreover, the sensitivity parameters are not symmetric in the set of partialed out variables which impedes their interpretation.

This issue can be circumvented under the additional assumption that the components of $U$ are partially uncorrelated given $(X,Z)$, i.e.\ $R_{U_i \sim U_j\vert X,Z}=0$ for all $1\leq i < j\leq l$. This is justified by the following result which is closely related to Wright's path analysis. Our proof, however, only relies on the algebraic relationships of the $R^2$-calculus and does not consult the underlying DAG.

\begin{lemma}\label{lem:independent-expansion}
Assume the setting of Proposition~\ref{lem:objective-multi} and further suppose that all components of $U$ are partially uncorrelated given $(X,Z)$. Then,
\begin{equation}\label{eq:independent-expansion}
    \beta_{Y\sim D\vert X,Z} - \beta_{Y\sim D\vert X,Z,U}
    = \sum_{j=1}^l \beta_{Y\sim U_j\vert X,Z,D,U_{-j}}\,\beta_{U_j\sim D\vert X,Z},
\end{equation}
where $U_{-j} = (U_1,\ldots,U_{j-1},U_{j+1},\ldots, U_l)$.
\end{lemma}
\begin{proof}
For ease of notation, we only indicate partialing out $(X,Z)$ in the estimands and drop the $(X,Z)$-dependence in the other quantities.

Due to the assumption of partial uncorrelatedness of the components of $U$ and Lemma \ref{lem:proj-HS}~(ii), we can decompose $Y$ as follows
\begin{equation*}
    Y = Y^{\perp D,U} + P_{D^{\perp U}} Y + \sum_{j=1}^l P_{U_j} Y.
\end{equation*}
Plugging this relationship into the definition of $\beta_{Y\sim D\vert X,Z}$, using linearity of the covariance and the formula for projections on a one-dimensional space \eqref{eq:1d-proj} yields
\begin{align*}
    \beta_{Y\sim D\vert X,Z} &= \frac{\cov(Y,D)}{\var(D)} = 0 +
    \frac{\cov(P_{D^{\perp U}} Y, D)}{\var(D)} + \sum_{j=1}^l \frac{\cov(P_{U_j} Y, D)}{\var(D)}\\
    &= \frac{1}{\var(D)}\left[\cov\bigg(\frac{\cov(Y,D^{\perp U})}{\var(D^{\perp U})}\, D^{\perp U}, D\bigg) + \sum_{j=1}^l \cov\bigg(\frac{\cov(Y,U_j)}{\var(U_j)}\,U_j,D\bigg)\right]\\
    &= \frac{\cov(Y^{\perp U},D^{\perp U})}{\var(D)} + \sum_{j=1}^l \frac{\cov(Y,U_j)}{\var(U_j)} \frac{\cov(D,U_j)}{\var(D)}\\
    &= \beta_{Y\sim D\vert X,Z,U} \frac{\sigma^2_{D\sim U}}{\sigma^2_D} 
    + \sum_{j=1}^l R_{Y\sim U_j} \frac{\sigma_Y}{\sigma_{U_j}}\, \beta_{U_j\sim D}.
\end{align*}
By applying the definition of the $R^2$-value, we derive
\begin{equation*}
\beta_{Y\sim D\vert X,Z} - \beta_{Y\sim D\vert X,Z,U} =
-\beta_{Y\sim D\vert X,Z,U}\, R^2_{D\sim U} + \sum_{j=1}^l R_{Y\sim U_j} \frac{\sigma_Y}{\sigma_{U_j}} \, \beta_{U_j\sim D}.
\end{equation*}
Next, we use rule [ii] of the $R^2$-calculus -- independent additivity -- on $R^2_{D\sim U}$ and rewrite $\beta_{Y\sim D\vert X,Z,U}$ in terms of $R$-values and $\sigma$-values, i.e.\ standard deviations:
\begin{align*}
    \beta_{Y\sim D\vert X,Z} - \beta_{Y\sim D\vert X,Z,U}
    &\stackrel{\text{[ii]}}{=} 
    - R_{Y\sim D\vert U} \frac{\sigma_{Y\sim U}}{\sigma_{D\sim U}}
    \sum_{j=1}^l R^2_{D\sim U_j} + \sum_{j=1}^l R_{Y\sim U_j}  \frac{\sigma_Y}{\sigma_{U_j}} \, \beta_{U_j\sim D}\\
    &= \sum_{j=1}^l \beta_{U_j\sim D}\bigg[R_{Y\sim U_j}  \frac{\sigma_Y}{\sigma_{U_j}} - \frac{\sigma_D}{R_{D\sim U_j}\sigma_{U_j}} R_{Y\sim D\vert U}\frac{\sigma_{Y\sim U}}{\sigma_{D\sim U}}R^2_{D\sim U_j}\bigg]
\end{align*}
In order to extract the factor $\sigma_{Y\sim D + U_{-j}}/\sigma_{U_j\sim D+U_{-j}}$, we apply rule [iii] -- decomposition of unexplained variance -- six times and arrive at
\begin{multline}
\label{eq:indep-exp-intermediate}
    \beta_{Y\sim D\vert X,Z} - \beta_{Y\sim D\vert X,Z,U}
    \stackrel{\text{[iii]}}{=} \sum_{j=1}^l \beta_{U_j\sim D}\, \frac{\sigma_{Y\sim D+U_{-j}}}{\sigma_{U_j\sim D+U_{-j}}}
    \Bigg[R_{Y\sim U_j} \sqrt{\frac{1-R^2_{U_j\sim D + U_{-j}}}{1-R^2_{Y\sim D+U_{-j}}}}\\
    - R_{Y\sim D\vert U}\, R_{D\sim U_j} \sqrt{\frac{(1-R^2_{U_j\sim D + U_{-j}})(1-R^2_{Y\sim U_j\vert U_{-j}})}{(1-R^2_{D\sim U})(1-R^2_{Y\sim D\vert U_{-j}})}}\,\Bigg].
\end{multline}
We concentrate on the term in brackets, denoted by $T_j$. Invoking rule [v] -- reduction of partial correlation -- and the (conditional) independence assumption, we infer
\begin{gather*}
    R_{Y\sim U_j} \stackrel{\text{[v]}}{=} R_{Y\sim U_j\vert U_{-j}} \sqrt{1-R^2_{Y\sim U_{-j}}},\quad
    R_{D\sim U_j} \stackrel{\text{[v]}}{=} R_{D\sim U_j\vert U_{-j}} \sqrt{1-R^2_{D\sim U_{-j}}},\\
    R^2_{U_j\sim D + U_{-j}} = R^2_{U_j\sim D\vert U_{-j}}.
\end{gather*}
We insert these relationships into the expression of $T_j$ and simplify it via rule [iii]. Then, we apply rule [iv] -- recursion of partial correlation -- on $R_{Y\sim D\vert U}$ and simplify the resulting expression
{\allowdisplaybreaks
\begin{align*}
    T_j &=
    \begin{multlined}[t]
    R_{Y\sim U_j\vert U_{-j}} \sqrt{1-R^2_{U_j\sim D\vert U_{-j}}} \sqrt{\frac{1-R^2_{Y\sim U_{-j}}}{1-R^2_{Y\sim D+U_{-j}}}}\\
    - R_{Y\sim D\vert U}\, R_{D\sim U_j\vert U_{-j}}\sqrt{1-R^2_{D\sim U_{-j}}} \sqrt{\frac{(1-R^2_{U_j\sim D \vert U_{-j}})(1-R^2_{Y\sim U_j\vert U_{-j}})}{(1-R^2_{D\sim U})(1-R^2_{Y\sim D\vert U_{-j}})}}
    \end{multlined}\\
    &\stackrel{\text{[iii]}}{=} R_{Y\sim U_j\vert U_{-j}} \sqrt{\frac{1-R^2_{D\sim U_{j}\vert U_{-j}}}{1-R^2_{Y\sim D\vert U_{-j}}}} - R_{Y\sim D\vert U} R_{D\sim U_j\vert U_{-j}} \sqrt{\frac{1-R^2_{Y\sim U_j\vert U_{-j}}}{1-R^2_{Y\sim D\vert U_{-j}}}}\\
    &\stackrel{\text{[iv]}}{=} \! R_{Y\sim U_j\vert U_{-j}} \sqrt{\frac{1-R^2_{D\sim U_{j}\vert U_{-j}}}{1-R^2_{Y\sim D\vert U_{-j}}}} \!-\! R_{D\sim U_j\vert U_{-j}} \frac{R_{Y\sim D\vert U_{-j}}\!\! -\! R_{Y\sim U_j\vert U_{-j}}R_{D\sim U_j\vert U_{-j}}}{\sqrt{1-R^2_{Y\sim D\vert U_{-j}}} \sqrt{1-R^2_{D\sim U_j\vert U_{-j}}}}\\
    &= \frac{R_{Y\sim U_j\vert U_{-j}} (1-R^2_{D\sim U_{j}\vert U_{-j}}) - R_{Y\sim D\vert U_{-j}}R_{D\sim U_j\vert U_{-j}} - R_{Y\sim U_j\vert U_{-j}}R^2_{D\sim U_j\vert U_{-j}}}{\sqrt{1-R^2_{Y\sim D\vert U_{-j}}} \sqrt{1-R^2_{D\sim U_j\vert U_{-j}}}} \\
    &= \frac{R_{Y\sim U_j\vert U_{-j}} - R_{Y\sim D\vert U_{-j}}\, R_{U_j\sim D\vert U_{-j}} }{\sqrt{1-R^2_{Y\sim D\vert U_{-j}}} \sqrt{1-R^2_{U_j\sim D\vert U_{-j}}}}
    = R_{Y\sim U_j\vert D, U_{-j}}.
    \end{align*}
}
Returning to equation \eqref{eq:indep-exp-intermediate}, we plug in $T_j = R_{Y\sim U_j\vert D,U_{-j}}$ and thus finish the proof
\begin{equation*}
     \beta_{Y\sim D\vert X,Z} \!-\! \beta_{Y\sim D\vert X,Z,U} \!=\! \sum_{j=1}^l \beta_{U_j\sim D}\, \frac{\sigma_{Y\sim D+U_{-j}}}{\sigma_{U_j\sim D+U_{-j}}} R_{Y\sim U_j\vert D,U_{-j}}
     \!=\! \sum_{j=1}^l \beta_{Y\sim U_j\vert X,Z D,U_{-j}}\, \beta_{U_j\sim X,Z,D}.
     \hfill\qedhere
\end{equation*}
\end{proof}

Lemma \ref{lem:independent-expansion} helps us express the bias of the OLS estimand in terms of partial $R$-values which serve as sensitivity parameters. Whether these are intuitive, depends on the causal structure of the underlying DAG. In the case of two partially uncorrelated unmeasured variables $U_1$ and $U_2$ which confound or mediate $\beta$ -- the direct effect of $D$ on $Y$ --, the sensitivity parameters $(R_{D\sim U_1\vert X,Z}, R_{D\sim U_2\vert X,Z})$ and $(R_{Y\sim U_1\vert X,Z,D,U_2}, R_{Y\sim U_2\vert X,Z,D, U_1})$ are indeed intuitive. The former tuple targets the dependence between $D$ and $U$, the latter tuple focuses on the direct effects of $U$ on $Y$ regressing out the remaining variables.


The following proposition demonstrates how Lemma~\ref{lem:independent-expansion} can be applied. We identify the bias of the OLS estimand in terms of the intuitive sensitivity parameters.
\begin{proposition}
In the setting of Proposition~\ref{lem:objective-multi}, let $U = (U_1,U_2)$ be a two-dimensional random vector and assume $R_{U_1\sim U_2\vert X,Z}=0$. Then,
\begin{equation*}
    \beta = \beta_{Y\sim D\vert X,Z} - \sum_{j=1}^2
  \frac{R_j\,f_j}{\sqrt{1-f^2_j\, f^2_{-j}+\left(R_{-j}\sqrt{(1-R^2_j)\big/(1-R^2_{-j})}
  -R_j\,f_j\, f_{-j}\right)^{2}}} \frac{\sigma_{Y\sim X + Z + D}}{\sigma_{D\sim X + Z}},
\end{equation*}
where $R_j$ and $f_j$ abbreviate $R_{Y\sim U_{j}\vert X,Z,D,U_{-j}}\!$ and $f_{D\sim U_j\vert X,Z}$, respectively, for ${j\!\in\!\{1,2\}}$.
\end{proposition}

\begin{proof}
Analogously to previous proofs, we only indicate partialing out $(X,Z)$ in the estimands. In order to rewrite and simplify the the difference of the OLS estimand and the causal effect, we invoke Lemma~\ref{lem:independent-expansion} and the rule on decomposition of unexplained variance
\begin{align}
\beta_{Y\sim D\vert X,Z} - \beta_{Y\sim D\vert X,Z,U} &= \sum_{j=1}^2 R_{Y\sim U_j\vert D,U_{-j}} \frac{\sigma_{Y \sim D + U_{-j}}}{\sigma_{U_j \sim D + U_{-j}}}
R_{D \sim U_j} 
\frac{\sigma_{U_j}}{\sigma_D}
\nonumber\\
&\stackrel{\text{[iii]}}{=} \sum_{j=1}^2 R_{Y\sim U_j\vert D,U_{-j}} R_{D \sim U_j}\sqrt{\frac{1-R^2_{Y\sim U_{-j}\vert D}}{1-R^2_{U_j \sim D + U_{-j}}}}
\frac{\sigma_{Y\sim D}}{\sigma_D}.\label{eq:app-proof-2conf1}
\end{align}
Leveraging the assumption $R_{U_1\sim U_2} = 0$, we rewrite $R_{U_j \sim U_{-j}\vert D}$ via the recursive partial correlation formula [iv]; moreover, we use the decomposition of unexplained variance [iii] on $1-R^2_{U_i \sim D + U_{-i}}$ as follows
{\allowdisplaybreaks
\begin{gather}
 R_{U_j \sim U_{-j}\vert D} \stackrel{\text{[iv]}}{=} \frac{R_{U_j\sim U_{-j}} - R_{U_j\sim D}\, R_{U_{-j}\sim D}}{\sqrt{1-R^2_{U_j \sim D}}\sqrt{1-R^2_{U_{-j}\sim D}}} = -f_{D\sim U_j}\,f_{D\sim U_{-j}},\label{eq:app-proof-2conf2}\\
 1-R^2_{U_i \sim D + U_{-i}} \stackrel{\text{[iii]}}{=} (1-R^2_{D\sim U_i})(1-R^2_{U_i\sim U_{-i}\vert D}).\nonumber
\end{gather}
}
Inserting these relationships into \eqref{eq:app-proof-2conf1}, we find
\begin{equation}\label{eq:app-proof-2conf3}
\begin{aligned}
    \beta_{Y\sim D\vert X,Z} - \beta_{Y\sim D\vert X,Z,U} &= \sum_{j=1}^2 R_{Y\sim U_j\vert D,U_{-j}}\,f_{D \sim U_j}\sqrt{\frac{1-R^2_{Y\sim U_{-j}\vert D}}{1-f^2_{D\sim U_1}f^2_{D\sim U_2}}}
\frac{\sigma_{Y\sim D}}{\sigma_D}\\
    &= \sum_{j=1}^2 R_j\,f_j \sqrt{\frac{1-R^2_{Y\sim U_{-j}\vert D}}{1-f^2_1 f^2_2}}
\frac{\sigma_{Y\sim D}}{\sigma_D}.
\end{aligned}
\end{equation}
Lastly, we aim to express $\sqrt{1-R^2_{Y\sim U_{-j}\vert D}}$ in terms of the other sensitivity parameters. To this end, we use the three-variable identity [vi] with $Y \equiv Y$, $X \equiv U_{-j}$, $W \equiv U_{j}$ and $Z \equiv D$, where we replace $R_{U_j\sim U_{-j}\vert D}$ according to \eqref{eq:app-proof-2conf2}:
\begin{equation*}
\frac{R_{-j}}{\sqrt{1-R^2_{-j}}}\sqrt{1-R^2_j}-f_1 f_2\, R_j \stackrel{\text{[vi]}}{=} f_{Y\sim U_{-j}\vert D}\sqrt{1-f^2_1 f^2_2}.
\end{equation*}
By definition, the identity $\sqrt{1-R^2} = 1/\sqrt{1+f^2}$ holds true for any (partial) $R^2$ and its corresponding $f^2$. Thus, we get
\begin{equation*}
    \sqrt{1-R^2_{Y\sim U_{-j}\vert D}} =
    \left[1+ \left(\frac{R_{-j}}{\sqrt{1-R^2_{-j}}}\sqrt{1-R^2_j}-f_1 f_2\, R_j\right)^2\!\!\bigg/\Big(1-f_1^2\, f_2^2\Big) \right]^{-1/2}.
\end{equation*}
Substituting the $\sqrt{1-R^2_{Y\sim U_{-j}\vert D}}$ term in (\ref{eq:app-proof-2conf3}) for the expression above proves the form of the second summand that was required.
\hfill
\end{proof}

\section{Details on the Sensitivity Model}\label{app:sens-model}

\subsection{Relation between Regression and IV Sensitivity Parameters}\label{sec:reg-iv-relationship}

In order to derive the equality constraint \eqref{eq:reg-iv-connection}, we apply the three-variable identity from the $R^2$-calculus twice: First, we set $Y \equiv
Y$, $X \equiv Z$, $W \equiv U$ and $Z \equiv (X,D)$; second, $Y \equiv U$, $X \equiv Z$, $W \equiv D$ and $Z \equiv X$. This yields

\begin{equation*}
    \begin{aligned}
        f_{Y\sim Z \vert X, U, D}\sqrt{1-R^{2}_{Y\sim U \vert X,Z,D}}
  &= f_{Y\sim Z\vert X,D}\sqrt{1-R^{2}_{Z\sim U \vert X,D}} - R_{Y\sim
    U\vert X, Z, D} R_{Z\sim U\vert X,D},\\
    f_{Z\sim U \vert X,D} \sqrt{1-R^{2}_{D\sim U\vert X,Z}}
  &= f_{Z \sim U\vert X} \sqrt{1-R^{2}_{D\sim Z\vert X}} - R_{D\sim Z\vert X} R_{D\sim U\vert X, Z}.
    \end{aligned}
\end{equation*}

This constraint suffices to recover point-identification of $\beta$ if $Z$ is a valid instrument. Indeed, if we set $R_{Z\sim U\vert X} = 0$ and $R_{Y\sim Z\vert X, U, D} = 0$ in the equations above and simplify them, we get the relationship
\begin{equation*}
    f_{D\sim U\vert X, Z}\,R_{Y\sim U\vert X, Z, D} = - \frac{f_{Y\sim Z \vert X, D}}{R_{D\sim Z \vert X}}.
\end{equation*}
Plugging this result into equation \eqref{eq:objective-iv}, we directly obtain $\beta = \beta_{D\sim Z\vert X, Y\sim Z\vert X}$.


\subsection{Optimization Constraints}\label{sec:optimization-contraints}

\begin{table}[t!]
    \centering
    \caption{Equality and inequality constraints stemming from comparative bounds.}
    \label{tab:comparative-constraints}
    \vspace{0.3cm}
    \begin{adjustbox}{max width=1.1\textwidth,center}
\begin{tabular}{ccr}
    \toprule
    Edge & Optimization constraint &\\
    \midrule
    \circled{1}   &  $\displaystyle R^2_{D \sim U \vert X,Z}
  \leq b_{U\! D} \frac{R^2_{D\sim \xd_J \vert \xt,\xd_I,Z}}{1-R^2_{D\sim \xd_{I^c}\vert \xt,\xd_I,Z}}$ & \addtag[eq:reg-cond-d2] \\[6ex]
    \circled{2}  & $\displaystyle R_{Y\sim U \vert X,Z,D} = \frac{R_{Y\sim U \vert X,Z} - R_{Y\sim D\vert X,Z}\, R_{D\sim U \vert X,Z}}{\sqrt{1-R^{2}_{Y\sim D\vert X,Z}}\sqrt{1-R^{2}_{D \sim U \vert X,Z}}}$ & \addtag[eq:reg-con-y22]\\[6ex]
    & $\displaystyle R^2_{Y\sim U \vert X,Z} \leq b_{UY}\frac{R^2_{Y\sim \xd_J\vert \xt, \xd_I,Z}}{1-R^2_{Y\sim \xd_{I^c}\vert \xt,\xd_I,Z}}$ & \addtag[eq:reg-con-y21] \\[6ex]
    & $\begin{multlined}
R_{Y\sim U \vert X,Z} = \frac{1}{\sqrt{1-R^{2}_{Y\sim \xd_{I^{c}}\vert \xt,\xd_{I},Z}}} \bigg[R_{Y\sim D
   \vert \xt,\xd_{I},Z}R_{D\sim U \vert X,Z}\sqrt{1-R^{2}_{D\sim \xd_{I^{c}}\vert \xt,\xd_{I},Z}} \\[-4ex]
+ R_{Y\sim U
   \vert \xt,\xd_{I},Z,D}\sqrt{1-R^{2}_{Y\sim D\vert
   \xt,\xd_{I},Z}}\sqrt{1-R^2_{D\sim U\vert X, Z}(1-R^{2}_{D\sim \xd_{I^{c}}\vert\xt, \xd_{I},Z})}\bigg]
  \end{multlined}$ & \addtag[eq:reg-con-y32] \\[9ex]
    \circled{3}  & $\displaystyle R^2_{Z\sim U\vert X} \leq b_{U\!Z}\, R^2_{Z\sim \xd_j\vert \xt,\xd_{-j}} \frac{1-R^2_{Z\sim \xd_j\vert \xt,\xd_{-j}}}{1- b_{U\!Z}\, R^4_{Z\sim \xd_j\vert \xt,\xd_{-j}}}$ & \addtag[eq:iv-con-z2]\\[6ex]
    \circled{4} & $\displaystyle \begin{multlined}
    f_{Y\sim \xd_j\vert \xt, \xd_{-j}, Z, U, D} \sqrt{1-R^2_{Y\sim U\vert X, Z, D}}
    = \bigg[f_{Y\sim \xd_j\vert \xt, \xd_{-j}, Z, D}\sqrt{1-R^2_{D\sim U\vert X,Z}} \\[-4ex]
    +R_{Y\sim U\vert X, Z, D}\,R_{D\sim \xd_j\vert \xt,\xd_{-j},Z}\,R_{D\sim U\vert X,Z} \bigg]\Big/\sqrt{1-R^2_{D\sim U\vert X,Z}(1-R^2_{D\sim \xd_j\vert \xt,\xd_{-j},Z})}
\end{multlined}$ & \addtag[eq:iv-con-ex22]\\[3ex]
    \bottomrule
    \end{tabular}
    \end{adjustbox}
\end{table}

In Table \ref{tab:constraints}, we list different sensitivity bounds that practitioners can specify. For direct constraints, practitioners need to specify a lower and an upper bound -- $B^l$ and $B^u$ -- on the values of the respective sensitivity parameters. For comparative constraints, they choose a comparison random variable (or group of random variables) and specify a positive number $b$ that relates the explanatory capability of $U$ to the chosen variable(s). While we can immediately add direct sensitivity bounds to the constraints of the optimization problem, we need to reformulate comparative bounds and possibly add equality constraints to relate them to the regression and IV sensitivity parameters. Table~\ref{tab:comparative-constraints} lists these optimization constraints emerging from the comparative bounds in Table~\ref{tab:constraints}. The derivations of the respective equalities and inequalities along with further explanations are given in the the following subsection.

If only constraints on $U\rightarrow D$ and $U\rightarrow Y$ are given, the user-specified sensitivity bounds in  Table~\ref{tab:constraints} and (additional) optimization constraints in Table~\ref{tab:comparative-constraints} can all be reformulated as bounds on the two sensitivity parameters $R_{Y\sim U\vert X,Z,D}$ and $R_{D\sim U\vert X,Z}$; this can be seen by substituting the equality into the inequality constraints. Hence, the resulting sensitivity model has two degrees of freedom. Similarly, if additionally constraints on $U\leftrightarrow U$ and/or $Z\rightarrow Y$ are specified, the sensitivity model can be parameterized in terms of $R_{Y\sim U\vert X,Z,D}$, $R_{D\sim U\vert X,Z}$ and $R_{Z\sim U\vert X}$ and is thus three-dimensional. 

In practice, we work with the formulation involving equality constraints and additional unestimable parameters, such as $R_{Y\sim U\vert X,Z}$, to keep the the algebraic expressions tractable. In spirit, this is similar to the use of slack variables in optimization.

\subsection{Constraints implied by Comparative Bounds}\label{sec:constraints-comparative}
In this subsection, we derive the optimization constraints implied by comparative sensitivity bounds which are listed in Table~\ref{tab:comparative-constraints} and comment on them. We recall the notation of the main article: We denote $[\pd]:=\{1,\ldots,\pd\}$; for $I \subseteq [\pd]$ and $\xd\in \R^{\dot{p}}$, define $\xd_I := (\xd_i)_{i\in I}$ and $I^c := [\pd]\setminus I$; lastly, $\xd_{-j} := \xd_{\{j\}^c}$ for any $j \in [\pd]$.

\subsubsection{Bound on $U \rightarrow D$}
The comparative bound
\begin{equation*}
  R^{2}_{D \sim U\vert \xt,\xd_{I},Z} \leq b_{U\!D} R^{2}_{D \sim
    \xd_{J}\vert \xt,\xd_{I},Z},~I \subset [\pd],~J \subseteq
  I^c,b_{U\! D} \geq 0
\end{equation*}
means that the unmeasured confounder $U$ can explain at most $b_{U\! D}$ times as much variance of $D$ as $\xd_J$ does, after accounting for the effect of $(\xt, \xd_I,Z)$ on $D$. For practical purposes, a good choice of the comparison sets is $J = \{j\}$ and $I = J^c$. We can relate $R_{D \sim U \vert \xt, \xd_I, Z}$ in the last bound to $R_{D \sim U \vert X,Z}$ via the $R^2$-calculus. By using $R^2_{U\sim \xd_{I^c}\vert \xt, \xd_I,Z}=0$, which follows from the assumption in \eqref{eq:xd-cond-indep}, and applying the reduction
of partial correlation with $Y \equiv U$, $X\equiv D$, $Z \equiv (\xt, \xd_I,Z)$ and $W \equiv \xd_{I^c}$, we obtain the optimization constraint \eqref{eq:reg-cond-d2}:
\begin{equation*}
  R^2_{D \sim U \vert X,Z}
  \overset{[\text{v}]}{=} \frac{R^2_{D\sim U \vert \xt, \xd_I,Z}}{1-R^2_{D\sim \xd_{I^c}\vert \xt,\xd_I,Z}}
  \leq b_{U\! D} \frac{R^2_{D\sim \xd_J \vert \xt,\xd_I,Z}}{1-R^2_{D\sim \xd_{I^c}\vert \xt,\xd_I,Z}}.
\end{equation*}

\subsubsection{Bound on $U\rightarrow Y$}
For the effect of $U$ on $Y$, we consider two types of comparative bounds depending on whether $D$ is regressed out:
\begin{equation*}
  R^2_{Y\sim U\vert \xt, \xd_I,Z} \leq b_{UY} R^2_{Y\sim \xd_J\vert \xt, \xd_I,Z},\qquad
  R^2_{Y\sim U\vert \xt, \xd_I, Z, D} \leq b_{UY} R^2_{Y\sim \xd_J\vert \xt, \xd_I,Z, D},
\end{equation*}
where $I \subset [\pd],~J \subseteq I^c,b_{U\! Y} \geq 0$. When comparing the explanatory capability of the variables $U$ and $\xd_J$, it is natural to regress out all other variables. However, regressing out $D$, a potential common child of $X$ and $U$, may introduce dependence between $U$ and $Y$; this is essentially the point made by \citet{hernan1999} in their criticism of \citet{lin}. Thus, we consider comparative bounds without and with $D$.

For the first comparative bound, we may apply rule [v] as in the previous subsection and obtain the optimization constraint \eqref{eq:reg-con-y21}:
\begin{equation*}
  R^2_{Y\sim U \vert X,Z} \stackrel{\text{[v]}}{=} \frac{R^2_{Y\sim U \vert \xt,\xd_I,Z}}{1-R^2_{Y\sim \xd_{I^c}\vert \xt,\xd_I,Z}} \leq b_{UY}\frac{R^2_{Y\sim \xd_J\vert \xt, \xd_I,Z}}{1-R^2_{Y\sim \xd_{I^c}\vert \xt,\xd_I,Z}}.
\end{equation*}
Since this inequality is not a constraint on the sensitivity parameter $R_{Y\sim U\vert X,Z,D}$, we require another constraint that links the above to the regression sensitivity parameters. To this end, we use rule [iv] of the $R^2$-calculus which yields constraint \eqref{eq:reg-con-y22}:
\begin{equation*}
  R_{Y\sim U \vert X,Z,D} \stackrel{\text{[iv]}}{=} \frac{R_{Y\sim U \vert X,Z} - R_{Y\sim D\vert X,Z}\, R_{D\sim U \vert X,Z}}{\sqrt{1-R^{2}_{Y\sim D\vert X,Z}}\sqrt{1-R^{2}_{D \sim U \vert X,Z}}}.
\end{equation*}
Hence, the first type of comparative bound implies the inequality and equality constraint above in the optimization problem.

The second type of comparative bounds, which additionally partials out $D$, involves the unknown quanitity $R_{Y\sim U\vert \xt, \xd_I, Z,D}$. In order to link it -- and thus the user-specified bound -- to the regression sensitivity parameters, we employ rule [v] -- reduction of partial correlation -- and the recursive partial correlation formula [iv] for $R_{Y\sim U \vert \xt, \xd_I,Z,D}$ and infer
\begin{multline*}
\!R_{Y\sim U \vert X,Z} \stackrel{\text{[v]}}{=}\!\!
    \frac{R_{Y\sim U \vert\xt,\xd_{I},Z}}{\sqrt{1\!-\!R^{2}_{Y\sim \xd_{I^{c}}\vert \xt,\xd_{I},Z}}}
    \!\!\stackrel{\text{[iv]}}{=}\!\! \frac{1}{\sqrt{1\!-\!R^{2}_{Y\sim \xd_{I^{c}}\vert \xt,\xd_{I},Z}}}\!
\Bigg[\!R_{Y\sim D\vert
   \xt,\xd_{I},Z}\,R_{D \sim U\vert \xt,\xd_{I},Z}\\
   + R_{Y\sim U
   \vert \xt,\xd_{I},Z,D}\sqrt{1-R^{2}_{Y\sim D\vert
   \xt,\xd_{I},Z}}\sqrt{1-R^{2}_{D \sim U\vert \xt,\xd_{I},Z}}
\Bigg].
\end{multline*}
This equation contains the unknown quantity $R_{D\sim U \vert \xt, \xd_I,Z}$ which can be expressed in terms of $R_{D\sim U\vert X,Z}$ via rule [v]
\begin{equation*}
    R_{D\sim U \vert X,Z} \stackrel{\text{[v]}}{=} \frac{R_{D\sim U \vert \xt, \xd_I,Z}}{\sqrt{1-R^2_{D\sim \xd_{I^c}\vert \xt,\xd_I,Z}}}.
\end{equation*}
Plugging this relationship into the equation above, we arrive at the optimization constraint \eqref{eq:reg-con-y32}:
\begin{multline*}
R_{Y\sim U\vert X,Z} = \frac{1}{\sqrt{1-R^{2}_{Y\sim \xd_{I^{c}}\vert \xt,\xd_{I},Z}}}
\Bigg[R_{Y\sim D\vert \xt, \xd_I,Z} \sqrt{1-R^2_{D\sim \xd_{I^c}\vert \xt,\xd_I,Z}} R_{D\sim U\vert X,Z}\\
+R_{Y\sim U\vert \xt,\xd_{I},Z,D}\sqrt{1-R^{2}_{Y\sim D\vert\xt,\xd_{I},Z}}\sqrt{1- R^{2}_{D\sim U \vert X,Z}\big(1-R^2_{D\sim \xd_{I^c}\vert \xt, \xd_I,Z}\big)}
\Bigg].
\end{multline*}
Since this equation contains not only the regression sensitivity parameters but also $R_{Y\sim U\vert X,Z}$, we add the equality constraint \eqref{eq:reg-con-y22} to the optimization problem as well.

\subsubsection{Bound on $U \leftrightarrow Z$}
The bound on the correlation between $U$ and $Z$ given by
\begin{equation*}
    R^2_{Z \sim U\vert \xt,\xd_{-j}} \leq b_{U\!Z} R^2_{Z \sim
      \xd_j\vert \xt, \xd_{-j}},~j \in [\dot{p}],~b_{U\!Z}\geq 0
\end{equation*}
is in fact equivalent to a constraint on the IV sensitivity parameter $R_{Z\sim U\vert X}$. To show this, we first relate $R_{Z \sim U\vert \xt,\xd_{-j}}$ to $R_{Z \sim U\vert X}$ via the conditional independence assumption. Using $R_{U\sim \xd_j \vert \xt, \xd_{-j},Z}=0$, which follows from the assumption $R^2_{U\sim \xd\vert \xt,Z}=0$, and recursion of partial correlation [iv], we find
\begin{gather*}
    0 = R_{U\sim \xd_j\vert \xt_I,\xd_{-j},Z} \stackrel{\text{[iv]}}{=}  \frac{R_{U\sim \xd_j\vert \xt,\xd_{-j}}-R_{Z\sim \xd_j\vert \xt,\xd_{-j}}R_{Z\sim U\vert \xt,\xd_{-j}}}{\sqrt{1-R_{Z\sim \xd_j\vert \xt,\xd_{-j}}^2}\sqrt{1-R_{Z\sim U\vert \xt,\xd_{-j}}^2}}\\*
    \Leftrightarrow\quad
    R_{U\sim \xd_j\vert \xt,\xd_{-j}}=R_{Z\sim \xd_j\vert \xt,\xd_{-j}}R_{Z\sim U\vert \xt,\xd_{-j}}.
\end{gather*}
Employing this relationship and rule [iv] again, we infer
\begin{align*}
    R_{Z\sim U\vert X} \!\stackrel{\text{[iv]}}{=}\! \frac{R_{Z\sim U\vert \xt,\xd_{-j}}\!\!\!\! -\!\! R_{Z\sim \xd_j\vert \xt,\xd_{-j}}R_{U\sim \xd_j\vert \xt,\xd_{-j}}}{\sqrt{1\!\!-\!\!R_{Z\sim \xd_j\vert \xt,\xd_{-j}}^2}\sqrt{1\!\!-\!\!R_{U\sim \xd_j\vert \xt,\xd_{-j}}^2}}
    = \frac{R_{Z\sim U\vert \xt,\xd_{-j}}\sqrt{1\!\!-\!\!R^2_{Z\sim \xd_j\vert \xt,\xd_{-j}}}}{\sqrt{1\!-\!R^2_{Z\sim \xd_j\vert \xt,\xd_{-j}}\!R_{Z\sim U\vert \xt,\xd_{-j}}^2}}.
\end{align*}
As the right-hand side above is monotone in $R_{Z\sim U\vert \xt, \xd_j}$, the user-specified bound is in fact equivalent to the optimization constraint \eqref{eq:iv-con-z2}:
\begin{equation*}
    R^2_{Z\sim U\vert X} \leq b_{U\!Z}\, R^2_{Z\sim \xd_j\vert \xt,\xd_{-j}} \frac{1-R^2_{Z\sim \xd_j\vert \xt,\xd_{-j}}}{1- b_{U\!Z}\, R^4_{Z\sim \xd_j\vert \xt,\xd_{-j}}}.
\end{equation*}

\subsubsection{Bound on $Z\rightarrow Y$}
The comparative bound on the direct effect of $Z$ on $Y$, as given by
\begin{equation*}
    R^2_{Y\sim Z\vert X,U,D} \leq b_{ZY} R^2_{Y\sim \xd_j\vert \xt,
      \xd_{-j},Z,U,D},~j \in [\dot{p}], b_{ZY} \geq 0,
\end{equation*}
 is unusual in the sense that the sets of variables that are regressed out are different in the two partial $R^2$-values. It is difficult to specify comparative bounds for the exclusion restriction as the corresponding sensitivity parameter $R_{Y\sim Z\vert X,U,D}$ partials out $U$. Therefore, we cannot directly compare $U$ to an observed covariate, e.g.\ $\xd_j$, and the right-hand side of the bound cannot be estimated. For this reason, we resort to the proposed adjustment sets because we can connect $R_{Y\sim \xd_j\vert \xt,\xd_{-j},Z,U,D}$ to the regression sensitivity parameters. To this end, we employ rule [vi] -- the three-variable identity -- with $Y \equiv Y$, $X \equiv \xd_j$, $W \equiv U$ and $Z \equiv (\xt,\xt_{-j},D)$ which yields
\begin{multline*}
    f_{Y\sim \xd_j\vert \xt, \xd_{-j}, Z, U, D} \sqrt{1-R^2_{Y\sim U\vert X, Z, D}}\\ \stackrel{\text{[vi]}}{=}
    f_{Y\sim \xd_j\vert \xt, \xd_{-j}, Z, D} \sqrt{1-R^2_{U\sim \xd_j\vert \xt, \xd_{-j}, Z, D}} - R_{Y\sim U\vert X, Z, D}\, R_{U\sim \xd_j\vert \xt, \xd_{-j}, Z, D}.
\end{multline*}
Furthermore, we use $R_{U\sim \xd_j\vert \xt,\xd_{-j},Z} =0$ both to simplify the following recursive partial correlation formula [iv] and to apply the reduction of partial correlation formula [v] to $R_{D\sim U\vert X,Z}$
\begin{align*}
    R_{U\sim \xd_j\vert \xt, \xd_{-j}, Z, D} &\stackrel{\text{[iv]}}{=}
    \frac{R_{U\sim \xd_j\vert \xt, \xd_{-j}, Z}-R_{D\sim \xd_j\vert \xt, \xd_{-j}, Z}\,R_{D\sim U\vert \xt, \xd_{-j}, Z}}{\sqrt{1-R_{D\sim \xd_j\vert \xt, \xd_{-j}, Z}^2}\sqrt{1-R_{D\sim U\vert \xt, \xd_{-j}, Z}^2}}\\
    &= - f_{D\sim \xd_j\vert \xt, \xd_{-j}, Z}\,f_{D\sim U\vert \xt, \xd_{-j}, Z},\\
    R_{D\sim U\vert \xt, \xd_{-j}, Z} &\stackrel{\text{[v]}}{=} R_{D\sim U\vert X,Z} \sqrt{1-R^2_{D\sim \xd_j\vert \xt,\xd_{-j},Z}}.
\end{align*}
Inserting these two relationships in the three-variable identity above and cancelling some terms, we arrive at the optimization constraint \eqref{eq:iv-con-ex22}:
\begin{equation*}
\begin{multlined}
    f_{Y\sim \xd_j\vert \xt, \xd_{-j}, Z, U, D} \sqrt{1-R^2_{Y\sim U\vert X, Z, D}}
    = \bigg[f_{Y\sim \xd_j\vert \xt, \xd_{-j}, Z, D}\sqrt{1-R^2_{D\sim U\vert X,Z}} \\
    +R_{Y\sim U\vert X, Z, D}\,R_{D\sim \xd_j\vert \xt,\xd_{-j},Z}\,R_{D\sim U\vert X,Z} \bigg]\Big/\sqrt{1-R^2_{D\sim U\vert X,Z}(1-R^2_{D\sim \xd_j\vert \xt,\xd_{-j},Z})}.
\end{multlined}
\end{equation*}



\section{Adapted Grid-Search Algorithm}\label{app:algorithm}

\begin{algorithm}
\setstretch{1.15}
    \SetKwComment{Comment}{\vspace{-1cm}/* }{ */}
    \SetSideCommentLeft
    \KwIn{
    bounds on $U \rightarrow D$: $B_{1,j}^l, B_{1,j}^u, B_{1,j}^l(\hth),B_{1,j}^u(\hth)$,\\ \phantom{Inputttt} bounds on $U\rightarrow Y$: $B_{2,j}^l, B_{2,j}^u, B_{3,j}^l(\hth),B_{3,j}^u(\hth), B_{4,j}^l(\hth),B_{4,j}^u(\hth)$,\\
    \phantom{Inputttt} bounds on $U\leftrightarrow Z$: $B_{6,j}^l, B_{6,j}^u, B_{6,j}^l(\hth),B_{6,j}^u(\hth)$,\\
    \phantom{Inputttt} bounds on $Z\rightarrow Y$:
    $B_{7,j}^l, B_{7,j}^u, b_7$,\\
    \phantom{Inputttt} number of grid points: $N_{\psi_1}, N_{\psi_2}, N_{\psi_5}$.}
    \KwOut{grid approximation}
        \caption{Grid approximation of $\Psi^*(\hth)$}
        \label{alg:grid-complex}
        \DontPrintSemicolon
        $\psi_{1}^l \gets \max_j\{B_{1,j}^l, B_{1,j}^l(\hth)\}$ \Comment*[r]{U -> D}
        $\psi_{1}^u \gets \min_j\{B_{1,j}^u, B_{1,j}^u(\hth)\}$\;
        \For{$i \in \{1,\ldots,N_{\psi_1}\}$}{

            $\psi_{1,i} \gets \psi_{1}^l + (i-1)\, (\psi_{1}^u-\psi_{1}^l)/(N_{\psi_1}-1)$\;

            $\psi_{3,i}^l \gets \max_j\{B_{3,j}^l(\hth)\} \vee \max_{j}\{g_3(B_{4,j}^l(\hth), \psi_{1,i}, \hth)\}$ \Comment*[r]{U -> Y}
            $\psi_{3,i}^u \gets \min_j\{B_{3,j}^u(\hth)\} \wedge \min_{j}\{g_3(B_{4,j}^u(\hth), \psi_{1,i}, \hth)\}$\;

            $\psi_{2,i}^l \gets \max_j\{B_{2,j}^l\} \vee g_2(\psi_{3,i}^l, \psi_{1,i}, \hth)$\;
            $\psi_{2,i}^u \gets \min_j\{B_{2,j}^u\} \wedge g_2(\psi_{3,i}^u, \psi_{1,i}, \hth)$\;
            \uIf{$\psi_{2,i}^l > \psi_{2,i}^u$}{
                $\psi_{1,i} \gets$ \texttt{NA};
                $\quad \psi_{2,i}^l \gets$ \texttt{NA};
                $\quad \psi_{2,i}^u \gets$ \texttt{NA}\;
            }\Else{
                $\psi_{6}^l \gets \max_j\{B_{6,j}^l, B_{6,j}^l(\hth)\}$ \Comment*[r]{U <-> Z}
        $\psi_{6}^u \gets \min_j\{B_{6,j}^u, B_{6,j}^u(\hth)\}$\; 
                $\psi_{5,i}^l \gets f^{-1}(\,g_{f,5}(\psi_6^l, \psi_{1,i},\hth)\,)$\;
                $\psi_{5,i}^u \gets f^{-1}(\,g_{f,5}(\psi_6^u, \psi_{1,i},\hth)\,)$\;
                \For{$k \in \{1,\ldots,N_{\psi_2}\}$}{
                    $\psi_{2,i,k} \gets \psi_{2,i}^l + (k-1)\, (\psi_{2,i}^u-\psi_{2,i}^l)/(N_{\psi_2}-1)$\;
                    $\mathrm{comp}_{i,k}\gets$ \texttt{False}\;
                    $\psi_{7,i,k}^l \gets \max_j\{B_{7,j}^l\} \vee - \sqrt{b_7} \cdot \lvert f^{-1}(\, g_{f,8}(\psi_{1,i}, \psi_{2,i,k},\hth) \,)\rvert$ \Comment*[r]{Z -> Y}
                    $\psi_{7,i,k}^u \gets \min_j\{B_{7,j}^u\} \wedge \sqrt{b_7} \cdot \lvert f^{-1}(\, g_{f,8}(\psi_{1,i}, \psi_{2,i,k},\hth) \,)\rvert$\;
                    \If{$\psi_{7,i,k}^l \leq \psi_{7,i,k}^u$}{
                        \For{$m \in \{1,\ldots,N_{\psi_5}\}$ and $\neg\mathrm{comp}_{i,k}$}{
                            $\psi_{5,i,m} \gets \psi_{5,i}^l + (m-1)\, (\psi_{5,i}^u-\psi_{5,i}^l)/(N_{\psi_5}-1)$\;
                            $\psi_{7,i,k,m} \gets f^{-1}(\,g_{f,7}(\psi_{2,i,k}, \psi_{5,i,m}, \hth)\,)$\;
                            $\mathrm{comp}_{i,k} \gets \psi_{7,i,k}^l \leq \psi_{7,i,k,m}$ \texttt{and} $\psi_{7,i,k,m} \leq \psi_{7,i,k}^u$\;
                        }
                    }
                }
                $\psi_{2,i}^l \gets \min_k\{\psi_{2,i,k}\colon \mathrm{comp}_{i,k}=\text{\texttt{True}}\}$\;
                $\psi_{2,i}^u \gets \max_k\{\psi_{2,i,k}\colon \mathrm{comp}_{i,k}=\text{\texttt{True}}\}$\;
            }
            
        }
        \KwRet{$\cup_{i=1}^{N_{\psi_1}} \{\psi_{1,i}\} \times \{\psi_{2,i}^l, \psi_{2,i}^u\}$}
\end{algorithm}

The proposed algorithm first constructs a set of equally spaced points that are (approximately) contained in $\Psi^*(\hat{\theta})$; then, it evaluates $\beta$ over this set and takes the minimum/maximum of the obtained $\beta$-values. The latter step is straightforward whereas the former is complex when multiple interlocking constraints are present. In order to keep the the notation short, we introduce some abbreviations:
\begin{alignat*}{4}
    \psi_1 &= R_{D\sim U\vert X,Z}, &\qquad
    \psi_2 &= R_{Y\sim U\vert X,Z,D}, &\qquad
    \psi_3 &= R_{Y\sim U\vert X,Z}, &\qquad
    \psi_4 &= R_{Y\sim U\vert \xt,\xd_I, Z,D},\\
    \psi_5 &= R_{Z\sim U\vert X,Z}, &\qquad
    \psi_6 &= R_{Z\sim U\vert X}, &\qquad
    \psi_7 &= R_{Y\sim Z\vert X,U,D}, &\qquad
    \psi_8 &= R_{Y\sim \xd_j\vert \xt, \xd_{-j},Z,U,D}.
\end{alignat*}
With slight abuse of notation, the parameters $\psi_4$ and $\psi_8$ can also be vectors if multiple constraints with different choices of $I$ and $j$, respectively, are specified. Further, the estimable parameters, e.g.\ $R_{Y\sim D\vert X,Z}$ or $R^2_{D\sim \xd_j\vert \xt, \xd_{-j},Z}$, are denoted by $\theta$ and the corresponding estimator by $\hth$. The inverse of the $f$-transformation $x \mapsto y = x/\sqrt{1-x^2}$ is denoted by $f^{-1}$ and given by $y \mapsto x = y/\sqrt{1+y^2}$. Both functions are monotonically increasing.

For direct constraints, $B_{k,j}^l$ and $B_{k,j}^u$ denote the $j$-th lower and upper bounds on $\psi_k$, respectively. Since comparative bounds are data-dependent, we add the argument $\hth$ to the respective bounds: For instance, $B_{3,1}^l(\hth)$ is the first lower bound on $\psi_3$ which follows from specifying a comparative constraint on $U\rightarrow Y$ without partialing out $D$ according to \eqref{eq:reg-con-y21}. The scalar or vector $b_7$ denotes the $b$-factor(s) in comparative bounds on ${Z\rightarrow Y}$. Beyond inequality constraints, comparative bounds also imply equality constraints $g_k$ and $g_{f,k}$ in the optimization problem. The former represents an equality constraint that determines $\psi_k$ as a function of other sensitivity and estimable parameters. For instance, $g_2$ stems from \eqref{eq:reg-con-y22} and $g_3$ from \eqref{eq:reg-con-y32}. The latter $g_{f,k}$, denotes a function that determines $f(\psi_k)$: $g_{f,7}$ and $g_{f,5}$ come from the first and second line of \eqref{eq:reg-iv-connection}, respectively, and $g_{f,8}$ stems from \eqref{eq:iv-con-ex22}. Inserting vectors instead of scalars as arguments of $g_k$ or $g_{f,k}$ is interpreted as componentwise evaluation. Like the inverse $f$-transformation, the constraints $g_k$ and $g_{f,k}$ are monotonically increasing in one of their arguments. This allows to ``push forward'' a bound on this argument onto $\psi_k$ and $f(\psi_k)$, respectively.

In order to compute a set of points that is approximately contained in $\Psi^*(\hat{\theta})$, we first discretize the interval of all possible $\psi_1$-values and construct the vector $(\psi_{1,i})_i \in \R^{N_{\psi_1}}$ which contains $N_{\psi_1}$ equally spaced points. (This corresponds to discretizing the interval $[\min\{\psi_1 \colon P_{\psi_1} \neq \emptyset\}, \max\{\psi_1 \colon {P_{\psi_1} \neq \emptyset}\}]$.) Second, we construct the vectors $(\psi_{2,i}^l)_i, (\psi_{2,i}^u)_i \in \R^{N_{\psi_1}}$ which approximate the corresponding minima and maxima of $\beta$ at the respective $\psi_1$-value. Thus, we can create the points
\begin{equation*}
    \bigcup_{i=1}^{N_{\psi_1}} \{\psi_{1,i}\} \times \{\psi_{2,i}^l, \psi_{2,i}^u\},
\end{equation*}
which are (approximately) contained and equally spaced in $\Psi^*(\hat{\theta})$. Evaluating the objective~$\beta$ over this set has complexity $\mathcal{O}(N_{\psi_1})$.

In case that only bounds on $U\rightarrow D$ and $U\rightarrow Y$ are specified, the computational complexity of generating the $\psi_1$-, $\psi_2^l$- and $\psi_2^u$-vectors grows linearly in $N_{\psi_1}$ and the computed points are actually elements of $\Psi^*(\hat{\theta})$ instead of merely approximating it. 
If at least one bound on $U \leftrightarrow Z$ or $Z \rightarrow Y$ is specified, determining the $\psi_2^l$- and $\psi_2^u$-vectors is more involved. We perform a two-dimensional grid search to determine $\psi_2$-values that are compatible with the imposed OLS- and IV-constraints. This ultimately amounts to checking if a $\psi_7$-value (that is $R_{Y\sim Z\vert X,U,D}$) exists that lies within the resulting bounds on it. Accordingly, the computational complexity is $\mathcal{O}(N_{\psi_1}\cdot N_{\psi_2} \cdot N_{\psi_5})$.

Algorithm~\ref{alg:grid-complex} contains the pseudocode. Note that the stated version of the algorithm is not the most computationally efficient but eases readability. 

\begin{remark}
    If the sensitivity model is well-specified, the probability of the estimated constraint set $\Psi(\hat{\theta})$ being empty is expected to converge to zero as the sample size grows. However, for moderate sample sizes an empty estimated constraint set may occasionally occur. In this case, we take a conservative approach and set the optimal value to $\infty$ or $-\infty$ depending on which end of the PIR is considered.
\end{remark}

\section{Simulation Study}\label{sec:simulation-study}
We investigate the empirical coverage of sensitivity intervals computed with the bootstrap in two scenarios: a regression model with one additional covariate and an instrumental variable model. In both set-ups, we set the nominal level to 90\%.

\subsection{Linear Regression Simulation}
We generate a sample of $n$ i.i.d.\
random vectors $(\eps_U, \eps_X,
\eps_X, \eps_Y)^T\sim \mathcal{N}(0,\Id)$ and compute the variables in
the model using the following linear structural equations:
\begin{equation*}
    U := \eps_U,\quad
    X := \eps_X,\quad
    D := X + U + \eps_D,\quad
    Y := D + 2X + U + \eps_Y.
\end{equation*}
Based on these structural equations, we derive the covariance matrix of the involved random variables
\begin{equation*}
    \var\bigg(
    \begin{pmatrix}
        U\\ X\\ D\\Y
    \end{pmatrix}
    \bigg)=
    \begin{pmatrix}
        1 & 0 & 1 & 2\\
        0 & 1 & 1 & 3\\
        1 & 1 & 3 & 6\\
        2 & 3 & 6 & 15
    \end{pmatrix}.
\end{equation*}
It can be used to compute (partial) $R$-values as well as the bias $\beta_{Y\sim D\vert X} - \beta = 1/2$. If the comparative constraints
\begin{equation*}
    R^2_{D\sim U} \leq R^2_{D\sim X},\qquad
    R^2_{Y\sim U} \leq \frac{4}{9}\, R^2_{Y\sim X}
\end{equation*}
are specified, the partially identified region is $[1,(3+\sqrt{3})/2]$. Hence, the true value $\beta=1$ equals the lower end of the PIR. The bounds above are sharp in the sense that the lower end of the partially identified range can only be reached when both inequalities are active, i.e.\ hold with equality.

In order to construct sensitivity intervals, we generate bootstrap
samples of the observed data and solve the corresponding optimization
problems. Then, we use percentile, basic and BCa bootstrap
\citep[chap.\ 5]{davison1997} to compute the
lower and upper end of the sensitivity interval. This approach is
compared to the heuristic sensitivity intervals described in Section \ref{sec:inference} as well as the oracle 90\% confidence interval, which could be computed if $U$
was observed.

We simulate data for different sample sizes $n$ and repeat each such experiment 1000 times to compute the empirical coverage and length of the sensitivity/confidence intervals. More specifically, we evaluate the empirical coverage of $\beta$ and the PIR for different sensitivity intervals and adapt the notion of length. In order to account for the fact that the length of typical confidence intervals approaches zero as $n\to \infty$ whereas the length of valid sensitivity intervals is lower bounded by the length of the PIR, we use the distance between the lower end of an interval and 1, when it covers 1, as length.

\begin{table}
  \caption{\label{tab:sim-reg} Simulation results of the linear regression
    example.}
  \centering
  \vspace{0.3cm}
    \begin{tabular}{llcccc}
    \hline
    \multicolumn{1}{c}{$n$} & \multicolumn{1}{c}{Method} & \multicolumn{2}{c}{Coverage} & \multicolumn{2}{c}{Length}\\ \cline{3-6}
                             & & $\beta$ & PIR & Mean & Median \\ \hline
    200 & Percentile bootstrap & 94.3\% & 91.6\% & 2.877 & 0.760 \\
    & Basic bootstrap & 86.2\% & 78.9\% & 0.523 & 0.306 \\
    & BCa bootstrap & 93.9\% & 90.6\% & 2.758 & 0.730 \\
    & Heuristic & 73.0\% & 47.9\% & 0.335 & 0.207 \\
    & Oracle & 91.0\% & - & 0.125 & 0.124 \\[1ex]
    500 & Percentile bootstrap & 95.7\% & 91.0\% & 0.455 & 0.329 \\
    & Basic bootstrap & 87.6\% & 81.5\% & 0.252 & 0.196 \\
    & BCa bootstrap & 95.2\% & 90.4\% & 0.439 & 0.318 \\
    & Heuristic & 72.3\% & 44.3\% & 0.161 & 0.121 \\
    & Oracle & 90.4\% & - & 0.079 & 0.077 \\[1ex]
    1000 & Percentile bootstrap &
    95.9\% & 91.7\% & 0.245 & 0.216 \\
    & Basic bootstrap & 90.7\% & 85.1\% & 0.175 & 0.151 \\
    & BCa bootstrap & 95.2\% & 91.1\% & 0.241 & 0.212 \\
    & Heuristic & 72.5\% & 46.6\% & 0.114 & 0.095 \\
    & Oracle & 89.2\% & - & 0.056 & 0.054 \\[1ex]
    2000 & Percentile bootstrap & 95.3\% & 91.7\% & 0.147 & 0.134 \\
    & Basic bootstrap & 89.1\% & 84.7\% & 0.119 & 0.104 \\
    & BCa bootstrap & 94.5\% & 91.0\% & 0.146 & 0.133 \\
    & Heuristic & 70.4\% & 43.9\% & 0.074 & 0.063 \\
    & Oracle & 90.6\% & - & 0.038 & 0.037\\
    \hline
    \end{tabular}
\end{table}
The results of this simulation study are summarized in Table \ref{tab:sim-reg}. Both percentile and BCa bootstrap exhibit coverage of the PIR close to the envisaged level of 90\%; its coverage of $\beta$ is close to 95\%. The latter is expected as the true value of $\beta$ is the lower end of the PIR. By comparison, the empirical coverage of sensitivity intervals constructed via basic bootstrap is 5 to 10 percentage points below the required level. Hence, we use BCa bootstrap to construct sensitivity intervals in the data example in the main text. Moreover, this simulation study illustrates that heuristic sensitivity intervals do not possess frequentist coverage guarantees: the empirical coverage of the PIR is consistently below 50\%. Finally, we see that sensitivity intervals are substantially longer than the oracle confidence interval. We attribute the increased length to the uncertainty stemming from estimating the constraints. In this simulation study, we did not encounter cases where the estimated constraint set was empty on a bootstrap sample.
\begin{figure}[t!]
    \centering
    \includegraphics[scale=0.45]{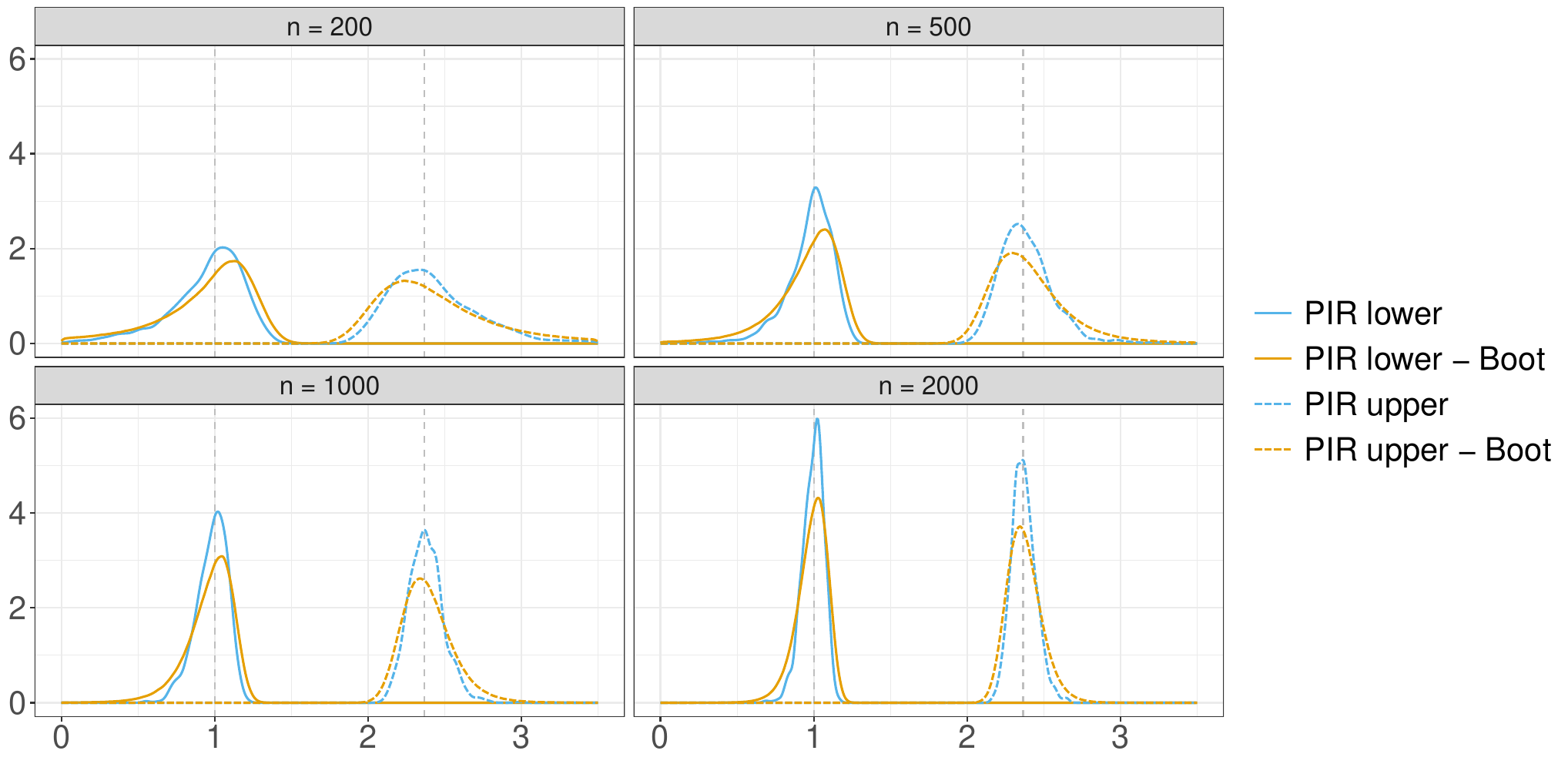}
    \caption{\label{fig:distribution-comparison} Empirical distribution of the lower and upper end of the PIR as well as the corresponding bootstrap distributions.}
\end{figure}

In order to investigate the coverage of bootstrap sensitivity intervals more closely, we consider the distribution of the estimated upper and lower end of the PIR as well as the corresponding bootstrap distributions. Figure \ref{fig:distribution-comparison} depicts the estimates of these distributions based on 1000 repetitions of the experiment. For small sample sizes $n$, we notice that the bootstrap distribution is both biased and skewed. Both phenomena diminish as $n$ grows so that the bootstrap distribution approximates the target distribution more closely. 

\subsection{Linear Instrumental Variable Simulation}
We generate 100 i.i.d.\ samples from the distribution $(\eps_U, \eps_Z, \eps_D, \eps_Y)^T \sim \mathcal{N}(0, \Id)$ and compute the variables of the model as follows
\begin{equation*}
    U := \eps_U,\quad
    Z := \eps_Z,\quad
    D := Z + U + \eps_D,\quad
    Y := D + U + \eps_Y.
\end{equation*}
This data-generating process fulfills the instrumental variable assumptions which renders $\beta = 1$ point identified. Hence, a sensitivity interval where the IV-related sensitivity parameters are set to zero ought to be comparable with the confidence interval that is based on the asymptotic normality of the TSLS estimator. In order to use Algorithm~\ref{alg:grid-complex}, we slightly relax the IV assumptions requiring $R_{Z\sim U\vert X}, R_{Y\sim Z\vert X,U,D}\! \in [-0.002, 0.002]$ and further set $R_{D\sim U\vert X, Z} \in [-0.999,0.999]$ to bound it away from $-1$ and $1$.

We compute the empirical coverage and length of sensitivity intervals constructed via percentile, basic and BCa bootstrap, the heuristic sensitivity intervals and the oracle confidence intervals over 1000 repetitions of the experiment. Due to the high computational costs, we conduct this simulation study only for sample size $n=100$.

\begin{table}[t!]
  \caption{\label{tab:sim-iv} Simulation results of the instrumental variable example.}
  \centering
  \vspace{0.3cm}
    \begin{tabular}{lccc}
    \hline
    \multicolumn{1}{c}{Method} & \multicolumn{1}{c}{Coverage} & \multicolumn{2}{c}{Length}\\ \cline{3-4}
    &  & Mean & Median \\ \hline
    Percentile bootstrap & 92.2\% & 0.328 & 0.304 \\
    Basic bootstrap & 94.1\% & 0.262 & 0.239 \\
    BCa bootstrap & 91.7\% & 0.330 & 0.301 \\
    Heuristic & 99.9\% & 1.721 & 1.037 \\
    Oracle & 88.8\% & 0.179 & 0.170\\
    \hline
    \end{tabular}
\end{table}
The results of this experiment are stated in Table \ref{tab:sim-iv}. We notice that the bootstrap sensitivity intervals are on par with the oracle confidence interval in terms of coverage and comparable in terms of length. By contrast, the heuristic sensitivity intervals exhibit very high coverage but their length is too long to be informative in practice. In this simulation study, 4 of the $10^3 \cdot 10^3 = 10^6$ constructed bootstrap samples led to an empty constraint set. 

\section{Comparison Points in $R$-contour Plots}\label{app:comparison-points}
In Section~\ref{sec:r-contours}, we add comparison points to the $R$-contour plots to provide context for assessing the strength of the unmeasured confounder $U$. First, we provide the formulas of the comparison points and discuss them; then, we state the mathematical derivations.

\subsection{Formulas of Different Comparison Points}
In the main text, we consider three different kinds of comparison points. The first approach follows a proposal by \citet{imbens_2003} and adds the points
\begin{equation*}\label{eq:informal-comp-point}
    \left(\sqrt{b}\, \rh_{D \sim X_j\vert X_{-j},Z},\, \sqrt{b}\, \rh_{Y\sim X_j\vert
  X_{-j},Z,D}\right),
\end{equation*}
for certain choices of $b > 0$ and $j \in [p]$ to the contour plot. However, this method of constructing benchmarks for $R_{D\sim U\vert X,Z}$ and $R_{Y\sim U\vert X,Z,D}$ is not entirely honest because different sets of covariates are conditioned on. Moreover, regressing out a potential collider $D$ may leave $\hat{R}^2_{Y\sim X_j\vert X_{-j},Z,D}$ hard to interpret.

Therefore, we aim to provide more rigorous comparison points. To this end, we introduce the abbreviations 
\begin{equation*}
    \rh_D :=\rh_{D\sim \xd_j\vert \xt, \xd_{-j}, Z},\qquad \rh_Y := \rh_{Y\sim \xd_j\vert\xt,\xd_{-j},Z,D},
\end{equation*}
where $j \in [\dot{p}]$; the corresponding $\rh^2$- and $\fh$-values are defined accordingly.

First, we consider the sensitivity parameter $R_{D\sim U\vert X,Z}$. According to \eqref{eq:reg-cond-d2}, if $U$ explains $b$ times as much variance of $D$ as $\xd_j$ and the respective correlations of $U$ and $\xd_j$ with $D$ have the same sign, then
\begin{equation}\label{eq:comp-point-d}
  R_{D\sim U \vert X,Z} = \sqrt{b}\, \frac{R_{D\sim \xd_j \vert \xt,
      \xd_{-j},Z}}{\sqrt{1-R^2_{D\sim \xd_{j}\vert \xt,\xd_{-j},Z}}} = \sqrt{b}\, f_{D\sim \xd_j \vert \xt,\xd_{-j},Z}.
\end{equation}

For the second sensitivity parameter $R_{Y\sim U\vert X,Z,D}$, we can make two types of comparisons: unconditional and conditional on $D$, cf.\ sensitivity bounds on $U\rightarrow Y$ in Table~\ref{tab:constraints}. In the former case, we assume that $U$ explains $b$ times as much variance in $Y$ as $\xd_j$ does, given $(\xt, \xd_{-j},Z)$, and that the respective correlations have the same sign. In the latter case, we additionally condition on $D$. For the first comparison, using the relationships \eqref{eq:reg-con-y22} and \eqref{eq:reg-con-y21} yields
\begin{equation}\label{eq:r-contour-comp1}
     R_{Y\sim U\vert X,Z,D} = \frac{1}{\sqrt{1-(1+b)R^2_{D\sim \xd_j\vert \xt, \xd_{-j},Z}}}\, \sqrt{b}\,f_{Y\sim \xd_j\vert \xt, \xd_{-j},Z,D}.
\end{equation}
For the second comparison, we apply a similar reasoning and obtain
\begin{equation}\label{eq:r-contour-comp2}
\begin{multlined}
    R_{Y\sim U\vert X,Z,D}\\[1ex] = \frac{\sqrt{1-(1+b)R_{D\sim \xd_j\vert \xt, \xd_{-j},Z}^2\!+b R_{D\sim \xd_j\vert \xt, \xd_{-j},Z}^4}\!+\!R_{D\sim \xd_j\vert \xt, \xd_{-j},Z}^2}{\sqrt{1-(1+b)R^2_{D\sim \xd_j\vert \xt, \xd_{-j},Z}}}\, \sqrt{b}\,f_{Y\sim \xd_j\vert \xt, \xd_{-j},Z,D}.
\end{multlined}
\end{equation}
Hence, in order to compare $U$ to $b$ times the explanatory capability of $\xd_j$, we can add the estimated comparison points
\begin{equation*}
    \left(\sqrt{b} \fh_D,\, \frac{1}{\sqrt{1-(1+b)\rh_D^2}}\, \sqrt{b}\,\fh_Y\right),\quad \text{or}\quad
    \left( \sqrt{b} \fh_D,\, \frac{\sqrt{1-(1+b)\rh_D^2+b \rh_D^4}+\rh_D^2}{\sqrt{1-(1+b)\rh_D^2}}\, \sqrt{b}\,\fh_Y\right),
\end{equation*}
depending on whether we compare $U$ and $\xd_j$ unconditional or conditional on $D$. The corresponding comparison points for $R^2$- instead of $R$-values were suggested by \citet{ch}. Moreover, we remark that the two proposed points are identical if $b=1$.

\subsection{Derivation of Comparison Points}\label{sec:derivation-comparison-points}
To shorten the notation in the following, we introduce the abbreviation $W := (\xt, \xd_{-j}, Z)$. In the main text, we consider the case where $U$ explains $b$ times as much variance in $D$ and $Y$ as $\xd_j$ does -- conditional on $W$ or $(W,D)$. Here, we consider the more general case, where the multiplicative constant $b$ may be different for the relationship between $U \rightarrow D$ and $U \rightarrow Y$. We indicate this with a corresponding subscript.

First, we consider the sensitivity parameter $R_{D\sim U\vert X,Z}$. Due to \eqref{eq:reg-cond-d2}, if $U$ explains $b_D$ times as much variance of $D$ as $\xd_j$ -- conditional on $W$ -- and the respective correlations of $U$ and $\xd_j$ with $D$ have the same sign, then
\begin{equation}\label{eq:comp-point-d}
  R_{D\sim U \vert X,Z} = \sqrt{b_D}\, \frac{R_{D\sim \xd_j \vert W}}{\sqrt{1-R^2_{D\sim \xd_{j}\vert W}}} = \sqrt{b_D}\, f_{D\sim \xd_j \vert W}.
\end{equation}

In the following, we derive a similar formula for $R_{Y\sim U\vert X,Z,D}$ both for comparisons unconditional on $D$ and conditional on $D$. In either case, we make use of the following equation which follows from applying rule~[vi] of the $R^2$-calculus with $Y \equiv Y$, $X \equiv \xd_j$, $W \equiv D$ and $Z = W$:
\begin{equation}\label{eq:derivation-comp-points}
    R_{Y\sim D\vert W, \xd_j} \stackrel{\text{[vi]}}{=} \frac{1}{R_{D\sim \xd_j\vert W}}\!\left[
    f_{Y\sim \xd_j\vert W}\,\sqrt{1-R^2_{D\sim \xd_j\vert W}} - f_{Y\sim \xd_j\vert W, D}\,\sqrt{1-R^2_{Y\sim D\vert W, \xd_j}}  
    \right]\!.
\end{equation}

\subsubsection{Comparison unconditional on $D$}\label{sec:comp-point-unconditional}
Assuming that $U$ explains $b_Y$ times as much variance in $Y$ as $\xd_j$ does -- conditional on $W$ -- and that the respective correlations have the same sign, \eqref{eq:reg-con-y21} implies
\begin{equation*}
    R_{Y\sim U\vert W,\xd_j} = \sqrt{b_Y}\,\frac{R_{Y\sim \xd_j\vert W}}{\sqrt{1-R^2_{Y\sim \xd_j\vert W}}} = \sqrt{b_Y}f_{Y\sim \xd_j\vert W}.
\end{equation*}
Inserting this relationship as well as \eqref{eq:comp-point-d} into \eqref{eq:reg-con-y22} yields
\begin{equation*}
    R_{Y\sim U\vert X,Z,D} = R_{Y\sim U\vert W,\xd_j,D} = \frac{\sqrt{b_Y}f_{Y\sim \xd_j\vert W} - R_{Y\sim D\vert W, \xd_j}\,\sqrt{b_D}f_{D\sim \xd_j\vert W}}{\sqrt{1-R^2_{Y\sim D\vert W, \xd_j}}\sqrt{1-b_D f_{D\sim \xd_j\vert W}^2}}.
\end{equation*}
We reformulate this expression by invoking \eqref{eq:derivation-comp-points}, multiplying numerator and denominator with $\sqrt{1-R^2_{Y\sim \xd_j\vert W}}$ and applying rule [iii] of the $R^2$-calculus
\begin{equation*}
    R_{Y\sim U\vert X,Z,D} = \frac{\left(\sqrt{b_Y}-\sqrt{b_D}\right)\,R_{Y\sim \xd_j\vert W} + \sqrt{b_D}\,\sqrt{\frac{1-R^2_{Y\sim D+\xd_j\vert W}}{1-R^2_{D\sim \xd_j\vert W}}}\, f_{Y\sim \xd_j\vert W, D}}{\sqrt{1-R^2_{Y\sim D+\xd_j\vert W}}\sqrt{1-b_D f_{D\sim \xd_j\vert W}^2}}.
\end{equation*}
Equation \eqref{eq:r-contour-comp1} follows from setting $b_Y = b_D = b$ and simplifying the resulting expression.

\subsubsection{Comparison conditional on $D$}\label{sec:comp-point-conditional}
Assuming that $U$ explains $b_Y$ times as much variance in $Y$ as $\xd_j$ does -- conditional on $W$ and $D$ -- and that the respective correlations have the same sign implies
\begin{equation*}
    R_{Y\sim U\vert W,D} = \sqrt{b_Y} R_{Y\sim \xd_j\vert W,D}.
\end{equation*}
Plugging this relationship as well as \eqref{eq:comp-point-d} into \eqref{eq:reg-con-y32} leads to
\begin{multline*}
    R_{Y\sim U\vert W,\xd_j} = \frac{1}{\sqrt{1-R^2_{Y\sim \xd_j\vert W}}} \bigg[R_{Y\sim D\vert W}\,\sqrt{b_D}\,R_{D\sim \xd_j\vert W}\\ + \sqrt{b_Y}\,R_{Y\sim \xd_j\vert W,D}\,\sqrt{1-R^2_{Y\sim D\vert W}}\sqrt{1-b_D R^2_{D\sim \xd_j\vert W}}\bigg].
\end{multline*}
To get a formula for $R_{Y\sim U\vert X,Z,D}$, we insert the equation above into \eqref{eq:reg-con-y22}, multiply numerator and denominator by $\sqrt{1-R^2_{Y\sim \xd_j\vert W}}$, simplify the resulting expression via rule [iii] and obtain
\begin{multline*}
        R_{Y\sim U\vert X,Z,D} = 
        -\frac{R_{Y\sim D\vert W, \xd_j}\,\sqrt{b_D}f_{D\sim \xd_j\vert W} \sqrt{1-R^2_{Y\sim\xd_j\vert W}}}{\sqrt{1-R^2_{Y\sim D+\xd_j\vert W}}\sqrt{1-b_D f_{D\sim \xd_j\vert W}^2}}\\[1.5ex]
        +\frac{R_{Y\sim D\vert W}\,\sqrt{b_D}\,R_{D\sim \xd_j\vert W}
        + \sqrt{b_Y}\,R_{Y\sim \xd_j\vert W,D}\,\sqrt{1-R^2_{Y\sim D\vert W}}\sqrt{1-b_D R^2_{D\sim \xd_j\vert W}}}{\sqrt{1-R^2_{Y\sim D+\xd_j\vert W}}\sqrt{1-b_D f_{D\sim \xd_j\vert W}^2}}.
\end{multline*}
Invoking \eqref{eq:derivation-comp-points} in the first equality below and applying rule [iv] with $Y \equiv Y$, $X \equiv \xd_j$, $W \equiv D$ and $Z \equiv W$ in the second, we arrive at
{\allowdisplaybreaks
\begin{align*}
    R_{Y\sim U\vert X,Z,D}&=
    \begin{multlined}[t]
    \frac{\sqrt{b_D}(R_{Y\sim D\vert W}\,R_{D\sim \xd_j\vert W}-R_{Y\sim \xd_j\vert W})
        +\sqrt{b_D}\, \sqrt{\frac{1-R^2_{Y\sim D+\xd_j\vert W}}{1-R^2_{D\sim X_j\vert W}}}\, f_{Y\sim \xd_j\vert W,D}}{\sqrt{1-R^2_{Y\sim D+\xd_j\vert W}}\sqrt{1-b_D f_{D\sim \xd_j\vert W}^2}}\\[1ex]
        + \frac{\sqrt{b_Y}\,R_{Y\sim \xd_j\vert W,D}\,\sqrt{1-R^2_{Y\sim D\vert W}}\sqrt{1-b_D R^2_{D\sim \xd_j\vert W}}}{\sqrt{1-R^2_{Y\sim D+\xd_j\vert W}}\sqrt{1-b_D f_{D\sim \xd_j\vert W}^2}}
    \end{multlined}\\[2ex]
    &=
    \begin{multlined}[t]
    \frac{\left(\sqrt{b_Y} \sqrt{1-b_D R^2_{D\sim \xd_j\vert W}} - \sqrt{b_D}\sqrt{1-R^2_{D\sim \xd_j\vert W}}\right)R_{Y\sim \xd_j\vert W,D}\,\sqrt{1-R^2_{Y\sim D\vert W}}}{\sqrt{1-R^2_{Y\sim D+\xd_j\vert W}}\sqrt{1-b_D f_{D\sim \xd_j\vert W}^2}}\\[1ex]
        + \frac{\sqrt{b_D}\, \sqrt{\frac{1-R^2_{Y\sim D+\xd_j\vert W}}{1-R^2_{D\sim X_j\vert W}}}\, f_{Y\sim \xd_j\vert W,D}}{\sqrt{1-R^2_{Y\sim D+\xd_j\vert W}}\sqrt{1-b_D f_{D\sim \xd_j\vert W}^2}}.
    \end{multlined}
\end{align*}
}
Equation \eqref{eq:r-contour-comp2} now follows from setting $b_Y = b_D = b$ and simplifying the resulting expression.

\section{Additional $R$-contour Plot}
\begin{figure}[t!]
    \centering
    \makebox{\includegraphics[scale=0.55]{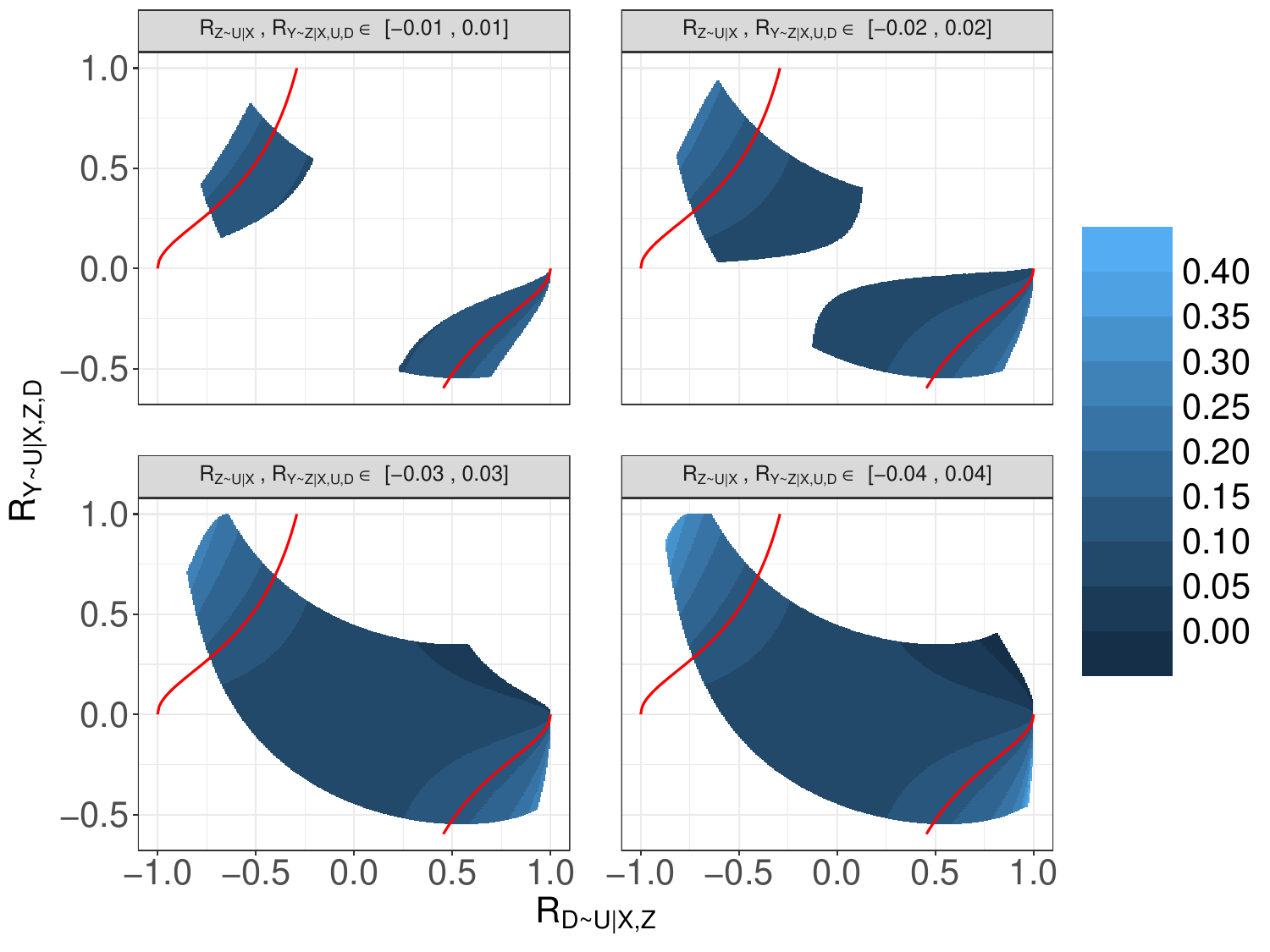}}
    \caption{\label{fig:r-contour-iv}$R$-sensitivity contours for the lower end of the estimated PIR: The lines correspond to the values of $R_{D\sim U\vert X,Z}$ and $R_{Y\sim U\vert X,Z,D}$ that conform with the IV-assumptions according to \eqref{eq:reg-iv-connection}.}
\end{figure}

In this section, we illustrate the utility of the $R$-contour plot as a way to visualize the estimated feasible set $\Psi(\hth)$. To this end, we again use the dataset analysed in \citet{card}. Sensitivity analysis with multiple bounds often entails a complex set of constraints. Consider the following sensitivity model
\begin{gather*}
    R_{Z\sim U\vert X},\, R_{Y\sim Z\vert X, U, D} \in [-r, r],\quad r\in\{0.01, 0.02, 0.03, 0.04\},\\
    R^2_{Y\sim U\vert \xt, \xd_{-j}, Z, D} \leq 5\, R^2_{Y\sim \xd_j\vert \xt, \xd_{-j}, Z, D},\qquad
   R_{D\sim U\vert X,Z} \in [-0.99, 0.99],
\end{gather*}
where $r$ parameterizes the degree of deviation from the instrumental variables assumptions and the covariate $\xd_j$ is the indicator for being black.

Figure \ref{fig:r-contour-iv} shows the estimated feasible set $\Psi(\hat{\theta})$ for different values of $r$. For $r=0.01$, the feasible set is small and concentrated around the lines that correspond to $R_{D\sim U\vert X,Z} = R_{Y\sim  U\vert X,Z,D} = 0$ (the instrumental variable is valid). As $r$ increases, the feasible set becomes larger as expected. The curved shape of the region of feasible values is a result of the comparative bound on $U\rightarrow Y$ and the associated constraints. Note that -- despite seemingly simple constraints -- $\Psi(\hat{\theta})$ is not convex and for small values of $r$ not even connected. Moreover, we observe that $\beta$ assumes its most extreme values as $R_{D\sim U\vert X,Z}$ approaches 1. This highlights the importance of bounding $R_{D\sim U\vert X,Z}$ away from $-1$ and $1$ to ensure that the PIR has finite length.

\section{Choice of Hyper-parameters}
\begin{table}[tbh]
  \caption{\label{tab:hyperparameters} Hyper-parameters for different plots and simulation examples.}
  \centering
  \vspace{0.3cm}
    \begin{tabular}{lccc}
    \hline
    & $N_{\text{grid}}$ & $ N_{\text{b-contour}} $ & $N_{\text{boot}}$ \\\hline
    Figure \ref{fig:iv-sensint} & 200 & - & 3500\\
    Figure \ref{fig:b-contour-reg} & 200 & 30 & - \\
    Figure \ref{fig:b-contour-iv} & 200 & 30 & - \\
    Figure \ref{fig:r-contour-reg} & 400 & - & -\\
    Figure \ref{fig:r-contour-iv} & 300 & - & -\\
    Table \ref{tab:sim-reg} & 200 & - & 2500\\
    Table \ref{tab:sim-iv} & 200 & - & 1000\\
    \hline
    \end{tabular}
\end{table}

In Table \ref{tab:hyperparameters}, we list the hyper-parameters of Algorithm \ref{alg:grid-complex} that were used for different data analyses. The mesh size of the grid is the same in every dimension, that is the numbers of points considered per dimension $N_{\psi_1}$, $N_{\psi_2}$, and $N_{\psi_5}$ are equal. We define 
$N_{\text{grid}}:= N_{\psi_1} =~N_{\psi_2}=~N_{\psi_5}$. The number of points per dimension for $b$-contour plots and
the number of bootstrap samples are denoted by $N_{\text{b-contour}}$
and $N_{\text{boot}}$, respectively.

In the simulation study and data example in this work, we found that the PIR estimates change only marginally for values of $N_{\text{grid}}$ larger than 200. We recommend to consider at least 100 points per grid dimension, i.e.\ $N_{\text{grid}} = 100$. The rough structure of the $b$-contours often becomes apparent for $N_{\text{b-contour}}$ as low as $10$. Due to the computational costs of the optimization problem, we choose a relatively low number of bootstrap samples. The simulation studies empirically confirm that percentile and BCa bootstrap sensitivity intervals achieve good coverage nonetheless.

\bibliographystyle{jasa3}
\bibliography{reference}

\end{document}